\documentclass[11pt]{article}
\newcommand{\compileACM}{0}

\usepackage{ifthen}
\usepackage[margin=1in,letterpaper]{geometry}

\ifthenelse{\equal{\compileACM}{0}}{
\usepackage{amsmath,amssymb, amsthm}
}
{ 
\usepackage{etex}
}

\usepackage{enumitem}
\setlist{nolistsep}

\usepackage{xcolor}
\definecolor{darkblue}{rgb}{0,0,0.5}
\definecolor{darkgreen}{rgb}{0,0.5,0}
\usepackage[colorlinks=true,linkcolor=darkblue,urlcolor=darkblue,citecolor=darkgreen]{hyperref}
\usepackage{graphics,graphicx}
\usepackage{subfigure}
\usepackage{mathrsfs} 
\usepackage{bbm} 
\usepackage{fullpage}
\usepackage{complexity}
\usepackage{float}
\usepackage{nicefrac}
\usepackage{xspace}
\usepackage{microtype}

\usepackage{tikz}
\usetikzlibrary{calc,decorations.pathreplacing}

\usepackage{tabularx}
\usepackage{multirow}

\usepackage{appendix}

\ifthenelse{\equal{\compileACM}{1}}{
\usepackage{environ}
\NewEnviron{killcontents}{}

}
{ 
}

\ifthenelse{\equal{\compileACM}{1}}{
\newcommand{\lightparagraph}[1]{{\vspace{0.25em} \noindent \emph{#1.}}}
}
{
\newcommand{\lightparagraph}[1]{\paragraph{#1.}}
}

 \floatstyle{ruled}
\newfloat{protocol}{tph}{lop}
\floatname{protocol}{Protocol}
\newfloat{scheme}{tph}{lop}
\floatname{scheme}{Scheme}
\newfloat{functionality}{htp}{lop}
\floatname{functionality}{Functionality}
\newfloat{simulator}{tph}{lop}
\floatname{simulator}{Simulator}
\newfloat{experiment}{tph}{lop}
\floatname{experiment}{Experiment}

\usepackage{xy}
\xyoption{matrix}
\xyoption{frame}
\xyoption{arrow}
\xyoption{arc}

\usepackage{ifpdf}
\ifpdf
\else
\PackageWarningNoLine{Qcircuit}{Qcircuit is loading in Postscript mode.  The Xy-pic options ps and dvips will be loaded.  If you wish to use other Postscript drivers for Xy-pic, you must modify the code in Qcircuit.tex}
\xyoption{ps}
\xyoption{dvips}
\fi

\entrymodifiers={!C\entrybox}

\newcommand{\qw}[1][-1]{\ar @{-} [0,#1]}
\newcommand{\qwx}[1][-1]{\ar @{-} [#1,0]}
\newcommand{\cw}[1][-1]{\ar @{=} [0,#1]}
\newcommand{\cwx}[1][-1]{\ar @{=} [#1,0]}
\newcommand{\gate}[1]{*+<.6em>{#1} \POS ="i","i"+UR;"i"+UL **\dir{-};"i"+DL **\dir{-};"i"+DR **\dir{-};"i"+UR **\dir{-},"i" \qw}
\newcommand{\meter}{*=<1.8em,1.4em>{\xy ="j","j"-<.778em,.322em>;{"j"+<.778em,-.322em> \ellipse ur,_{}},"j"-<0em,.4em>;p+<.5em,.9em> **\dir{-},"j"+<2.2em,2.2em>*{},"j"-<2.2em,2.2em>*{} \endxy} \POS ="i","i"+UR;"i"+UL **\dir{-};"i"+DL **\dir{-};"i"+DR **\dir{-};"i"+UR **\dir{-},"i" \qw}

\newcommand{\control}{*!<0em,.025em>-=-<.2em>{\bullet}}

\newcommand{\ctrl}[1]{\control \qwx[#1] \qw}

\newcommand{\targ}{*+<.02em,.02em>{\xy ="i","i"-<.39em,0em>;"i"+<.39em,0em> **\dir{-}, "i"-<0em,.39em>;"i"+<0em,.39em> **\dir{-},"i"*\xycircle<.4em>{} \endxy} \qw}

\newcommand{\multigate}[2]{*+<1em,.9em>{\hphantom{#2}} \POS [0,0]="i",[0,0].[#1,0]="e",!C *{#2},"e"+UR;"e"+UL **\dir{-};"e"+DL **\dir{-};"e"+DR **\dir{-};"e"+UR **\dir{-},"i" \qw}

\newcommand{\rstick}[1]{*!L!<-.5em,0em>=<0em>{#1}}
\newcommand{\lstick}[1]{*!R!<.5em,0em>=<0em>{#1}}
\newcommand{\ustick}[1]{*!D!<0em,-.5em>=<0em>{#1}}

\newcommand{\Qcircuit}{\xymatrix @*=<0em>}

\newcommand{\pureghost}[1]{*+<1em,.9em>{\hphantom{#1}}}

\newcommand{\cgate}[1]{*+<.6em>{#1} \POS ="i","i"+UR;"i"+UL **\dir{-};"i"+DL **\dir{-};"i"+DR **\dir{-};"i"+UR **\dir{-},"i" \cw}
\newcommand{\multicgate}[2]{*+<1em,.9em>{\hphantom{#2}} \POS [0,0]="i",[0,0].[#1,0]="e",!C *{#2},"e"+UR;"e"+UL **\dir{-};"e"+DL **\dir{-};"e"+DR **\dir{-};"e"+UR **\dir{-},"i" \cw}

\newcommand{\getsr}{{\:\stackrel{{\scriptscriptstyle\hspace{0.2em}\$}}{\leftarrow}\:}}

\providecommand{\Pr}[1]{\text{Pr}\[#1\]}

\newcommand{\qmeter}{\mathrel{\ooalign{\hss/\hss\cr$\frown$}}}

\newcommand{\cA}{\mathcal{A}}

\newcommand{\cF}{\mathcal{F}}
\newcommand{\cG}{\mathcal{G}}

\newcommand{\cZ}{\mathcal{Z}}
\newcommand{\sX}{\mathsf{X}}

\newcommand{\sZ}{\mathsf{Z}}
\newcommand{\sA}{\mathsf{A}}
\newcommand{\sB}{\mathsf{B}}

\newcommand{\sP}{\mathsf{P}}
\newcommand{\sR}{\mathsf{R}}

\newcommand{\cFunc}[2]{{\cF_{#1}^{\text{\upshape #2}}}}

\newcommand{\etal}{\emph{et~al.}\xspace}

\newcommand{\defeq}{\stackrel{\smash{\textnormal{\tiny def}}}{=}}
\newcommand{\trans}{{\scriptstyle\mathsf{T}}}

\newcommand{\enc}{{\operatorname{enc}}}
\newcommand{\simu}{{\operatorname{sim}}}
\newcommand{\dec}{{\operatorname{dec}}}

\newtheorem{theorem}{Theorem}

\newtheorem{lemma}[theorem]{Lemma}
\newtheorem{corollary}{Corollary}[theorem]
\newtheorem{proposition}[theorem]{Proposition}

\newtheorem{defn}[theorem]{Definition}
\newenvironment{definition}{\begin{defn}}{\qed\end{defn}}

\newtheorem{conjecture}[theorem]{Conjecture}

\newcommand{\pa}[1]{(#1)}
\newcommand{\Pa}[1]{\left(#1\right)}

\newcommand{\set}[1]{\{#1\}}

\newcommand{\bra}[1]{\langle#1|}
\newcommand{\ket}[1]{|#1\rangle}

\DeclareMathOperator{\trace}{Tr}
\newcommand{\ptr}[2]{\trace_{#1}\pa{#2}}
\newcommand{\Ptr}[2]{\trace_{#1}\Pa{#2}}

\newcommand{\tinyspace}{\mspace{1mu}}
\newcommand{\abs}[1]{|\tinyspace#1\tinyspace|}

\newcommand{\norm}[1]{\lVert\tinyspace#1\tinyspace\rVert}

\newcommand{\tnorm}[1]{\norm{#1}_{\trace}}

\newcommand{\identity}{\mathbbm{1}}
\newcommand{\idsup}[1]{\identity_{#1}}

\newcommand{\ol}[1]{{\overline{#1}}}
\def\ot{\otimes}

\newcommand{\cnot}{\textsc{cnot}}

\newcommand{\iin}{{\mathsf{in}}}
\newcommand{\out}{{\mathsf{out}}}
\newcommand{\smm}{{\mathsf{sim}}}

\newcommand{\cont}[1]{c\textrm{-}{#1}}
\newcommand{\tildecont}[1]{c\textrm{-}\tilde{#1}}
\newcommand{\tildeconta}[1]{c\textrm{-}\tilde{#1}_a}

\def\cA{\mathcal{A}}
\def\cB{\mathcal{B}}

\def\cF{\mathcal{F}}
\def\cG{\mathcal{G}}

\def\cZ{\mathcal{Z}}

\def\bQ{\mathbf{Q}}

\def\sA{\mathsf{A}}
\def\sB{\mathsf{B}}
\def\sC{\mathsf{C}}
\def\sD{\mathsf{D}}
\def\sE{\mathsf{E}}

\def\sK{\mathsf{K}}

\def\sM{\mathsf{M}}

\def\sP{\mathsf{P}}

\def\sR{\mathsf{R}}
\def\sS{\mathsf{S}}

\def\sW{\mathsf{W}}
\def\sX{\mathsf{X}}

\def\sZ{\mathsf{Z}}

\title{\bf Quantum one-time programs}
 \author{
  Anne Broadbent$^\dagger$ $\qquad$ Gus Gutoski$^\ddagger$ $\qquad$ Douglas Stebila$^*$\\[3mm]
  {\small\it
  \begin{tabular}{c}
    {\large$^\dagger$}Institute for Quantum Computing and Department of Combinatorics and Optimization \\
    University of Waterloo,
    Waterloo, Ontario, Canada \\
    \href{mailto:albroadb@iqc.ca}{\tt albroadb@iqc.ca} \\ [1mm]
    {\large$^\ddagger$}Institute for Quantum Computing and School of Computer Science \\
    University of Waterloo,
    Waterloo, Ontario, Canada \\
    \href{mailto:gus.gutoski@uwaterloo.ca}{\tt gus.gutoski@uwaterloo.ca} \\ [1mm]
    {\large$^*$}School of Electrical Engineering and Computer Science, Science and Engineering Faculty\\
    Queensland University of Technology, Brisbane, Queensland, Australia \\
    \href{mailto:stebila@qut.edu.au}{\tt stebila@qut.edu.au}
  \end{tabular}
  }
}

\date{November 6, 2012}
\begin{document}
\maketitle

\begin{abstract}
One-time programs are modelled after a black box that allows a single evaluation of a function, and then self-destructs. Because software can, in principle, be copied, general one-time programs  exists only in the hardware token model: it has been shown that any function admits a one-time program as long as we assume access to physical devices called \emph{one-time memories}. Quantum information, with its well-known property of \emph{no-cloning}, would, at first glance, prevent the basic copying attack for classical programs. We show that this intuition is false: one-time programs for both classical and quantum maps, based solely on quantum information, do not exist, even with computational assumptions. We complement this strong impossibility proof  by an equally strong possibility result: assuming the same basic one-time memories as used for classical one-time programs, we show that every quantum map has a quantum one-time program that is secure in the universal composability framework. Our construction relies on a new, simpler quantum authentication scheme and corresponding mechanism for computing on authenticated data.
\end{abstract}

\ifthenelse{\equal{\compileACM}{0}}{
\newpage
\setcounter{tocdepth}{2}
\tableofcontents
\newpage
}{}


\section{Introduction}
\label{sec:intro}

A one-time program for a function $f$, as introduced by Goldwasser, Rothblum and Kalai \cite{GKR08}, is a cryptographic primitive by which a receiver may evaluate $f$ on only one input, chosen by the receiver at run time: no efficient adversary, after evaluating the one-time program on $x$, should be able to learn anything about $f(y)$ for any $y\neq x$ beyond what can be inferred from $f(x)$.
Secure one-time programs could have far reaching implications in software protection, digital rights management, and electronic cash or token schemes.  (For example, coins are a program that can only be run once, and thus cannot be double spent.)

One-time programs cannot be achieved by software alone, as any software can be copied and executed multiple times.
Thus, any hope of achieving any one-time property must necessarily rely on an additional assumptions such as secure hardware or interaction; in particular, computational assumptions alone will not suffice.

\subsection{Classical one-time programs from one-time memories}

Goldwasser~\etal showed how to construct a one-time program for any function~$f$ using a very basic hypothetical hardware device called a \emph{one-time memory (OTM)}.
Inspired by the interactive cryptographic primitive \emph{oblivious transfer}, each OTM stores two secret strings $(s_0, s_1)$.
A receiver requests one of these two strings by specifying a single-bit input $c\in\set{0,1}$.
The OTM reveals $s_c$ and then self-destructs: the other string $s_{\overline{c}}$ is lost forever.

One advantage of using OTMs as a building block is their simplicity: an OTM is an extremely basic device that does not perform any computation.
Using the simplest possible hardware device allows for easier scrutiny against potential hardware flaws such as side-channel attacks.
Moreover, the functionality of an OTM is independent of the program itself, and thus OTMs could be mass-produced for a variety of programs.
The use of tamper-proof hardware in cryptography is an old and recurring theme \cite{Smid81}, and OTMs in particular have lead to a recent revival in ascertaining what cryptographic primitives can be constructed using minimalistic hardware assumptions that could not otherwise be achieved.

\textbf{Non-interactive secure two-party computation.}
Goyal, Ishai, Sahai, Venkatesan, and Wadia \cite{GISVW10} improved on the work of Goldwasser~\etal~\cite{GKR08} in several ways.
First, they consider a more general primitive, which they call \emph{non-interactive secure two-party computation}, in which two parties wish to evaluate a publicly known function $f(x,y)$.
One party---the \emph{sender}---is given the input string $x$.
The sender uses $x$ to prepare a ``program'' $p(x)$ and sends this program to the \emph{receiver}.
The receiver wishes to use the program $p(x)$ in order to evaluate $f(x,y)$ for any input string $y$ of her choice.
Like one-time programs, after evaluating $f(x,y)$, no adversary should be able to learn anything about $f(x,y')$ for any $y'\neq y$ beyond what can be inferred from $f(x,y)$.
The one-time programs of Goldwasser~\etal
are recovered as a special case of this primitive by viewing the input~$x$ as a description of a function $g_x(y)$ and the publicly known function $f$ as a ``universal computer'' that produces $f(x,y)=g_x(y)$.
Non-interactive secure two-party computation is impossible in the plain model for the same reason that one-time programs are impossible in the plain model: software can always be copied.

Second, the one-time programs of Goldwasser~\etal are secure against a malicious receiver; the issue of a malicious sender does not arise in their setting.
By contrast, in the more general setting of non-interactive secure two-party computation one could also consider malicious \emph{senders}.
The protocol of Goyal \etal is secure against both a malicious receiver \emph{and} a malicious sender.

Third, Goldwasser \etal use OTMs for large strings, whereas Goyal \etal require OTMs for only \emph{single bits}.
Finally, Goldwasser \etal establish security against computationally bounded adversaries, whereas Goyal \etal establish statistical universally composable security. \looseness=-1

\textbf{Terminology:}
For brevity, we use the term \emph{one-time program (OTP)} synonymously with ``non-interactive secure two-party computation''.

\ifthenelse{\equal{\compileACM}{0}}{
\subsection{Impossibility of quantum one-time programs in the plain model}
}
{ 
\subsection{Impossibility of quantum one-time \\ programs in the plain model}
}

In contrast to ordinary classical information, quantum information cannot in general be copied: measurement is an irreversible, destructive process~\cite{wootters1982single}.
The no-cloning property of quantum information is credited for such classically impossible feats as quantum money \cite{AC12, MS10, wiesner1983conjugate}, quantum key distribution \cite{bennett1984quantum}, and quantum copy-protection \cite{Aaronson09}.
It is thus natural to ask if one-time programs can be added to this list of quantum cryptographic primitives: \emph{does quantum information allow for one-time programs without hardware assumptions?}
(When there are no hardware assumptions, we refer to this as the \emph{plain quantum model}.)

We observe in Section~\ref{sec:necessity} that the answer to this question is a strong \emph{no}: although quantum information cannot be copied, a quantum ``program state'' for $f$ can always be re-constructed by a reversible adversary after each use so as to obtain the evaluation of $f$ on multilple distinct inputs.
In particular, computational assumptions do not help to achieve quantum one-time programs.

\textbf{One-time programs for quantum channels.}
By analogy to classical functions acting on bits, one could also consider a one-time program for a quantum channel $\Phi:(\sA,\sB)\to\sC$ acting on multi-qubit registers $\sA$ (the sender's input), $\sB$ (the receiver's input), and $\sC$ (the receiver's output).
The security goal is similar in spirit to that for classical functions:
for each joint state $\rho$ of the input registers $(\sA,\sB)$ no adversary should be able to learn anything about $\Phi(\rho')$ for any state $\rho' \ne \rho$ beyond what can be inferred from $\Phi(\rho)$.

Here, again, our previous observation on the impossibility of quantum one-time programs in the plain model holds: if a quantum program allows the adversary to evaluate $\Phi(\rho)$ then that same program can be recovered by a reversible adversary in order to subsequently evaluate $\Phi(\rho')$ for any~$\rho'$ that can be obtained from~$\rho$ by a local operation on~$\sB$.

\ifthenelse{\equal{\compileACM}{0}}{
\subsection{Main result: quantum one-time programs from one-time memories}
}
{ 
\subsection{Main result: quantum one-time \\ programs from one-time memories}
}

Given that one-time programs do not exist for arbitrary quantum channels in the plain quantum model, and that one-time programs do exist for arbitrary classical functions in the OTM model, a natural question arises:  \emph{what additional assumptions are required to achieve one-time programs for quantum channels?}

Our main result (Theorem~\ref{thm:main}) is that any channel $\Phi$ can be compiled into a \emph{quantum one-time program (QOTP)}, which is a combination of quantum states and OTMs that allows a receiver to evaluate $\Phi$ exactly once.
In particular, (and perhaps surprisingly) single-bit \emph{classical} one-time memory devices suffice to establish \emph{quantum} one-time programs for arbitrary \emph{quantum} channels.
An informal version of Theorem~\ref{thm:main} is as follows:

\vspace{0.5em}
\noindent\textbf{Main theorem (informal).}
\emph{For each channel $\Phi$ specified by a quantum circuit, there is a non-interactive two-party protocol for the secure evaluation of $\Phi$, assuming classical one-time memory devices and an honest sender.
Moreover, this protocol is universally composable (UC-secure).}\vspace{0.5em}

We provide a fully rigorous proof of statistical universally composable (UC) security of our QOTPs.
The question of QOTPs that are also secure against a malicious sender is left for future work.
For simplicity, we restrict our attention to the case of  \emph{non-reactive} one-time quantum programs.
The more general scenario of being able to query a program a bounded number of times (including the case of an $n$-use program) may be implemented using standard techniques as is done in the classical case (see Section~\ref{sec:COTP}).
Most of the components of our QOTP for $\Phi$ are independent of the sender's input register $\sA$ and so can be compiled (and mass-produced) by the sender before he receives his input.
As a corollary of our main result we obtain the UC security of the protocol for \emph{delegated quantum computations} of Aharonov, Ben-Or and Eban~\cite{AharonovBE10}; see Section~\ref{sec:UC-sec-of-ABE10}.

\textbf{Summary of techniques.}
Our protocol employs a method for \emph{quantum computing on authenticated data (QCAD)}, which allows for the application of quantum gates to authenticated quantum data without knowing the authentication key.
We propose a new authentication scheme, called the \emph{trap scheme}, and show that it allows for QCAD (Section \ref{sec:techniques-Q-auth}).
Prior to our work, the only authentication scheme known to admit QCAD was the \emph{signed polynomial scheme} of Ben-Or, Cr\'epeau, Gottesman, Hassidim, and Smith \cite{Ben-OrC+06} (see also~\cite{AharonovBE10}).
We compare our trap scheme to the signed polynomial scheme in Section~\ref{sec:techniques-computing-on-auth}.
Recently, and independently of our work, it was shown by Dupuis, Nielsen, and Salvail~\cite{DNS12} that the \emph{Clifford authentication scheme} also admits QCAD.

In  methods for QCAD, universal quantum computation can only be performed if the receiver (who holds the authenticated data) is allowed to exchange classical messages with the sender (who knows the authentication key).
To keep our protocol non-interactive, all the classical interaction is encapsulated by a \emph{bounded, reactive classical one-time program (BR-OTP)} prepared by the sender.
(The existence of secure reactive one-time programs follows from the work of \cite{GISVW10} as described in Section~\ref{sec:COTP}.)
This program for the BR-OTP depends upon the authentication key chosen for the sender's input register, but \emph{not} on the contents of that register.
Thus, by selecting this key in advance, the BR-OTP can be prepared (or mass-produced) before the sender gets his input register.

In order to implement QCAD, the receiver's input must be authenticated prior to computation.
This is accomplished non-interactively by having the sender prepare a pair of registers in a special ``teleport-through-encode'' state, which is a maximally entangled state with an authentication operation applied to one of the two registers.
The authentication key is determined by the (classical) result of the Bell measurement used for teleportation.
The sender allows the receiver to non-interactively de-authenticate the output at the end of the computation by preparing another pair of registers in a ``teleport-through-decode'' state.
In order to successfully de-authenticate, the receiver's messages to the BR-OTP must be consistent with the secret authentication key held by the BR-OTP.
Otherwise, the BR-OTP simply declines to reveal the final decryption key for the receiver's output.

\subsection{Formalizing impossibility}

In preparing a formal proof that one-time programs do not exist in the plain model, one immediately encounters a pathological class of functions that do, in fact, admit classical one-time programs
in the plain model.
For example, consider the function $f(a,b)=a+b$.
A potential ``one-time'' program for $f$ has the sender simply reveal $a$ to the receiver.
Curiously, this is indeed a ``one-time'' program, because this behaviour can be simulated with one-shot access to the ideal functionality $f(a, \cdot)$: the query $f(a,0)$ reveals $a$, which is exactly the one-time program prepared by the sender.
Put another way: $f$ is a function for which there exists a one-time program in the plain model, but only because the one-time program reveals enough information that even a simulator with one-shot access to $f(a,\cdot)$ can gain all information required to compute the function.
This phenomenon is somewhat akin to trivially obfuscatable functions \cite{BGIRSVY01}.

An interesting question thus arises: \emph{what is the class of classical functions or quantum channels that admit one-time programs in the plain model?}
We provide a complete characterization of this class for classical functions.
In particular, we define a class of \emph{unlockable} functions consisting of those functions $f$ for which there exists at least one \emph{key input} $b_0$ such that for all $a$ the value $f(a,b)$ can be recovered from $f(a,b_0)$ for any desired $b$.
We prove the following.

\vspace{0.5em}
\noindent\textbf{Impossibility theorem for classical functions (informal).}
\emph{If $f$ is unlockable then $f$ admits a trivial classical one-time program in the plain model.
Conversely, if $f$ is not unlockable then $f$ does not admit a quantum one-time program in the plain quantum model.}\vspace{0.5em}

The situation for quantum channels is very interesting.
We propose two definitions for the unlockability of a channel, which we call \emph{weakly unlockable} and \emph{strongly unlockable}
(Definition \ref{def:unlockable}).
We prove the following.

\vspace{0.5em}
\noindent\textbf{Impossibility theorem for quantum channels (informal).}
\emph{If $\Phi$ is strongly unlockable then $\Phi$ admits a trivial quantum one-time program in the plain quantum model.
Conversely, if $\Phi$ is not weakly unlockable then $\Phi$ does not admit a quantum one-time program in the plain quantum model.}\vspace{0.5em}

By definition, every strongly unlockable channel is also weakly unlockable.
We conjecture that the converse also holds (Conjecture \ref{conj:bug}.)
In lieu of a full proof, we provide a high-level outline of what such a proof might look like in Section \ref{sec:conjecture}.
It appears, in asking a question about unlockable channels, that we have stumbled upon a deep and interesting question relating to the invertible subspaces of an arbitrary channel, akin to the
so-called ``decoherence-free'' subspaces studied in the literature on quantum error correction.

\subsection{Related work}
\label{sec:intro:open}

\lightparagraph{Copy-protection}
In software copy-protection~\cite{Aaronson09}, a program can be evaluated (a possibly unlimited number of times), but it should be impossible for the program to be ``split'' or ``copied'' into parts allowing separate executions. As with OTPs, copy-protection cannot be achieved by software means only. OTPs provide a hardware solution by enforcing that the program  be run only once. However, the more interesting question is if quantum information alone (with computational assumptions) can provide a solution. Aaronson~\cite{Aaronson09} has proposed such schemes based on plausible, but non-standard, cryptographic assumptions.  It is an open problem if quantum copy-protection could be based on standard assumptions. In contrast, the security of quantum OTPs is based on simple OTMs; it could be beneficial to study quantum copy-protection in light of our result.

\lightparagraph{Quantum money}
Our construction  establishes quantum authentication codes (see Section~\ref{sec:techniques-Q-auth}) that seem to provide a concrete and efficient realization of the ``hidden subspaces'' used for public-key quantum money scheme of Aaronson and Christiano~\cite{AC12}.
Our QOTPs can also be used to implement non-interactive verification for quantum coin schemes~\cite{MS10}.

\lightparagraph{Program obfuscation}
A related but different task is program obfuscation, in which the receiver should not be able to efficiently ``learn'' anything from the description of the program that he could not also efficiently learn from the input-output behaviour of the program.
In the case of classical information, it is known that  secure program obfuscation is not possible in the plain model~\cite{BGIRSVY01}.
As with copy-protection,  OTPs provide a hardware solution by enforcing that the obfuscated program can be run only a limited number of times.
Again, the more interesting question is if quantum information alone (with computational assumptions) can provide a solution; the  impossibility proof for obfuscation breaks down in the quantum case due to the no-cloning theorem.
It is an open problem whether there is a way to substitute the assumption of secure hardware in our main possibility result with a computational assumption in order to realize quantum program obfuscation.

\section{Notation and tools}
\label{sec:notation}

A \emph{quantum register} is a specific quantum system (composed, say, out of qubits).
Quantum registers are typically denoted with sans serifs such as $\sA$, $\sB$, \emph{etc.}
A \emph{quantum channel} $\Phi:\sA\to\sB$ refers to any physically-realizable mapping on quantum registers.
We use the terms \emph{quantum map} and \emph{quantum channel} interchangeably.

\subsection{Universal composability}
\label{sec:definitions}
The universal composability (UC) framework provides an extremely high standard for establishing a strong and  rigorous  notion of security. The basic idea is to postulate an ideal world, where the protocol parties interact with an \emph{ideal functionality}, which is secure by definition.
Then, we consider the real world, where the protocol parties execute the actual protocol.
\emph{UC-security} holds if, for every real-world adversary, there exists a \emph{simulator} in the ideal world (taking the role of the real-world adversary) such that no \emph{environment} can distinguish the real and ideal worlds.
The environment is powerful: it provides the parties' inputs, receives outputs and \emph{interacts} with the adversary at arbitrary points in the protocol.
  \emph{Perfect}, \emph{statistical}, and \emph{computational} notions of UC security exist depending on the two worlds being  equal  or statistically or computationally indistinguishable.

The UC  framework was developed in the classical world by Canetti~\cite{Can01} and
independently (under the name of reactive simulatability) by
Pfitzmann and Waidner~\cite{PW01}. The model was extended to the
quantum world by Ben-Or and Mayers~\cite{BOM04} and independently by
Unruh~\cite{U04}. Here, we follow the simplified UC framework as presented by Unruh~\cite{U10}, to which we refer for the description of the model, definitions, and theorems.

Via the composition theorem, this framework allows  to rigorously prove results that are maximally useful for future work: our results can  easily be embedded within a larger construction and at the same time we can use prior constructions without having to re-visit their security proofs.
Another key result is Unruh's \emph{quantum lifting theorem}~\cite{U10} establishing that, in the statistical case, classical-UC-secure protocols are also quantum-UC-secure.

Our main possibility result heavily relies on classical one-time programs, for which a (classical) UC-secure instantiation exists assuming one-time memories~\cite{GISVW10}.
In particular, we require  bounded reactive OTPs, which we construct by extending the results of Goyal~\etal~\cite{GISVW10} (see Section~\ref{sec:COTP}).

Some notable variations of Unruh's terminology and definitions follow.

Out of the two protocol parties  (the \emph{sender} and the
\emph{receiver}), we consider security only in the case of the receiver being a corruption party. This is the meaning of, \emph{e.g.}~``\emph{$\pi$ statistically
quantum-UC-emulates $\rho$ in the case of a corrupted receiver''}.

In our scenario involving a corrupted receiver,  we make use of a  property related to the transitivity of UC-security~\cite{U10} (see Lemma~\ref{lem:R-transitive}); the proof is very similar to the original proof of transitivity.

\begin{lemma}
\label{lem:R-transitive}
Let $\pi$, $\sigma$ and $\rho$ be protocols with parties $\{\verb"sender", \verb"receiver"\}$. Suppose that $\pi$
statisically quantum-UC-emulates~$\rho$ and that $\rho$ statistically quantum-UC-emulates~$\sigma$ in the case of a corrupted receiver. Then~$\pi$ statistically quantum-UC-emulates~$\sigma$ in the case of a corrupted receiver.
\end{lemma}

Combining Lemma~\ref{lem:R-transitive} with the quantum universal composition theorem~\cite{U10}  and the quantum lifting theorem~\cite{U10} we get
Corollary~\ref{cor:composition} below, which is essential in the
proof of our main possibility theorem (Theorem~\ref{thm:main}).

\begin{corollary}\label{cor:composition}
Let $\pi$, and $\rho$ be protocols with parties $\{\verb"sender", \verb"receiver"\}$.
Suppose $\pi$ statistically classical-UC-emulates $\cF$.
Suppose that $\rho^\cF$ statistically quantum-UC-emulates $\cG$ in the case of a corrupted receiver.
Then $\rho^\pi$ statistically quantum-UC-emulates $\cG$ in the case of a corrupted receiver.
\end{corollary}

\subsection{Ideal functionalities}
\label{sec:definitions:functionalities}

We now describe the relevant ideal functionalities. All
functionalities involve two parties, the
\emph{sender} and the \emph{receiver}. An ideal functionality may
exist in multiple instances and involves various parties. Formally,
instances are denoted by \emph{session identifiers} and each
instance involves labelled parties. For the sake of simplicity, we
have omitted these identifiers as they should be implicitly clear
from the context.

The ideal functionality $\cFunc{}{OTM}$ for a one-time
memory (OTM) is a two-step process modelled after oblivious
transfer. We sometimes refer to this functionality
$\cFunc{}{OTM}$ as an \emph{OTM token}.

\begin{functionality} \caption{Ideal Functionality
$\cFunc{}{OTM}$.} \label{ideal-funct:OTM}
\begin{enumerate}
\item \textbf{Create:} Upon input $(s_0, s_1)$ from the sender with $s_0, s_1 \in
\{0,1\}$, send \verb"create" to the receiver and store~$(s_0,
s_1)$.

\item \textbf{Execute:} Upon input $c \in \{0,1\}$ from the
receiver, send $s:=s_c$ to the receiver. Delete any trace of this
instance.
\end{enumerate}
\end{functionality}

Next we describe the ideal functionalities $\cFunc{f}{OTP}$ and $\cFunc{\Phi}{OTP}$ of a one-time program for a classical function~$f$ and a quantum channel $\Phi$, respectively.
Note that the map that is computed ($f$ or $\Phi$) is a public parameter of the functionality and it takes an input from the sender and an input from the receiver. We thus view these ideal functionalities as having the property of hiding the sender's \emph{input} only.
If the intention is to to hide the map~$m$ itself---as in the intuitive notion of one-time programs---then we can consider a universal map~$U$ that takes as part of its sender's input a program register representing~$m$ (see~\cite{BFGH10, NC97, SR07}).

\begin{functionality}
\caption{Ideal functionality $\cFunc{f}{OTP}$ for a classical function $f: \{0,1\}^{n+m} \rightarrow  \{0,1\}^k$.
\label{ideal-funct:COTP}}
\begin{enumerate}
\item \textbf{Create:} Upon input  $a \in \{0,1\}^{n}$ from the sender,
send \verb"create" to the receiver and store $a$.
\item \textbf{Execute:} Upon input $b \in \{0,1\}^m$ from the
receiver, send $f(a,b)$ to the receiver. Delete any trace of this
instance.
\end{enumerate}
\end{functionality}

\begin{functionality}
\caption{Ideal functionality $\cFunc{\Phi}{OTP}$ for a quantum channel $\Phi : (\sA,\sB) \rightarrow \sC$.
\label{ideal-funct:QOTP}}
\begin{enumerate}
\item \textbf{Create:} Upon input register $\sA$ from the sender, send \verb"create" to the
receiver and store the contents of register $\sA$.
\item \textbf{Execute:} Upon input register $\sB$ from the
receiver, evaluate $\Phi$ on registers $\sA,\sB$ and send the contents of the output register $\sC$ to
the receiver. Delete any trace of this instance.
\end{enumerate}
\end{functionality}

It is clear from the description of these ideal functionalities that they may be called only a single time.
However, we may sometimes emphasize this in expressions such as ``one-shot access to an ideal functionality computing~$f$''.

Functionalities \ref{ideal-funct:OTM}--\ref{ideal-funct:QOTP} are \emph{sender-oblivious} since they take an input from the sender and an input from the receiver and deliver the result of the functionality to the receiver (but not the sender).
Moreover, they are \emph{non-reactive} since they interact with the sender and the receiver in a single round.
\emph{Reactive} functionalities are more general, potentially having several rounds of inputs and outputs and maintaining state between rounds.
In  Section~\ref{sec:COTP} we consider an ideal functionality for \emph{bounded reactive classical one-time programs}; the ideal functionality for bounded-reactive OTPs is specified in%
\ifthenelse{\equal{\compileACM}{1}}{
the full version.
}{
Appendix~\ref{sec:Appendix-BRCOTPs}.
}

\subsection{Classical one-time programs} \label{sec:COTP}

Our construction relies heavily on classical OTPs, the construction of which is given by Goyal~\etal~\cite{GISVW10}:

\begin{theorem}
\label{thm:COTP} Let $f$ be a
non-reactive, sender-oblivious, polynomial-time computable classical
two-party functionality. Then there exists an efficient,
 non-interactive protocol which statistically
classical-UC-emulates $\cFunc{f}{OTP}$ in the
$\cFunc{}{OTM}$-hybrid model.
\end{theorem}

\ifthenelse{\equal{\compileACM}{0}}{
In Appendix~\ref{sec:Appendix-BRCOTPs}, we
}{
In the full version, we
}%
use straightforward techniques to extend this result to sender-oblivious, polynomial-time computable, \emph{bounded reactive} classical two-party functionalities.
The main result on reactive OTPs, as used in our construction in Section~\ref{sec:construction}, is:

\begin{corollary}\label{cor:reactive-COTP}
There exists a non-interactive protocol $\sigma$ that
statistically classical-UC-emulates $\cFunc{g_1,\ldots g_\ell}{BR-OTP}$
in the $\cFunc{}{OTM}$-hybrid model.
\end{corollary}

\section{Impossibility of non-trivial OTP\lowercase{s} in the plain model}
\label{sec:necessity}

We now consider whether classical functions or quantum channels admit one-time programs in the plain quantum model.
We will see that it is precisely maps that are \emph{unlockable}---meaning there is a \emph{key}\footnote{Note we use ``key'' not in the cryptographic sense of a secret key, but in the metaphorical sense of something that unlocks a lock.} input to the map that unlocks enough information to fully simulate the map---that have one-time programs in the plain model, and that these one-time programs are in a sense trivial.
For quantum channels, we will have two versions of unlockable, the difference being whether the key that unlocks the channel is a state (\emph{strongly unlockable}) or a channel that transforms a given input (\emph{weakly unlockable}).

Our \emph{possibility} result shows that every strongly unlockable channel admits a trivial one-time program in the plain quantum model, and in fact that this protocol is UC-secure.
Our \emph{impossibility} result  shows that every channel that is not weakly unlockable does not admit a one-time program in the plain quantum model; our impossibility result holds even if we relax to an approximate case or allow computational assumptions.%
\footnote{Although our impossibility result is stated in the UC framework, the impossibility is not an artifact of the high level of security required of UC, but seems inherent in the notion of OTPs, and the impossibility argument applied for any relaxation we attempted.}
As we will see, it is easy to establish that the weak and strong unlockable notions are equivalent for classical functions, but whether they are equivalent for quantum channels is an open question, which we note appears to be an interesting and deep question related to invertible subspaces of a channel.

\subsection{One-time programs in the plain quantum model}
\label{sec:otp-plain}
Here, we formalize some concepts relating to one-time programs in the plain quantum model.

A protocol for a quantum one-time program in the plain model consists of a single quantum message,  the \emph{program register} $\sP$, from the sender to receiver.
Thus, the actions of the honest sender in such a protocol are completely characterized by an \emph{encoding channel} $\enc:\sA\to\sP$ that the sender applies to her portion of the joint input so as to prepare a program register for the receiver.
In particular, for any joint state~$\rho$ of the input registers $(\sA,\sB)$ the joint state of the registers $(\sP,\sB)$ in the receiver's possession after the program register has been received is given by
$(\enc\ot\idsup{\sB})(\rho)$.
From this state we would like the receiver to be able to recover $\Phi(\rho)$.
That is, there should exist a \emph{decoding channel}
$\dec:(\sP,\sB)\to\sC$
for the honest receiver such that the channels
$\dec\circ\enc$ and~$\Phi$ are indistinguishable.
A comparison of the ideal and real models for one-time programs is given in Figure~\ref{fig:qotp-def}.

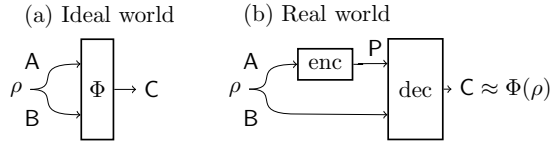
\begin{figure}[hbt]
\begin{center}
\scalebox{0.8}{
\begin{tikzpicture}[yscale=0.7]
\node [right] at (0,1.75) {(a) Ideal world};
\node (rho) at (0,0) {$\rho$};
\node [rectangle,thick,draw=black,inner ysep=20pt] (Phi) at (1.35,0) {$\Phi$};
\draw[->] (rho) .. controls +(0.8,0) and +(-0.8,0) .. node[above left] {$\sA$} (1,0.6) -- ($0.5*(Phi.west)+0.5*(Phi.north west)$);
\draw[->] (rho) .. controls +(0.8,0) and +(-0.8,0) .. node[below left] {$\sB$} (1,-0.6) -- ($0.5*(Phi.west)+0.5*(Phi.south west)$);
\node (Z) at (2.25,0) {$\sC$};
\draw[->] (Phi) -- (Z);
\end{tikzpicture}
~~~
\begin{tikzpicture}[yscale=0.7]
\node [right] at (0,1.75) {(b) Real world};
\node (rho) at (0,0) {$\rho$};
\node [rectangle,thick,draw=black,inner sep=5pt] (enc) at (1.5,0.59) {enc};
\draw[->] (rho) .. controls +(0.8,0) and +(-0.8,0) .. node[above left] {$\sA$} (1,0.59) -- (enc);
\node [rectangle,thick,draw=black,inner xsep=5pt,inner ysep=20pt] (dec) at (3,0) {dec};
\draw[->] (rho) .. controls +(0.8,0) and +(-0.8,0) .. node[below left] {$\sB$} (1,-0.59) -- ($0.5*(dec.west)+0.5*(dec.south west)$);
\draw[->] (enc) .. controls +(1,0) and +(-1,0) .. node[above right]{$\sP$} ($0.5*(dec.west)+0.5*(dec.north west)$);
\node (Z) at (4.5,0) {$\sC \approx \Phi(\rho)$};
\draw[->] (dec) -- (Z);
\end{tikzpicture}
}
\end{center}
\vspace{-1.8em}
\caption{
(a) In the ideal world, the receiver obtains the output of the ideal functionality for $\Phi$ on arbitrary input registers $(\sA,\sB)$.
(b) In the  real world, encoding and decoding maps implement the functionality, namely $\dec \circ \enc \approx \Phi$.
}
\label{fig:qotp-def}
\end{figure}

By the completeness of the dummy-adversary~\cite{U10}, it is sufficient, in order to establish UC-security, to consider only the case of the dummy-adversary who forwards the program register, $\sP$, to the environment.
Thus, UC-security is established by exhibiting a simulator that can re-create a state that is indistinguishable from the joint state \( (\enc\ot\idsup{\sB})(\rho) \) of registers $(\sP,\sB)$, using only the ideal functionality; recall indistinguishability is from the perspective of the environment, and could be perfect, statistical, or computational as appropriate.
The corresponding channels are depicted in Figure \ref{fig:qotp-security}. Here, the \emph{simulator} $(\simu_1,\simu_2)$ consists of channels $\simu_1 : \sB \to (\sB',\sM)$ and $\simu_2: (\sC,\sM) \to (\sP,\sB)$, where $\sM$ is a private memory register for the simulator, security holds if the  channels
$\simu_2\circ\Phi\circ \simu_1$ and $\enc\ot\idsup{\sB}$ are indistinguishable.

\begin{figure}[hbt]
\begin{center}
\scalebox{0.8}{
\begin{tikzpicture}[xscale=0.9,yscale=0.7]
\node [right] at (0,2) {(a) Real world};
\node (rho) at (0,0) {$\rho$};
\node [rectangle,thick,draw=black,inner xsep=3pt,inner ysep=5pt] (enc) at (1.5,1) {enc};
\draw[->] (rho) .. controls +(0.8,0) and +(-0.8,0) .. node[above left] {$\sA$} (1,1) -- (enc);
\node (Y) at (2.75,-1) {$\sB$};
\draw[->] (rho) .. controls +(0.8,0) and +(-0.8,0) .. node[below left] {$\sB$} (1,-1) -- (Y);
\node (P) at (2.75,1) {$\sP$};
\draw[->] (enc) -- (P);
\node at (0,-1.8) {};
\end{tikzpicture}
~
\begin{tikzpicture}[yscale=0.7]
\node [right] at (0,2) {(b) Simulator};
\node (rho) at (0,0) {$\rho$};
\node [rectangle,thick,draw=black,inner xsep=5pt,inner ysep=10pt] (sim1) at (1.6,-1) {$\simu_{1}$};
\draw[->] (rho) .. controls +(0.8,0) and +(-0.8,0) .. node[below left] {$\sB$} (1,-1) -- (sim1);
\node [rectangle,thick,draw=black,inner ysep=10pt] (Phi) at (3.25,1) {$\Phi$};
\draw[->] (rho) .. controls +(0.8,0) and +(-0.8,0) .. node[above left] {$\sA$} (1,1.38) -- ($0.5*(Phi.west)+0.5*(Phi.north west)$);
\draw[->] ($0.5*(sim1.east)+0.5*(sim1.north east)$) .. controls +(0.8,0) and +(-0.8,0) .. node[above left] {$\sB'$} ($0.5*(Phi.west)+0.5*(Phi.south west)$);
\node [rectangle,thick,draw=black,inner xsep=5pt,inner ysep=10pt] (sim2) at (4.9,-1) {$\simu_{2}$};
\draw[->] ($0.5*(sim1.east)+0.5*(sim1.south east)$) -- node[below] {$\sM$} ($0.5*(sim2.west)+0.5*(sim2.south west)$);
\draw[->] (Phi) .. controls +(0.8,0) and +(-0.8,0) .. node[above right] {$\sC$} ($0.5*(sim2.west)+0.5*(sim2.north west)$);
\node (P) at (6,-0.62) {$\sP$};
\draw[->] ($0.5*(sim2.east)+0.5*(sim2.north east)$) -- (P);
\node (Yout) at (6,-1.38) {$\sB$};
\draw[->] ($0.5*(sim2.east)+0.5*(sim2.south east)$) -- (Yout);
\end{tikzpicture}
}
\end{center}
\vspace{-1.8em}
\caption{
(a) The sender prepares the program register $\sP$ by applying $\enc$ to $\sA$.
The sender cannot touch $\sB$.
(b) A simulator $(\simu_{1}, \simu_{2}$) should be able to re-create an indistinguishable state of $(\sP,\sB)$ using only the ideal functionality $\Phi$.
}
\label{fig:qotp-security}
\end{figure}
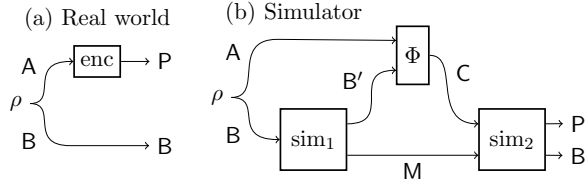

\subsection{Trivial one-time programs for unlockable channels}

Intuitively, a one-time program for a channel means that no receiver can learn more than he could given one-shot access to an ideal functionality.
However, there are certain channels where one-shot access to the ideal functionality is enough to fully simulate the map for a fixed choice of the sender's input.
Such channels---which we call \emph{unlockable}---effectively have trivial one-time programs, as we will see in this section.

\begin{definition}
\label{def:unlockable}
A channel $\Phi:(\sA,\sB)\to\sC$ is \emph{strongly unlockable} if there exists a register $\sK$, a \emph{key state} $\xi_0$ of $(\sB,\sK)$ and a \emph{recovery algorithm} (\emph{i.e.}, channel)
\( \cA:(\sC,\sK,\sB)\to \sC \) with the property that
$\cA\circ\Phi_0 \approx \Phi$,
where the channel $\Phi_0$ is specified by
\[ \Phi_0 : \sA\to(\sC,\sK) : A\mapsto(\Phi\ot\idsup{\sK})(A\ot\xi_0)\,. \]

A channel $\Phi:(\sA,\sB)\to\sC$ is \emph{weakly unlockable} if there exists a register $\sK$ and a \emph{key channel} $\Xi_0:\sB\to(\sB,\sK)$ such that the channel $\Phi\circ\Xi_0$ has the following property:
for every choice of registers~$\sE$ and channels $\Psi:\sB\to(\sB,\sE)$ for the receiver there exists a \emph{recovery algorithm} (\emph{i.e.}, channel)
$\cA_\Psi:(\sC,\sK)\to (\sC,\sE)$
such that:
\[ \cA_\Psi\circ\Phi\circ\Xi_0 \approx \Phi\circ\Psi\,. \]

Here, $\approx$ can denote perfect, statistical, or computational indistinguishability; in all cases, channels $\Phi_0$, $\cA$, $\Xi_{0}$, and $\cA_{\Psi}$ must be efficient.
\end{definition}

See Figure~\ref{fig:qotp-unlockable} for diagrams representing strongly and weakly unlockable channels.

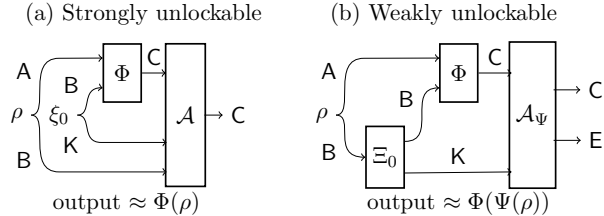
\begin{figure}[hbt]
\begin{center}
\scalebox{0.8}{
\begin{tikzpicture}[xscale=0.7,yscale=0.7]
\node [right] at (0,2.4) {(a) Strongly unlockable};
\node (rho) at (0,0) {$\rho$};
\node (xi0) at (1,0) {$\xi_{0}$};
\node [rectangle,thick,draw=black,inner xsep=5pt,inner ysep=10pt] (Phi) at (2.5,1) {$\Phi$};
\draw[->] (rho) .. controls +(0.8,0) and +(-0.8,0) .. node[above left] {$\sA$} (1,1.35) -- ($0.5*(Phi.west)+0.5*(Phi.north west)$);
\draw[->] (xi0) .. controls +(0.8,0) and +(-0.5,0) .. node[above left] {$\sB$} (2,0.65) -- ($0.5*(Phi.west)+0.5*(Phi.south west)$);
\node [rectangle,thick,draw=black,inner xsep=5pt,inner ysep=30pt] (X) at (4,0) {$\cA$};
\draw[->] (Phi) -- node[above] {$\sC$} (3.53,1);
\draw[->] (xi0) .. controls +(0.8,0) and +(-0.5,0) .. node[below left] {$\sK$} (2,-0.65) -- (3.55,-0.65);
\draw[->] (rho) .. controls +(0.8,0) and +(-0.8,0) .. node[below left] {$\sB$} (1,-1.35) -- (3.55,-1.35);
\node (Z) at (5.25,0) {$\sC$};
\draw[->] (4.45,0) -- (Z);
\node at (2.6,-2.1) {output $\approx \Phi(\rho)$};
\end{tikzpicture}
~~~
\begin{tikzpicture}[xscale=0.7,yscale=0.7]
\node [right] at (0,2.4) {(b) Weakly unlockable};
\node (rho) at (0,0) {$\rho$};
\node[rectangle,thick,draw=black,inner xsep=3pt,inner ysep=10pt] (Xi0) at (1.5,-1) {$\Xi_{0}$};
\draw[->] (rho) .. controls +(0.8,0) and +(-0.8,0) .. node[below left] {$\sB$} (1,-1) -- (Xi0);
\node[rectangle,thick,draw=black,inner xsep=5pt,inner ysep=10pt] (Phi) at (3.25,1) {$\Phi$};
\draw[->] (rho) .. controls +(0.8,0) and +(-0.8,0) .. node[above left] {$\sA$} (1,1.35) -- ($0.5*(Phi.west)+0.5*(Phi.north west)$);
\draw[->] ($0.5*(Xi0.east)+0.5*(Xi0.north east)$) .. controls +(0.8,0) and +(-0.8,0) .. node[above left] {$\sB$} ($0.5*(Phi.west)+0.5*(Phi.south west)$);
\node[rectangle,thick,draw=black,inner xsep=3pt,inner ysep=30pt] (APsi) at (5,0) {$\cA_{\Psi}$};
\draw[->] (Phi) -- node[above] {$\sC$} (4.5,1);
\draw[->] ($0.5*(Xi0.east)+0.5*(Xi0.south east)$) -- node[above] {$\sK$} (4.5,-1.35);
\node (Z) at (6.5,0.6) {$\sC$};
\draw[->] (5.5,0.6) -- (Z);
\node (E) at (6.5,-0.6) {$\sE$};
\draw[->] (5.5,-0.6) -- (E);
\node at (3.25,-2.1) {output $\approx \Phi(\Psi(\rho))$};
\end{tikzpicture}
}
\end{center}
\vspace{-1.8em}
\caption{
(a) For a strongly unlockable channel $\Phi$, there exists a key state $\xi_{0}$ and a recovery algorithm $\cA$ that allows computation of $\Phi(\rho)$ for any $\rho$.
(b) For a weakly unlockable channel $\Phi$, there exists a key channel $\Xi_{0}$ such that for any channel $\Psi$ there exists a recovery algorithm $\cA_{\Psi}$ that allows computation of $\Phi(\Psi(\rho))$ for any $\rho$.
}
\label{fig:qotp-unlockable}
\end{figure}

It is easy to see that every strongly unlockable channel is also weakly unlockable: if~$\xi_{0}$ is the key state for~$\Phi$, then the key channel~$\Xi_{0}$ generates $\xi_{0}$, sends the~$\sB$ register of~$\xi_{0}$ to ideal functionality~$\Phi$ and the~$\sK$ register of~$\xi_{0}$ and the~$\sB$ register of~$\rho$ to $\cA_{\Psi} = \cA \circ \Psi$.

When the channel $\Phi$ is an entirely classical mapping, the definitions of strongly unlockable and weakly unlockable are equivalent.  A simplification for the classical case is as follows (we restrict to the perfect case for clarity).
A classical function $f : A \times B \to C$ is \emph{unlockable} if there exists a \emph{key input} $b_{0} \in B$ and a \emph{recovery algorithm} $\cA : C \times B \to C$ such that, for all $a \in A$ and $b \in B$, we have that $f(a,b)=\cA(f(a,b_{0}), b)$.
Intuitively, for an unlockable function, there exists an algorithm that can compute all values of $f(a, \cdot)$ given a one-time program for $f(a,\cdot)$, but this is okay, because a simulator, given one-shot oracle access to $f(a, \cdot)$, can also compute $f(a, b)$ for all $b$.
This function is ``learnable'' in one shot, and so a simulator can do everything any algorithm can.

Simple examples of strongly unlockable channels include all unitary channels of the form $\Phi:X\mapsto UXU^*$ for some unitary $U$ and all constant channels of the form $\Phi:X\mapsto\ptr{}{X}\sigma$ for some fixed state $\sigma$.
Simple examples of unlockable functions include permutations.

We can now see that strongly unlockable channels have one-time programs; in fact, trivial one-time programs.

\begin{theorem}
\label{thm:q-poss}
Let $\Phi:(\sA,\sB)\to\sC$ be a non-reactive, sender-oblivious polynomial-time quantum computable two-party functionality. Then if $\Phi$ is strongly unlockable, there exists an efficient, quantum non-interactive protocol which  quantum-UC-emulates $\cFunc{\Phi}{OTP}$ in the plain quantum model.

This holds in the perfect, statistical and computational cases.
\end{theorem}

\begin{proof}

The protocol is simple.
\begin{enumerate}

\item
{\bf Create:} The sender prepares the program register $\sP=(\sC,\sK)$ by preparing registers $(\sB',\sK)$ in the key state $\xi_0$ and applying $\Phi$ to $(\sA,\sB')$.
In other words, the sender's encoding channel $\enc$ is given by \( \enc = \Phi_0 \), where $\Phi_{0}$ is as in the definition of strongly unlockable.
\item
{\bf Execute:} Because $\Phi$ is strongly unlockable, the receiver can recover the action of $\Phi=\cA\circ\Phi_0$ simply by applying the recovery algorithm $\cA$ to $(\sP,\sB)$.
In other words, the receiver's decoding channel $\dec$ is given by \( \dec = \cA. \)

\end{enumerate}
Clearly, because $\Phi$ is strongly unlockable, the output of the honest receiver in the real model is indistinguishable from the output of the ideal model.

According to the discussion in Section~\ref{sec:otp-plain}, in order to show that this protocol is secure, it suffices to exhibit a simulator that can emulate the channel $\enc\ot\idsup{\sB}=\Phi_0\ot\idsup{\sB}$ using only the ideal functionality.
But this is easy: the simulator can emulate $\Phi_0\ot\idsup{\sB}$ simply by preparing registers $(\sB',\sK)$ in the key state $\xi_0$ and using the ideal functionality.
Formally, the simulator $(\simu_1,\simu_2)$ is specified by $\simu_1 :\sB\to(\sB',\sK,\sB):X\mapsto\xi_0\ot Y$ and $\simu_{2} = I$, so that
\[ \simu_2 \circ\Phi\circ\simu_1 : X\ot Y\mapsto (\Phi\ot\idsup{\sK})(X\ot\xi_0)\ot Y = \Phi_0(X)\ot Y \]
as desired.
(See Figure \ref{fig:qotp-security}.)
\end{proof}

Theorem \ref{thm:q-poss} is in the quantum model; it is not hard to see by its proof that if $\Phi$ is in fact a classical channel, then the resulting protocol $\pi$ is a purely classical protocol.

\subsection{Impossibility of one-time programs for arbitrary channels}

Having seen that, in the plain model, one-time programs do exist for strongly unlockable channels, we now see that they do not exist for weakly unlockable channels.

\begin{theorem}
\label{thm:impossibility-q-maps}
Suppose $\Phi:(\sA,\sB)\to\sC$ is a non-reactive, sender-oblivious polynomial-time quantum computable two-party functionality, and suppose that~$\Phi$ admits an efficient,  non-interactive quantum protocol which quantum-UC-emulates $\cFunc{\Phi}{OTP}$ in the plain model. Then $\Phi$ is  weakly unlockable.

This holds in the perfect, statistical and computational cases.
\end{theorem}

The intuition of the proof is as follows.  If a channel has a one-time program, then for any adversary there exists a simulator that can match the behaviour of the adversary.  In particular, there must be a simulator that matches the behaviour of the dummy adversary that just outputs the program state: thus, there must be an algorithm that can reconstruct the program state given the output of the channel, thus allowing computation for any output, meeting the definition of a weakly unlockable channel.

\begin{proof}
Suppose $\Phi$ is as in the theorem statement.
By the discussion in Section~\ref{sec:otp-plain},  there exists an encoding channel $\enc:\sA\to\sP$ for the sender and a decoding channel $\dec:(\sP,\sB)\to\sC$ for the receiver such that  $\dec\circ\enc$ is indistinguishable from $\Phi$.
Moreover, there exists a simulator $(\simu_1,\simu_2)$ with
$\simu_2\circ\Phi\circ\simu_1$ being indistinguishable from $\enc\ot\idsup{\sB}.$

We claim that $\Phi$ is weakly unlockable with key channel $\Xi_0=\simu_1$ and recovery algorithm given by \( \cA_\Psi = \dec\circ\Psi\circ\simu_2 \) for any choice of $\Psi$.
To this end fix a choice of $\Psi:\sB\to(\sB,\sE)$.
The channels $\Psi$ and $\enc$ commute, as they act on different input registers.
Thus  using $\approx$ to indicate computational, statistical, or perfect indistinguishability, depending on the scenario,
\[
  \Phi\circ\Psi
  \approx \dec\circ\enc\circ\Psi =
  \dec\circ\Psi\circ\enc \approx
  \underbrace{ \dec\circ\Psi\circ\simu_2 }_{\cA_\Psi} {}\circ\Phi\circ{}\underbrace{\simu_1}_{\Xi_0} \enspace . \qedhere
\]
\end{proof}

An alternate intuition for the impossibility result for classical functions can be found by considering rewinding.
Any correct one-time program state $\rho_{x}$ for a classical function $f(x, \cdot)$ must result in the receiver obtaining an output state $\rho_{x,y}$ that is (almost) orthogonal in the basis in which the receiver measures it, because the measurement of $\rho_{x,y}$ results in $f(x,y)$ with (almost) certainty.
As a result, measurement does not disturb the state (much), so the receiver can reverse the computation to obtain (almost) the program state again, and then rerun the computation to obtain $f(x,y')$ for a different $y'$.
It is possible to give a proof for impossibility of OTPs for classical functions in the plain quantum model using this rewinding argument. Impossibility for classical functions also follows as a special case of the impossibility shown in~\cite{BCS12}.

\subsection{A conjecture on unlockable channels}
\label{sec:conjecture}

As noted earlier, every strongly unlockable channel is also weakly unlockable.
We conjecture that the converse also holds: every weakly unlockable channel is also strongly unlockable.
Though we do not yet have a formal proof of this conjecture for arbitrary $\Phi$, we can nonetheless provide a high-level outline of what such a proof might look like.

\begin{conjecture}
\label{conj:bug}
Every channel $\Phi:(\sA,\sB)\to\sC$ that is weakly unlockable is also strongly unlockable.
\end{conjecture}

\begin{proof}[Proof outline]

Let $\sigma$ be any mixed state of a register $\sB'\equiv\sB$ and consider the channel
\[ \Psi_\mathrm{dummy}:\sB\to(\sB',\sB):B\mapsto\sigma\ot B \]
so that
\[ \Phi\circ\Psi_\mathrm{dummy}:(\sA,\sB)\to(\sC,\sB):A\ot B \mapsto \Phi(A\ot\sigma)\ot B. \]
That is, the channel $\Psi_\mathrm{dummy}$ swaps out the receiver's input register $\sB$ with a dummy register $\sB'$ in state $\sigma$ before $\Phi$ is applied.
By definition, the channel $\Phi\circ\Psi_\mathrm{dummy}$ does not touch $\sB$.

Let $\Xi_0,\cA_\Psi$ witness the weak unlockability of $\Phi$.
For the choice $\Psi=\Psi_\mathrm{dummy}$ it holds that the channel
\[ \cA_{\Psi_\mathrm{dummy}} \circ\Phi\circ\Xi_0 = \Phi\circ\Psi_\mathrm{dummy}\]
also must not touch $\sB$.
Hence it must be that the channel $\Phi\circ\Xi_0$ is invertible on $\sB$.

By enlarging the key register $\sK$ as needed, we may assume without loss of generality that $\Xi_0$ is implemented by an isometry $S$.
Thus, in order to be invertible on $\sB$ it must be that for each pure state $\ket{\psi}$ of $\sB$ the state $S\ket{\psi}$ of $(\sB,\sK)$ can be recovered after $\Phi$ acts on $\sB$.
In other words, $\Phi$ must be invertible on the subspace $\cB_\mathrm{inv}$ of the state space of register $\sB$ containing the image of $\ptr{\sK}{S\ket{\psi}\bra{\psi}S^*}$ for every choice of $\ket{\psi}$.

Before putting $\sB$ through $\Phi$ an alternate key channel $\Xi_0'$ could coherently swap out the subspace $\cB_\mathrm{inv}$ with some portion of the key register $\sK$ in any fixed state $\xi$.
Then, after receiving $\sC$ from $\Phi$, the alternate recovery algorithm $\cA_{\Psi_\mathrm{dummy}}'$ could perform the inversion operation to recover $\xi$, swap this state back into the key register, and then perform the inverse of the inversion operation on $S\ket{\psi}$ so as to apply $\Phi$ to this state.
These modifications yield the desired simulator.
\end{proof}

\ifthenelse{\equal{\compileACM}{0}}{
For the proof of the conjecture to go through, we need the invertibility of $\Phi$ on the aforementioned substance $\cB_{\mathrm{inv}}$ to hold.  It appears, then, that we have stumbled upon an interesting and deep question relating to the invertible subspaces of a channel, akin to the ``decoherence-free'' subspaces studied in the literature on quantum error correction.
}{}

\ifthenelse{\equal{\compileACM}{0}}{
\section{Constructing quantum OTP\lowercase{s} from OTM\lowercase{s}}
}
{ 
\section{Constructing quantum OTP\lowercase{s} \\ from OTM\lowercase{s}}
}

\label{sec:construction}

We now  state our main possibility theorem which establishes non-interactive unconditionally secure quantum computation
using OTM tokens.

\begin{theorem}
\label{thm:main}Let $\Phi$ be non-reactive,
sender-oblivious polynomial-time quantum computable two-party
functionality. Then there exists an efficient, quantum
non-interactive protocol which  statistically
quantum-UC-emulates  $\cFunc{\Phi}{OTP}$ in the case of a corrupt receiver, in
the $\cFunc{}{OTM}$-hybrid model.
\end{theorem}

The proof of Theorem~\ref{thm:main} follows directly from
Theorem~\ref{thm:main-quantum} below, together with
Corollary~\ref{cor:reactive-COTP}, the quantum lifting theorem
as well as Lemma~\ref{lem:R-transitive}.
\ifthenelse{\equal{\compileACM}{0}}{}{
The proof of Theorem~\ref{thm:main-quantum} is presented in the full version; here, we present an overview of the proof and related techniques.
}

\begin{theorem}\label{thm:main-quantum}Let $\Phi$ be a
non-reactive, sender-oblivious polynomial-time quantum computable
two-party functionality.  Then there exists an efficient,
statistically quantum-UC-secure non-interactive protocol
which realizes $\cFunc{\Phi}{OTP}$ in the case of a corrupt receiver, in the
$\cFunc{}{BR-OTP}$-hybrid model.
\end{theorem}

\ifthenelse{\equal{\compileACM}{0}}{

}
{

\subsection{Quantum authentication codes}
\label{sec:techniques-Q-auth}
A quantum authentication scheme consists of procedures for encoding and decoding quantum information with a secret classical key~$k$ such that an adversary with no knowledge of~$k$ who tampers with encoded data will be detected with high probability. Quantum authentication codes were first introduced by Barnum,  Cr{\'e}peau,  Gottesman,  Smith and Tapp~\cite{barnum2002authentication}.

Some of the known quantum authentication schemes have the following general form.
Quantum information is encoded according to some quantum error detecting code~$E$ chosen uniformly at random from a special family~$\mathscr{E}$ of codes.
The encoded quantum data is then encrypted according to the quantum one-time pad, meaning that a uniformly random Pauli operation $P$ is applied to the data encoded under~$E$. The secret classical key for schemes of this form is the pair $(E,P)$ describing the choice of code $E$ and encryption Pauli~$P$.
Authenticated quantum data is later verified by decrypting according to~$P$ and then decoding according to~$E$.
Verification passes only if the error syndrome for~$E$ indicates no errors.
\textbf{Terminology:}
In this paper authentication schemes of this form are called \emph{encode-encrypt schemes}.

This  construction is desirable due to the remarkable property  (known as the \emph{Pauli twirl}~\cite{dankert2009exact}) that the Pauli encryption serves to render \emph{any attack} on the scheme equivalent to a \emph{probabilistic Pauli attack} on data encoded with a random code $E\in\mathscr{E}$.
Thus, to establish a secure authentication scheme one need only construct a family~$\mathscr{E}$ of quantum error detecting codes that detect Pauli attacks with high probability over the choice of~$E\in\mathscr{E}$.

Our new \emph{trap authentication scheme} 
falls in this family of codes. The family of codes  $\mathscr{E}$ is based on any quantum error detecting code~$C$ with distance~$d$ that encodes a single qubit into $n$ qubits. Authentication consists of first encoding a qubit under~$C$ and then appending $n$ qubits set to $\ket{0}$ and $n$ qubits set to $\ket{+}$. A random permutation (indexed by a classical key) is then applied.  The first use of this code was implicit in the Shor--Preskill security proof for quantum key distribution~\cite{SP00} (see also~\cite{broadbent2009universal}).

\subsection{Computing on authenticated data}
\label{sec:techniques-computing-on-auth}

At a high level, our main protocol uses quantum authentication codes in order to protect the data from any tampering by the receiver during the computation. An authentication code is insufficient, however, because we want to implement a channel on this authenticated data, as specified by a quantum circuit. For this, we use techniques for quantum computing on authenticated data, as first established for the \emph{signed polynomial code}~\cite{Ben-OrC+06} (see also~\cite{AharonovBE10}), and recently (and independently of our work), for the Clifford authentication code~\cite{DNS12}. More specifically, computing on authenticated data allows acting on the encoded registers in order to implement a known gate, \emph{but without knowledge of the key}.  Normally, any non-identity operation would invalidate the authenticated state, but our encoded operations, together with a \emph{key update} operation, effectively \emph{forces} the application of the desired gate, as otherwise the state would fail verification under the updated key.

Encoded gates are executed in a manner similar to encoded gates in fault-tolerant quantum computation: some gates (such as Pauli gates or the controlled-not gate) are transversal, while other (such as the $\pi/8$ gate) require both an auxiliary register and classical interaction with an entity who knows the encoding keys. This classical interaction makes our quantum one-time program ``interactive'', but only at a classical level. Thus, by extending classical one-time programs to \emph{reactive} functionalities, 
we manage to encapsulate this interaction into a classical one-time program.

\lightparagraph{Comparison with other methods of computing on authenticated data}
Although methods for computing on authenticated data were developed prior to our work, we believe that the simplicity of our trap scheme is an advantage.
For example, whereas the trap scheme is defined for qubits, the signed polynomial scheme of Ben-Or \etal \cite{Ben-OrC+06} acts on $d$-dimensional qudits with~$d$ dictating the security of the scheme.
By necessity, the universal gate set $U_d$ for the polynomial scheme is different for each $d$.
In order to use the polynomial scheme for computation on qubits, one must first embed the desired qu\emph{bit} computation into a qu\emph{dit} computation.
(For example, a naive approach is to simply embed one qubit into each qudit and use only the first two dimensions.)
However this embedding is chosen, one must demonstrate that gates in the original qubit computation can be implemented efficiently using gates from $U_d$.
Normally, an efficient transition between universal gate sets is implied by the Solovay--Kitaev algorithm, but this techniques scales poorly with $d$ and so cannot be used for this purpose for large $d$.
(See Dawson and Nielsen \cite{DawsonN06} and the references therein.)
Although we fully expect such an embedding to admit an efficient implementation, it appears that the issue has not been addressed in the literature.

Compared with the gate implementations in the recent Clifford scheme of Dupuis \etal \cite{DNS12}, our trap scheme is less complex and requires less communication between the sender and the receiver (the Clifford scheme requires communication for \emph{all} circuit elements).

\subsection{Gate teleportation}

The main outline of our protocol is now becoming clearer: the receiver executes the encoded circuit, using techniques for computing on authenticated data.
But how does the receiver get the authenticated version of her data in the first place? And how does the receiver get the decoded output? We resolve this by using encoding and decoding \emph{gadgets} that are inspired by the gate teleportation technique of Gottesman and Chuang~\cite{GottesmanC99}. In this technique, a quantum register undergoes a transformation by a quantum circuit via its teleportation through a special entangled state. In our case, encoding is performed by teleporting the input qubit through an EPR pair, half of which has itself undergone the encoding operation. By revealing the classical result of the teleportation, the authentication key for the output of this process is determined. Decoding is similar. The encoding and decoding gadgets are prepared by the sender as part of the quantum one-time program.

\subsection{Overview of protocol}
\label{sec:techniques-protocol}
Given the tools and techniques described above, the structure of our protocol is as follows (although we suggest the use of the trap authentication scheme, the protocol and proof is applicable to any encode-encrypt quantum authentication scheme that admits computing on authenticated data).
\begin{enumerate} [leftmargin=*]
\item The sender gives the receiver an authenticated version of the sender's input, together with auxiliary states required for evaluating the target circuit. The sender also prepares encoding and decoding gadgets.
\item The sender gives the receiver a bounded reactive classical one-time program that emulates the classical interaction that would occur when using the encoding and decoding gadgets, as well as for computing on authenticated data.
\item The receiver uses the encoding gadget to encode his input; he then performs the target circuit on the authenticated data by performing encoded gates. Finally, he decodes the output using the decoding gadget. All classical interaction is done via the classical one-time program.
\end{enumerate}
As a proof technique, we specify that the target circuit be given as a controlled-unitary, with the control set to~1, for a reason described in the simulation sketch below.

\subsection{Overview of simulator for proof}
\label{sec:techniques-sim}
In order to prove UC security, we must establish that every real-world adversary has an ideal-world simulator. This is done via a rigorous mathematical analysis of any arbitrary attack%
\ifthenelse{\equal{\compileACM}{0}}{%
; the exposition of our proof is aided by a \emph{table representation} we have developed for pure quantum states (see Section~\ref{sec:tabular-representation})%
}{}%
.
The simulator prepares a quantum one-time program as in Section~\ref{sec:techniques-protocol}, with the following modifications:
\begin{enumerate} [leftmargin=*]
\item The encoding gadget is split into two halves. The first half is a simple EPR pair, used to \emph{extract} the input from the adversary; the second half is an encoding gadget, used to \emph{insert} the output of the ideal functionality into the computation.
\item The control-bit for the controlled-unitary is~$0$.
\item The encoded input is a dummy encoded input.
\end{enumerate}
The simulator then executes the adversary on this ``quantum one-time program''. After the use of the encoding gadget, the receiver's input is determined and this is used as input into the ideal functionality. The output of the ideal functionality is then returned into the computation via the encoding gadget. Because the control bit is set to~0, the computation actually performs the identity. When the adversary is honest, the output will therefore be correct. For any behaviour of the adversary, our analysis shows that the ideal and real worlds are indistinguishable.
The simulator thus indistinguishably emulates the real-world behaviour of any adversary with just a single call to the ideal functionality, which establishes UC security.

}

\ifthenelse{\equal{\compileACM}{1}}{}{
The proof of Theorem~\ref{thm:main-quantum} is presented in the
 following sections, which we briefly highlight here; a detailed outline follows in the next section.

\begin{enumerate}
\item Section~\ref{sec:trap} presents our new \emph{trap authentication scheme}, a type of quantum authentication code.  We show how perform a universal set of quantum gates ($X$, $Y$, $Z$, {\sc cnot}, $i$-shift and $\pi/8$ phases, and $H$) on authenticated data without knowing the authentication key.
\item Section~\ref{sec:security} presents our protocol for quantum one-time programs and the proof its security.  Since computation on authenticated data requires updates to be performed that are dependent on the authentication key, our protocol uses a reactive classical one-time program (based on one-time memories) to allow the receiver to non-interactively implement the required operations to correctly compute on the sender's authenticated data.
\end{enumerate}

The following sections \ref{sec:techniques-Q-auth}--\ref{sec:techniques-sim} provide an overview of the proof and related techniques.


\subsection{Quantum authentication codes}
\label{sec:techniques-Q-auth}
A quantum authentication scheme consists of procedures for encoding and decoding quantum information with a secret classical key~$k$ such that an adversary with no knowledge of~$k$ who tampers with encoded data will be detected with high probability. Quantum authentication codes were first introduced by Barnum,  Cr{\'e}peau,  Gottesman,  Smith and Tapp~\cite{barnum2002authentication}.

Some of the known quantum authentication schemes have the following general form.
Quantum information is encoded according to some quantum error detecting code~$E$ chosen uniformly at random from a special family~$\mathscr{E}$ of codes.
The encoded quantum data is then encrypted according to the quantum one-time pad, meaning that a uniformly random Pauli operation $P$ is applied to the data encoded under~$E$. The secret classical key for schemes of this form is the pair $(E,P)$ describing the choice of code $E$ and encryption Pauli~$P$.
Authenticated quantum data is later verified by decrypting according to~$P$ and then decoding according to~$E$.
Verification passes only if the error syndrome for~$E$ indicates no errors.
\textbf{Terminology:}
In this paper authentication schemes of this form are called \emph{encode-encrypt schemes}.

This  construction is desirable due to the remarkable property  (known as the \emph{Pauli twirl}~\cite{dankert2009exact}) that the Pauli encryption serves to render \emph{any attack} on the scheme equivalent to a \emph{probabilistic Pauli attack} on data encoded with a random code $E\in\mathscr{E}$.
Thus, to establish a secure authentication scheme one need only construct a family~$\mathscr{E}$ of quantum error detecting codes that detect Pauli attacks with high probability over the choice of~$E\in\mathscr{E}$.

Our new \emph{trap authentication scheme} 
falls in this family of codes. The family of codes  $\mathscr{E}$ is based on any quantum error detecting code~$C$ with distance~$d$ that encodes a single qubit into $n$ qubits. Authentication consists of first encoding a qubit under~$C$ and then appending $n$ qubits set to $\ket{0}$ and $n$ qubits set to $\ket{+}$. A random permutation (indexed by a classical key) is then applied.  The first use of this code was implicit in the Shor--Preskill security proof for quantum key distribution~\cite{SP00} (see also~\cite{broadbent2009universal}).

\subsection{Computing on authenticated data}
\label{sec:techniques-computing-on-auth}

At a high level, our main protocol uses quantum authentication codes in order to protect the data from any tampering by the receiver during the computation. An authentication code is insufficient, however, because we want to implement a channel on this authenticated data, as specified by a quantum circuit. For this, we use techniques for quantum computing on authenticated data, as first established for the \emph{signed polynomial code}~\cite{Ben-OrC+06} (see also~\cite{AharonovBE10}), and recently (and independently of our work), for the Clifford authentication code~\cite{DNS12}. More specifically, computing on authenticated data allows acting on the encoded registers in order to implement a known gate, \emph{but without knowledge of the key}.  Normally, any non-identity operation would invalidate the authenticated state, but our encoded operations, together with a \emph{key update} operation, effectively \emph{forces} the application of the desired gate, as otherwise the state would fail verification under the updated key.

Encoded gates are executed in a manner similar to encoded gates in fault-tolerant quantum computation: some gates (such as Pauli gates or the controlled-not gate) are transversal, while other (such as the $\pi/8$ gate) require both an auxiliary register and classical interaction with an entity who knows the encoding keys. This classical interaction makes our quantum one-time program ``interactive'', but only at a classical level. Thus, by extending classical one-time programs to \emph{reactive} functionalities, 
we manage to encapsulate this interaction into a classical one-time program.

\lightparagraph{Comparison with other methods of computing on authenticated data}
Although methods for computing on authenticated data were developed prior to our work, we believe that the simplicity of our trap scheme is an advantage.
For example, whereas the trap scheme is defined for qubits, the signed polynomial scheme of Ben-Or \etal \cite{Ben-OrC+06} acts on $d$-dimensional qudits with~$d$ dictating the security of the scheme.
By necessity, the universal gate set $U_d$ for the polynomial scheme is different for each $d$.
In order to use the polynomial scheme for computation on qubits, one must first embed the desired qu\emph{bit} computation into a qu\emph{dit} computation.
(For example, a naive approach is to simply embed one qubit into each qudit and use only the first two dimensions.)
However this embedding is chosen, one must demonstrate that gates in the original qubit computation can be implemented efficiently using gates from $U_d$.
Normally, an efficient transition between universal gate sets is implied by the Solovay--Kitaev algorithm, but this techniques scales poorly with $d$ and so cannot be used for this purpose for large $d$.
(See Dawson and Nielsen \cite{DawsonN06} and the references therein.)
Although we fully expect such an embedding to admit an efficient implementation, it appears that the issue has not been addressed in the literature.

Compared with the gate implementations in the recent Clifford scheme of Dupuis \etal \cite{DNS12}, our trap scheme is less complex and requires less communication between the sender and the receiver (the Clifford scheme requires communication for \emph{all} circuit elements).

\subsection{Gate teleportation}

The main outline of our protocol is now becoming clearer: the receiver executes the encoded circuit, using techniques for computing on authenticated data.
But how does the receiver get the authenticated version of her data in the first place? And how does the receiver get the decoded output? We resolve this by using encoding and decoding \emph{gadgets} that are inspired by the gate teleportation technique of Gottesman and Chuang~\cite{GottesmanC99}. In this technique, a quantum register undergoes a transformation by a quantum circuit via its teleportation through a special entangled state. In our case, encoding is performed by teleporting the input qubit through an EPR pair, half of which has itself undergone the encoding operation. By revealing the classical result of the teleportation, the authentication key for the output of this process is determined. Decoding is similar. The encoding and decoding gadgets are prepared by the sender as part of the quantum one-time program.

\subsection{Overview of protocol}
\label{sec:techniques-protocol}
Given the tools and techniques described above, the structure of our protocol is as follows (although we suggest the use of the trap authentication scheme, the protocol and proof is applicable to any encode-encrypt quantum authentication scheme that admits computing on authenticated data).
\begin{enumerate} [leftmargin=*]
\item The sender gives the receiver an authenticated version of the sender's input, together with auxiliary states required for evaluating the target circuit. The sender also prepares encoding and decoding gadgets.
\item The sender gives the receiver a bounded reactive classical one-time program that emulates the classical interaction that would occur when using the encoding and decoding gadgets, as well as for computing on authenticated data.
\item The receiver uses the encoding gadget to encode his input; he then performs the target circuit on the authenticated data by performing encoded gates. Finally, he decodes the output using the decoding gadget. All classical interaction is done via the classical one-time program.
\end{enumerate}
As a proof technique, we specify that the target circuit be given as a controlled-unitary, with the control set to~1, for a reason described in the simulation sketch below.

\subsection{Overview of simulator for proof}
\label{sec:techniques-sim}
In order to prove UC security, we must establish that every real-world adversary has an ideal-world simulator. This is done via a rigorous mathematical analysis of any arbitrary attack%
\ifthenelse{\equal{\compileACM}{0}}{%
; the exposition of our proof is aided by a \emph{table representation} we have developed for pure quantum states (see Section~\ref{sec:tabular-representation})%
}{}%
.
The simulator prepares a quantum one-time program as in Section~\ref{sec:techniques-protocol}, with the following modifications:
\begin{enumerate} [leftmargin=*]
\item The encoding gadget is split into two halves. The first half is a simple EPR pair, used to \emph{extract} the input from the adversary; the second half is an encoding gadget, used to \emph{insert} the output of the ideal functionality into the computation.
\item The control-bit for the controlled-unitary is~$0$.
\item The encoded input is a dummy encoded input.
\end{enumerate}
The simulator then executes the adversary on this ``quantum one-time program''. After the use of the encoding gadget, the receiver's input is determined and this is used as input into the ideal functionality. The output of the ideal functionality is then returned into the computation via the encoding gadget. Because the control bit is set to~0, the computation actually performs the identity. When the adversary is honest, the output will therefore be correct. For any behaviour of the adversary, our analysis shows that the ideal and real worlds are indistinguishable.
The simulator thus indistinguishably emulates the real-world behaviour of any adversary with just a single call to the ideal functionality, which establishes UC security.


\section{The trap authentication scheme}
\label{sec:trap}

In this section we introduce the trap authentication scheme, which is an example of an encrypt-encode scheme as described in Section \ref{sec:techniques-Q-auth}.
We show how to implement gates from a universal set on data authenticated under this scheme (including measurement in the standard basis), from which it follows that the trap scheme admits QCAD.

\subsection{Definitions and notation}
\label{sec:trap:defs}

In this paper we assume that a quantum error correcting code is specified by a unitary operation~$E$ that can be implemented by a circuit consisting entirely of Clifford gates.
A data register $\sD$ is encoded under code~$E$ by preparing two syndrome registers $(\sX,\sZ)$ in the $\ket{0}$ state and applying~$E$ to~$(\sD,\sX,\sZ)$.
Data is decoded by applying the inverse circuit~$E^*$ and measuring the syndrome registers $(\sX,\sZ)$ in the computational basis.
Any non-zero syndrome measurement result indicates an error (or that cheating has been detected, depending on the context).
These assumptions are met, for example, by every stabilizer code.
In order to minimize the number of symbols for distinct quantum registers we adopt the convention that the tilde $\tilde{\phantom{x}}$  denotes an \emph{encoded register}.
For example, the encoded register $\tilde\sD$ consists of a data portion $\sD$ plus two syndrome registers $\sX,\sZ$ and can be viewed as a triple $\tilde\sD=(\sD,\sX,\sZ)$.

Let $E$ be a code for some data register $\sD$ and let $Q$ be any Pauli acting on $\tilde\sD$.
As $E$ is a Clifford circuit there must be a Pauli $Q_E$ acting on $\tilde\sD$ with $QE=EQ_E$.

\begin{definition}
\label{def:secure-against-Pauli-attacks}
A family $\mathscr{E}$ is said to be \emph{$\epsilon$-secure against Pauli attacks} if for each fixed choice of Pauli $Q$ acting on $\tilde\sD$ it holds that the probability (taken over a uniformly random choice of code $E\in\mathscr{E}$) that $Q_E$ acts nontrivially on logical data and yet has no error syndrome is at most $\epsilon$.
\end{definition}

Formally defining the security of an authentication scheme is tricky work.
(See Barnum \etal \cite{barnum2002authentication} for a discussion of definitional issues.)
Fortunately, the task becomes much easier if we restrict attention to encode-encrypt schemes described in Section \ref{sec:techniques-Q-auth}.
This happy state of affairs is a consequence of the fact that an arbitrary attack on data authenticated under an encode-encrypt scheme is equivalent to a probabilistic mixture of Pauli attacks.
(See below for futher discussion.)
Thus, an encode-encrypt scheme is $\epsilon$-secure against against arbitrary attacks if and only if the underlying family $\mathscr{E}$ of codes is $\epsilon$-secure against Pauli attacks.

\subsection{Properties of every encode-encrypt authentication scheme}
\label{sec:trap:properties}

Before introducing the trap scheme we review some facts about encrypt-encode authentication schemes as described in Section \ref{sec:techniques-Q-auth}.
These facts were largely known to the community prior to the present work and hence a full discussion is relegated to Appendix \ref{appendix:encode-encrypt}.
Our purpose in this paper is to collect these facts in one convenient place and to present them in a way that is amenable to a discussion of QCAD.
In Appendix \ref{appendix:encode-encrypt} we formalize and prove the following facts:
\begin{enumerate}

\item \label{it:pauli-general-secure}
Any family $\mathscr{E}$ of codes that is $\epsilon$-secure against Pauli attacks immediately induces an $\epsilon$-secure encode-encrypt scheme via the construction described in Section \ref{sec:techniques-Q-auth}.

\item \label{it:CSS-measure}
If the codes in $\mathscr{E}$ are CSS codes then measurement of logical data in the computational basis can be implemented by bitwise measurement of authenticated data followed by a classical decoding process in order to determine the measurement result.

\item
The measure-then-decode procedure of property \ref{it:CSS-measure} is equivalent to a decode-then-measure procedure in which the register $\tilde\sD$ is first de-authenticated and then the logical register $\sD$ is measured in the computational basis.
In this procedure, the $X$-syndrome register is also measured in the computational basis so as to check for errors, but the $Z$-syndrome register is simply discarded without any verification.

\item
For any encode-encrypt scheme thus constructed, the measure-then-decode procedure (or equivalently, the decode-then-measure procedure) of property \ref{it:CSS-measure} is also $\epsilon$-secure against arbitrary attacks.

\item \label{it:key-reuse}
In any encode-encrypt scheme the code key $E$ can be re-used to authenticate multiple distinct data registers, provided that each new register gets its own fresh Pauli key $P$.

\end{enumerate}
The proof of property \ref{it:key-reuse} is easy enough that we can give it immediately.
As noted in Section \ref{sec:techniques-Q-auth} (and established in the proof of property \ref{it:pauli-general-secure} in Appendix \ref{appendix:encode-encrypt}), the Pauli encryption serves to render any attack on an encode-encrypt scheme equivalent to a probabilistic mixture of Pauli attacks.
By definition a Pauli attack is a product attack on each physical qubit in the authenticated registers, so security against attacks on one register implies security against attacks on all registers.
This observation was originally made 
in the analysis of the polynomial authentication scheme \cite{Ben-OrC+06}.
(See Ref.\ \cite{AharonovBE10} for a slightly more detailed discussion.)

\subsection{Trap codes yield a secure authentication scheme}
\label{sec:proof-auth-trap}

In this section we describe a method by which any code $E$ with distance $d$ can be used to construct a family $\mathscr{E}$ of codes that is $(2/3)^{d/2}$-secure against Pauli attacks.
Codes of this form shall be called \emph{trap codes}.
It follows immediately from the discussion of Section \ref{sec:trap:properties} that trap codes yield a secure authentication scheme via the encode-encrypt construction of Section \ref{sec:techniques-Q-auth}.
This authentication scheme shall be called the \emph{trap scheme}.

Furthermore, if the underlying code $E$ is a CSS code then so is every member of the associated family~$\mathscr{E}$ of trap codes.
In this case it follows from the discussion of Section \ref{sec:trap:properties} that measurement of logical data in the computational basis can be implemented securely by bitwise measurement of physical data plus classical decoding.

For convenience we restrict attention to codes that encode one logical qubit into $n$ physical qubits.
(These are called $[[n,1,d]]$-codes.)
Given such a code, each member of the associated family $\mathscr{E}$ of trap codes is a $[[3n,1,d]]$ code that is uniquely specified by a permutation $\pi$ of $3n$ elements.
In particular, for each $\pi$ we construct a Clifford encoding circuit $E_\pi$ as follows.
\begin{enumerate}

\item
Encode the data qubit $\sD$ under $E$, producing an $n$-qubit system $(\sD,\sX,\sZ)$.

\item
Introduce two new $n$-qubit syndrome registers $\sX',\sZ'$ in states $\ket{0}^{\ot n},\ket{+}^{\ot n}$, respectively.

\item
Permute all $3n$ qubits of $(\sD,\sX,\sZ,\sX',\sZ')$ according to $\pi$.

\end{enumerate}
The $X$- and $Z$-syndrome registers for the code $E_\pi$ are $(\sX,\sX')$ and $(\sZ,\sZ')$, respectively so that $\tilde\sD=(\sD,(\sX,\sX'),(\sZ,\sZ'))$.
Security of this family against Pauli attacks is easy to prove.

\begin{proposition}[Security of trap codes against Pauli attacks]
\label{prop:trap-secure-against-Paulis}

  The family $\mathscr{E}$ of trap codes based on a code of distance $d$ is $(2/3)^{d/2}$-secure against Pauli attacks.

\end{proposition}

\textbf{Remark.}
The bound in Proposition \ref{prop:trap-secure-against-Paulis} is quite weak and can probably be strengthened significantly by a tighter analysis.
All that really matters is that the security parameter decreases exponentially in $d$.

\begin{proof}

Let $Q$ be a $3n$-qubit Pauli.
In order for $Q$ to act nontrivially on logical data it must have weight $w\geq d$, owing to the fact that the underlying code $E$ has distance $d$.
In this case $Q$ must distribute $w$ non-identity qubit Pauli operations over the $3n$ qubits without triggering any of the traps.
Let us bound the probability of such an event.

In order to have weight $w$ the Pauli $Q$ must specify either an $X$-Pauli on at least $w/2$ qubits or a $Z$-Pauli on at least $w/2$ qubits.
We analyze only the first case; a similar analysis applies to the second case.
If any of these qubits belong to the register $\sX'$ then $Q$ will be detected as an error.
Thus, to avoid detection all $w/2$ of these qubits must not belong to $\sX'$---a sample-without-replacement event whose probability of success is bounded by the probability of a successful sample-with-replacement event.
The probability of success in any one sample is at most $2/3$ and so the probability of $w/2$ successful samples with replacement is at most $(2/3)^{w/2}$.
\end{proof}

\subsection{Performing Gates on the Trap Code}
\label{appendix:TrapCode-GateGadget}

Authentication schemes that also allow for the implementation of a universal set of quantum gates on authenticated data without knowing the key hold great promise for a host of cryptographic applications.
In this section, we exhibit a universal gate set together with implementations of each gate in that set for the trap scheme. As discussed in Section~\ref{sec:techniques-computing-on-auth}, these techniques are used in our quantum one-time programs.

Let us be more explicit about what it means to apply gates to authenticated quantum data without knowing the key.
It is helpful to think of two parties: a trusted \emph{verifier} who prepares authenticated data with secret classical key $k$ and a malicious \emph{attacker} who is to act upon the authenticated data without knowledge of $k$.
The goal is to construct an authentication scheme with the property that for each gate $G$ belonging to some universal set of gates there exists a \emph{gadget} circuit $\tilde G$ that the attacker can perform on authenticated data so as to implement a logical $G$.
Furthermore, we require that the gadget $\tilde G$ be \emph{independent of the choice of classical key $k$} so that it may be implemented by an attacker without knowledge of $k$.

Normally, any non-identity gadget $\tilde G$ would invalidate the authenticated state.
We therefore require a scheme which allows the verifier to validate the state again simply by updating the classical key $k\mapsto k'$.
Moreover, by updating the key in this way the verifier effectively \emph{forces} the attacker to apply the desired gadget $\tilde G$ as otherwise the state would fail verification under the updated key $k'$.

Some gadgets are quite simple.
For example, we shall soon see that the gadget for a logical controlled-NOT in the trap scheme is simply a bitwise controlled-NOT applied to the physical qubits in the authenticated registers.
Other gadgets, however, are more complicated, owing partly to the fact that there is no quantum error detecting code that admits bitwise implementation of every gate in a universal gate set.

Following the example of the polynomial scheme of Ben-Or \etal \cite{Ben-OrC+06}, we borrow from the study of fault-tolerant quantum computation to complete a universal gate set by means of so-called ``magic states''.
A \emph{magic state gadget} for a logical gate $G$ is a circuit that takes an additional (authenticated) ancillary register as input, performs a measurement, and then performs a correction based on the result of that measurement.
A magic state gadget for $G$ works only when the ancillae are prepared in a special (authenticated) \emph{magic state} tailored specifically for the gate~$G$.
(For example, Ben-Or \etal exhibited a magic state and associated gadget for the generalized Toffoli gate under the polynomial scheme \cite{Ben-OrC+06}.)
Thus, implementation of a universal gate set on authenticated data requires the ability to prepare authenticated magic states and the ability to measure authenticated qubits in the computational basis.

\subsection{A universal gate set for the trap scheme}
\label{sec:trap:universal}

In Section \ref{sec:proof-auth-trap} we stipulated that a trap scheme can be constructed from any underlying code $E$ chosen from a large class of codes.
We also noted that if $E$ is a CSS code then so are its associated trap codes, from which it follows that measurement of logical data can be implemented by a simple bitwise measurement of authenticated data.

In addition to being a CSS code, our implementation of a universal gate set for the trap scheme requires an underlying code $E$ for which a logical Hadamard gate $H$ is implemented by bitwise $H$ on each physical qubit.
CSS codes with this property are sometimes called \emph{self-dual CSS codes}.
The seven-qubit Steane code is one example of a self-dual CSS code that suffices for this purpose.
Specifically, it follows from Proposition \ref{prop:trap-secure-against-Paulis} that if our trap scheme is to have security $(2/3)^{d/2}$ then it suffices to base the trap scheme upon the Steane code nested a sufficient number of levels so as to achieve distance $d$.

Our universal gate set consists of the following gates, listed in the order they are presented in the following subsections.
\begin{enumerate}

\item
The standard single-qubit Pauli gates, denoted $X,Y,Z$.

\item
The standard two-qubit controlled-NOT gate, denoted $\cnot$.

\item
The single-qubit $i$-shift phase gate, denoted $K$ and specified by
\( K : \ket{a}\mapsto i^a\ket{a} \) for $a\in\set{0,1}$.

\item
The single-qubit $\pi/8$-phase gate, denoted $T$ and specified by
\( T : \ket{a}\mapsto e^{ai\pi/4}\ket{a} \) for $a\in\set{0,1}$.

\item
The standard single-qubit Hadamard gate, denoted $H$.

\end{enumerate}
This gate set is redundant in the sense that only $\set{\cnot,H,T}$ are required to achieve universality.
Indeed, $T^2=K$ and $K^2=Z$, so why bother listing these extra gates?
The answer is that the gadgets for the Pauli and controlled-NOT gates are very simple.
By contrast, the gadget for $K$ is a magic state gadget that requires $\cnot$ and $Y$ and the gadget for $T$ is a magic state gadget that requires $\cnot$, $X$, and $K$.
Thus, in order to get a $T$ gate we must ``bootstrap'' our way up from Pauli gates, $\cnot$, and $K$.
(If an application calls for only Clifford circuits on authenticated data then the $T$ gate construction can be ignored, as $\set{\cnot,H,K}$ suffice to generate the Clifford circuits.)

\subsubsection{Pauli gates}
\label{sec:universal:pauli}

The gadgets for each Pauli gate $X,Y,Z$ are empty.
As with the polynomial scheme \cite{Ben-OrC+06}, in order to implement a logical Pauli gate in the trap scheme the attacker does absolutely nothing to the authenticated register and the verifier simply updates the Pauli key according to the desired Pauli gate.

In particular, logical Paulis for the seven-qubit Steane code admit straightforward bitwise implementations.
Thus, the verifier can implement a logical Pauli $Q$ in the trap scheme by modifying the Pauli key according to $P\mapsto PQ^{\ot 7}$ where $Q^{\ot 7}$ is applied to registers $(\sD,\sX,\sZ)$, leaving the Pauli key for the trap registers $\sX',\sZ'$ unchanged.

\subsubsection{The controlled-NOT gate}

The gadget for a controlled-NOT from logical qubit $a$ to logical qubit $b$ is a straightforward bitwise $\cnot$ applied from each physical qubit of $a$ to its corresponding physical qubit in $b$.

To see that this simple bitwise gadget implements logical $\cnot$ recall that every CSS code (including the Steane code) admits a bitwise implementation of logical $\cnot$.
Moreover, the $\cnot$ applied bitwise to the trap registers acts trivially on those registers:
\begin{align*}
  \cnot &: \ket{0}\ket{0} \mapsto \ket{0}\ket{0} \\
  &: \ket{+}\ket{+} \mapsto \ket{+}\ket{+}
\end{align*}
Finally, observe that bitwise $\cnot$ is invariant under permutation of the physical qubits \emph{provided that both data blocks are subjected to the same permutation}.
(See Section \ref{sec:trap:re-use} for further discussion.)

The Pauli key is updated according to the well-known effect of $\cnot$ on Pauli operations.
In particular, if the Pauli keys for the $i$th physical qubit from both data blocks are $X^p Z^q, X^rZ^s$, respectively, then the updated Pauli keys for these physical qubits are
\( X^p Z^{q+s}, X^{p+r} Z^s \).

\subsubsection{The $i$-shift gate}
\label{sec:universal:i-shift}

Readers familiar with the Steane code know that a bitwise $K$ gate applied to each physical qubit implements $K^*$ on logical data.
At first glance one might therefore hope that the $K$ gate, like $\cnot$, admits a simple bitwise gadget under the trap scheme.
Unfortunately, trap codes do not admit bitwise implementation of the $K$ gate even if the underlying code does admit such an implementation.
Bitwise implementation fails for trap codes because the trap qubits prepared in state $\ket{+}$ are mapped by $K$ to
$K\ket{+}=\ket{0} + i\ket{1}$.
A trap qubit in this state is detected as a $Z$-error with probability $1/2$.

Instead we require a more complicated magic state gadget for $K$ that uses only Pauli and $\cnot$ gates together with measurement in the computational basis.
Our gadget is a simple modification of the well-known fault-tolerant construction for the $\pi/8$-gate of Boykin \emph{et al.}\ \cite{BoykinM+99}. 
The logical gadget for the $K$ gate is depicted as follows.

\begin{align}
\Qcircuit @C=1em @R=0.7em {
  \lstick{\ket{0}+i\ket{1}} & \ctrl{1} & \gate{Y} \cwx[1] & \rstick{K\ket{\psi}} \qw \\
  \lstick{\ket{\psi}} & \targ & \meter
}
\end{align}
Here $\ket{\psi}$ denotes the arbitrary state of the data qubit; the magic state for this gadget is $\ket{0}+i\ket{1}$.
The physical gadget to be implemented by the attacker on authenticated data is derived from the above logical gadget by replacing the input qubits with authenticated registers and by replacing the logical $\cnot$, $Y$, and measurement with their respective physical gadgets.

Once the authenticated magic state has been prepared, all the gates in this gadget except the measurement can be implemented by solely by the attacker.
The verifier's knowledge of the secret key is required in order to decode the measurement result, which indicates whether a $Y$ correction is needed.

Since $Y$ is a Pauli gate and since Pauli gates require no action from the attacker, this gadget can be implemented with only one-way classical interaction from attacker to verifier.
In particular, the attacker implements the measurement by bitwise measurement of physical data as described in Section \ref{sec:trap:properties}.
The verifier decodes the measurement result (and checks for tampering in the process) and then implements the $Y$ correction (if it is needed) simply by updating the Pauli key for that qubit as described in Section \ref{sec:universal:pauli}.

\subsubsection{The $\pi/8$-phase gate}
\label{sec:universal:T}

The gadget for the $T$ gate is a magic state gadget that is very similar to the magic state gadget for the $K$ gate described Section \ref{sec:universal:i-shift} and consequently much of the discussion from that section applies here.
The original fault-tolerant construction due to Boykin \emph{et al.}\ \cite{BoykinM+99} can be used verbatim as the logical gadget for the $T$ gate in the trap scheme.
Their construction is reproduced below.

\begin{align}
\Qcircuit @C=1em @R=0.7em {
  \lstick{\ket{0}+e^{i\pi/4}\ket{1}} & \ctrl{1} & \gate{KX} \cwx[1] & \rstick{T\ket{\psi}} \qw \\
  \lstick{\ket{\psi}} & \targ & \meter
}
\end{align}
The magic state for this gadget is $\ket{0}+e^{i\pi/4}\ket{1}$.

Whereas the gadget for $K$ presented in Section \ref{sec:universal:i-shift} specified only a Pauli $Y$ correction, the correction required in the $T$ gadget is the Clifford gate $KX$.
Two-way communication between verifier and attacker is required to implement this gadget because the verifier must inform the attacker as to whether to apply a $K$ gate.
Naturally, this $K$ gate, if it is required, is achieved via the magic state gadget presented Section \ref{sec:universal:i-shift}, which requires a separate magic state of its own.

Since $K$ is not a Pauli gate subsequent computation on authenticated quantum data cannot proceed until the correction is applied (if it is needed).
Thus, the verifier must decode the classical measurement result and reply immediately to the attacker with instructions as to whether a correction is required.

\subsubsection{The Hadamard gate}

Similar to the $K$ gate, a logical $H$ gate can be implemented under the Steane code by applying $H$ bitwise to each physical qubit.
One might therefore hope that the $H$ gate admits a simple bitwise gadget under the trap scheme.
Unfortunately, as with the $K$ gate, trap codes do not admit bitwise implementation of $H$ even if the underlying code does admit such an implementation.
Bitwise implementation fails for trap codes because the $\ket{0}$ and $\ket{+}$ trap qubits are swapped by the action of bitwise $H$.
Each trap qubit is thus in a state that is detected as an error with probability $1/2$.

As with the implementations of the $K$ and $T$ gates described previously we shall implement the Hadamard by a magic state gadget.
Whereas the gadgets for $K,T$ are compact and efficient, the simplest known magic state gadget for the Hadamard gate is the teleport-through-Hadamard circuit, which is a special case of the gate teleportation protocol of Gottesman and Chuang \cite{GottesmanC99}.
This circuit is depicted in Figure \ref{fig:teleport-hadamard}.
The magic state for this gadget is the two-qubit maximally entangled state $\ket{00}+\ket{01}+\ket{10}-\ket{11}$.

\begin{figure}[htb]
\begin{center}
\begin{tikzpicture}
\node[left] (inBell) at (-1.5,0.5) {$\ket{00} + \ket{11}$};
\node[left] (inPsi) at (-1.5,-1) {$\ket{\psi}$};
\node[rectangle,thick,draw=black,inner xsep=3pt,inner ysep=3pt] (H1) at (0,1) {$H$};
\draw (inBell.east) .. controls +(0.8,0) and +(-0.8,0) .. (-0.5,1) -- (H1);
\coordinate (cnotStart) at (1.5,-1);
\fill (cnotStart) circle (2.5pt);
\draw (inPsi) -- (cnotStart);
\node[circle,thick,draw=black,inner sep=-1.5pt] (cnotPlus) at (1.5,0) {\bf +};
\draw (cnotStart) -- (cnotPlus);
\draw (inBell.east) .. controls +(0.8,0) and +(-0.8,0) .. (-0.5,0) -- (cnotPlus);
\node[rectangle,thick,draw=black,inner xsep=3pt,inner ysep=3pt] (H2) at (2.5,-1) {$H$};
\draw (cnotStart) -- (H2);
\node[rectangle,thick,draw=black,inner xsep=3pt,inner ysep=3pt] (M1) at (3.75,0) {$\qmeter$};
\draw (cnotPlus) -- (M1);
\node[rectangle,thick,draw=black,inner xsep=3pt,inner ysep=3pt] (M2) at (3.75,-1) {$\qmeter$};
\draw (H2) -- (M2);
\node[rectangle,thick,draw=black,inner xsep=3pt,inner ysep=3pt] (Z) at (4.75,1) {$Z$};
\draw (H1) -- (Z);
\node[rectangle,thick,draw=black,inner xsep=3pt,inner ysep=3pt] (X) at (5.75,1) {$X$};
\draw (Z) -- (X);
\draw[double] (M2) -| (X);
\draw[double] (M1) -| (Z);
\node[right] (outPsi) at (6.5,1) {$H\ket{\psi}$};
\draw (X) -- (outPsi);
\draw[dashed,draw=darkgreen] (-3.6,1.5)--(-3.6,-0.3)--(0.5,-0.3)--(0.5,1.5)--(-3.6,1.5);
\node[above] at (-1.55,1.5) {\small \color{darkgreen} Magic state};
\draw[dashed,draw=blue] (1.1,0.5)--(1.1,-1.5)--(4.3,-1.5)--(4.3,0.5)--(1.1,0.5);
\node[below] at (2.7,-1.5) {\small \color{blue} Bell measurement};
\end{tikzpicture}
\end{center}
\vspace{-1em}
\caption{Teleport-through-Hadamard circuit}
\label{fig:teleport-hadamard}
\end{figure}
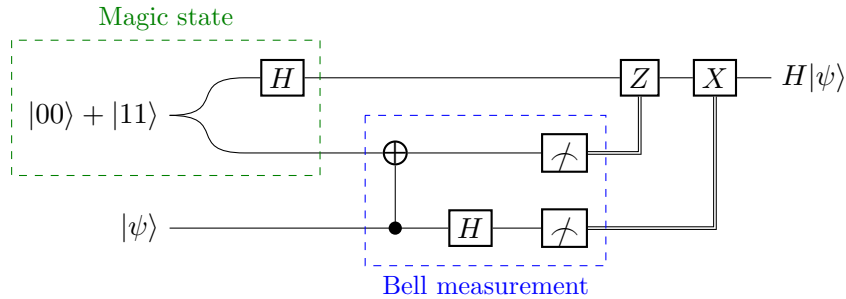

Implementation of the Bell measurement appearing in the above circuit requires that a Hadamard gate be applied to one of the two qubits immediately prior to measurement.
At first glance this requirement might appear circular, as we require a Hadamard gate in order to implement the Hadamard gate.
We claim, however, that it is possible to implement the Hadamard gate bitwise on authenticated data \emph{provided that the qubit is measured immediately afterward} as is the case for a Bell measurement.

This claim is not difficult to justify.
As mentioned above, the effect of the bitwise $H$ gate is to swap the syndrome registers $\sX',\sZ'$ immediately prior to measurement.
But it is trivial to modify any measure-then-decode procedure (such as that mentioned in Section \ref{sec:trap:properties}) so as to take this swap into account.

As with the gadget for $K$ presented in Section \ref{sec:universal:i-shift}, the correction in our gadget for $H$ is a Pauli gate.
Hence, implementation of this gadget requires only one-way classical communication from attacker to verifier.

\subsection{Miscellaneous remarks}

\subsubsection{On the need to re-use code keys in the trap scheme}
\label{sec:trap:re-use}

Classical keys for the trap scheme are specified by a pair $k=(\pi,P)$ indicating the choice of trap code $E_\pi$ and Pauli encryption $P$.
In Section \ref{sec:trap:universal} we saw that each gate $G$ in our universal gate set for the trap scheme has the property that the associated gadget $\tilde G$ is validated by updating \emph{only the Pauli key $P$} so that $(\pi,P)\mapsto(\pi,P')$.
This is fortunate, as our implementation of the controlled-NOT gate necessitates that every authenticated qubit share the same code key $\pi$.
As noted in Section \ref{sec:trap:properties}, the security of of any encode-encrypt scheme (including our trap scheme) is preserved even when the code key is re-used across multiple data registers.

The polynomial scheme of Ben-Or \etal has a similar structure \cite{Ben-OrC+06}.
In particular, a key in the polynomial scheme is a pair $k=(s,P)$ consisting of a ``sign key'' $s$ and a Pauli key $P$.
The generalized controlled-NOT and Toffoli gates for the polynomial scheme necessitate that each authenticated qudit use the same sign key but different Pauli keys.

\subsubsection{Clifford circuits can be implemented offline}

The $\pi/8$-phase ($T$) gate is the only gate in the universal set described in Section \ref{sec:trap:universal} that requires two-way interaction between the attacker and verifier.
Since any Clifford circuit can be implemented without $T$ gates, it follows that any Clifford circuit can be implemented on data authenticated under the trap scheme in an \emph{offline} manner.
In particular, the transmission to the verifier of all measurement results from all gadgets for the gates in a Clifford circuit can be put off until the very end of the computation, at which time the verifier can decode the results and deduce the effect of each correction on the final Pauli key.

\subsubsection{Any circuit can be implemented with only classical interaction}

We can see from the gate constructions of Section \ref{sec:trap:universal} that any quantum circuit whatsoever---be it a Clifford circuit or otherwise---can be implemented on authenticated data with only \emph{classical} interaction between attacker and verifier.

In our application to quantum one-time programs we exploit this fact in order to obtain a non-interactive protocol for two-party quantum computation by encapsulating all classical interaction inside a reactive \emph{classical} one-time program.

The need in our quantum one-time program for a \emph{reactive} classical one-time program is necessitated by the gadget for the $T$ gate.
In the special case that the quantum one-time program is for a Clifford circuit there is no need for a reactive classical one-time program: it suffices to use a non-reactive COTP to compute the final decryption key based on measurement results supplied by the user.

\subsubsection{Rigorous security}

In this section we described how gates from a universal set can be applied to data authenticated under the trap scheme.
Intuitively, we can see that the interactive protocol for implementing an arbitrary circuit is cryptographically secure, yet we did not provide a fully formal proof of security in this section.
If desired, a fully rigorous security proof for computation on authenticated data via the trap scheme can be obtained as a special case of the security proof for our main result on QOTPs in the next section.


\section{Statement and analysis of our QOTP}
\label{sec:security}

In this section we will formally describe and analyze our protocol for quantum one-time programs.  First, in Section~\ref{sec:security:preliminaries}, we describe a few preliminaries, in particular notation related to implementing quantum computation using Clifford gates and magic states.  In Sections~\ref{sec:qotp-spec} and~\ref{sec:qotp-spec-receiver} we specify the protocol actions for honest senders and receivers.  Sections~\ref{sec:receiver:arbitrary} through~\ref{sec:sim:analysis} contain the proof of Theorem~\ref{thm:main-quantum}, that our protocol is statistically quantum-UC-secure realization of $\cFunc{\Phi}{OTP}$ in the case of a corrupt user, in the $\cFunc{}{BR-OTP}$-hybrid model.  The proof makes use of a new tabular representation for operations which is given in Section~\ref{sec:sim:tabular-representation} (see also Appendix \ref{sec:tabular-representation}).  We first describe the general form of the security argument in Section~\ref{sec:receiver:arbitrary}, analyze the environment's interaction with the sender in Section~\ref{sec:user}, and describe and  analyze the simulator for an arbitrary adversary in Sections~\ref{sec:simulator} and~\ref{sec:sim:analysis}.

Our QOTP construction requires an encode-encrypt quantum authentication scheme that admits quantum computation on authenticated data,  such as the trap scheme presented in Section \ref{sec:trap}.
Our construction is completely independent of the specific choice of scheme.
Moreover, if the underlying authentication scheme is $\varepsilon$-secure then  our QOTP is $2\varepsilon$-secure.

Henceforth we assume that such a scheme is fixed---call this scheme $\bQ$.
We let $\mathscr{E}$ denote the family of codes upon which $\bQ$ is based and we let $\cG$ denote the universal gate set that can be implemented on data authenticated under~$\bQ$.

\subsection{Preliminaries}
\label{sec:security:preliminaries}

Suppose that we are given an arbitrary unitary $V$---specified as a circuit using gates from $\cG$ and acting upon register $\sR$---and we wish to construct a circuit that implements $V$ on data authenticated under $\bQ$.
We have already seen how to do this for the trap scheme.
In this section we review this simple process so as to establish basic concepts and notation that will be useful throughout Section~\ref{sec:security}.

\subsubsection{Quantum computation with Clifford gates and magic states}
\label{sec:universal}

Let $r$ denote the number of gates in $V$ that require magic states.
(For the trap scheme, $r$ is the total number of $K$, $H$, and $T$ gates in $V$.)
Alongside the data register $\sR$ we add $r$ registers $\sM_1,\dots,\sM_r$, which are assumed to be initialized to the appropriate magic states $\ket{\mu_1},\dots,\ket{\mu_r}$.
We refer to these~$r$ registers collectively as $\sM$ and to the collective state of $\sM$ as $\ket{\mu}$.

Let $V^{(0)},\dots,V^{(r)}$ be a partition of the gates of $V$ so that $V^{(i)}$ denotes the Clifford circuit consisting of all the gates occurring after measurement of the $i$th magic state and before measurement of the $(i+1)$th magic state.
Note that each $V^{(i)}$ acts only upon registers $(\sR,\sM_{i+1})$ and $V^{(r)}$ acts only upon $\sR$.
The circuit $V$ is implemented as follows:
\begin{enumerate}

\item
Apply $V^{(0)}$.

\item
For $i=1,\dots,r$:
  \begin{enumerate}

  \item
  Measure the magic state register $\sM_i$ and apply the Clifford correction $C^{(i)}$ indicated by that measurement result.

  \item
  Apply $V^{(i)}$.

  \end{enumerate}
\end{enumerate}

In order to see how the action of $V$ is recovered from the above process it is helpful to write this procedure as a channel on $\sR$ in Kraus form.
To this end let $a\in\set{0,1}^r$ be the $r$-dimensional vector of classical measurement results obtained in the above implementation of $V$ and let $V_a$ denote the Clifford circuit resulting from applying $V^{(0)},\dots,V^{(r)}$ interleaved with Clifford corrections $C^{(1)},\dots,C^{(r)}$ according to the measurement outcomes $a$.
Because each $V^{(i)}$ acts upon a unique magic state register, the desired channel on $\sR$ can be written as follows:
\begin{align}
  \label{eq:universal-G}
  \rho \mapsto \sum_{a\in\set{0,1}^r} \bra{a}V_a \Pa{\rho\ot\ket{\mu}\bra{\mu}} V_a^*\ket{a} \enspace .
\end{align}
It is a property of magic state implementations of gates that for each $a$ we have
\begin{align}
  \label{eq:universal-a}
  \bra{a}V_a \ket{\mu}=\frac{1}{2^{r/2}}V \enspace .
\end{align}
So the above channel \eqref{eq:universal-G} is equivalent to
\begin{align}
  \rho \mapsto \sum_{a\in\set{0,1}^r} \frac{1}{2^r} V\rho V^* = V\rho V^*
\end{align}
as desired.

If for some reason the measurement results $a$ are corrupted to some other vector $a'$ then the above procedure will mis-apply the Clifford corrections $C^{(1)},\dots,C^{(r)}$.
Let $V_{a-a'}$ denote the circuit derived from $V$ by inserting extra Clifford gates according to the corruption $a-a'$ so that
\begin{align}
  \label{eq:universal-corrupted}
  \bra{a'}V_a\ket{\mu}=\frac{1}{2^{r/2}}V_{a-a'} \enspace .
\end{align}

\subsubsection{Encoded circuits}
\label{sec:encoded-circuits}

In accordance with the convention introduced in Section \ref{sec:trap:defs}, any register denoted with a tilde $\tilde{\phantom{x}}$ is assumed to be accompanied by its own $X$- and $Z$-syndrome registers.
Given a logical register $\sR$ and a code $E\in\mathscr{E}$, the encoded register $\tilde\sR$ is obtained by applying the operator $E(I_\sR\ot\ket{0})$ to $\sR$.
For brevity we omit the initial state $\ket{0}$ of the syndrome registers and simply view $E$ as an isometry from $\sR$ to $\tilde\sR$ when it is convenient to do so.
When multiple registers $\sR_1,\dots,\sR_k$ are each encoded under the same code $E$ we write $E$ instead of $E^{\ot k}$.

Given a circuit $W$ acting on $\sR$ and composed entirely of states from $\cG$ that can be implemented without magic states (such as the circuits $V_a$ of the previous subsection), it is easy to construct the circuit $\tilde W$ acting on $\tilde\sR$ that implements $W$ on authenticated data: simply replace each logical gate with its equivalent on authenticated data.

Of course, encoding a register under a code $E$ and then applying $\tilde W$ is equivalent to first applying~$W$ and then encoding the result under $E$.
This identity is expressed succinctly under the above notation as
\[ \tilde W E = E W. \]

\subsection{Specification of the sender's message}
\label{sec:qotp-spec}

Let $\Phi:(\sA,\sB)\to\sC$ be a channel specified as a quantum circuit using gates from $\cG$.
In this section we specify the QOTP for $\Phi$.
More specifically, suppose that the input registers $(\sA,\sB)$ are prepared in some state (possibly entangled with other registers) and given to the sender and receiver, respectively.
In this section we show how to construct the sender's message to the receiver given $\sA$.

\subsubsection{Implementing a channel via controlled-unitary}

Without loss of generality we assume that the channel $\Phi$ has the form $\Phi:(\sA,\sB)\to\sB$.
That is, the receiver's output register $\sC=\sB$ has the same size as the input register.
Furthermore, without loss of generality we assume that the channel $\Phi$ is specified by a unitary circuit $U$ acting on registers $(\sA,\sB,\sE)$.
The extra register $\sE$ is an auxiliary register initialized to the $\ket{0}$ state.
The action of $\Phi$ is recovered from $U$ by discarding registers $(\sA,\sE)$ so that $\Phi(\rho)=\ptr{\sA\sE}{U\rho U^*}$ for all $\rho$.

Given a circuit $U$ one can efficiently find a circuit for the controlled-$U$ operation, which we denote $\cont{U}$.
In addition to the registers $(\sA,\sB,\sE)$, the circuit for $\cont{U}$ acts upon an additional qubit called the \emph{control qubit} so that for any pure state $\ket{\psi}$ of $(\sA,\sB,\sE)$ we have
\begin{align*}
\cont{U} &: \ket{\psi}\ot\ket{\mathrm{on}} \mapsto U\ket{\psi}\ot\ket{\mathrm{on}} \\
&: \ket{\psi}\ot\ket{\mathrm{off}} \mapsto \ket{\psi}\ot\ket{\mathrm{off}} \enspace .
\end{align*}
In an effort to minimize the number of distinct register names we bundle the control qubit into the ancillary register $\sE$ with the understanding that the initial state of $\sE$ is $\ket{0}$ for $U$ and $\ket{0}\ket{\mathrm{on}/\mathrm{off}}$ for~$\cont{U}$.

In our QOTP for $\Phi$ the sender and receiver shall implement the circuit for $\tildecont{U}$.
In this protocol the sender prepares the authenticated ancilla register $\tilde\sE$ (including the control qubit) and gives it to the receiver
(along with several other quantum registers and a reactive COTP to be specified shortly).
This control qubit is always initialized to the $\ket{\mathrm{on}}$ state.
As such, one might wonder why we bother implementing $\tildecont{U}$ instead of $\tilde U$.
The purpose of the control qubit is to facilitate the forthcoming security proof.

We also have an alternative QOTP in which $\tilde U$ is implemented directly with no need for $\tildecont{U}$.
However, the security proof for this alternative QOTP is more technically cumbersome than our protocol for $\tildecont{U}$ (see Section \ref{sec:sim:check-magic-only}) so we have elected to present only the protocol for $\tildecont{U}$ in this paper.
Whether or not the controlled-$U$ is necessary for our somewhat simpler security proof is an interesting unresolved question. 

\subsubsection{Specification}

Let $r$ be the number of gates in $\tildecont{U}$ that require magic states.
After the parties have received their input registers $\sA,\sB$, a non-interactive protocol for $\tildecont{U}$ consists of a single message from the sender to the receiver.
This message consists of the following objects:
\begin{enumerate}

\item
Quantum registers $\tilde\sA,\sB_\iin,\tilde\sB_\iin,\sB_\out,\tilde\sB_\out,\tilde\sE,\tilde\sM=(\tilde\sM_1,\dots,\tilde\sM_r)$ prepared in specific states described in Protocol~\ref{prot:prep-sender} below.

\item
An $(r+1)$-round reactive classical one-time program (BR-OTP) described in Protocol~\ref{prot:BR-OTP} below.

\end{enumerate}
In order to prepare this message, a code $E\in\mathscr{E}$ and encryption Paulis $P,S$ are chosen uniformly at random.
The Pauli $S$ acts on $\tilde\sB_\out$ and the Pauli $P$ acts on $(\tilde\sA,\tilde\sB_\iin,\tilde\sE,\tilde\sM)$.
(Here and throughout the paper we adopt the convention that the portion of a multi-register Pauli acting on a single register is denoted by the register name appearing in a subscript.
For example, the portion of $P$ acting on $\tilde\sM$ is denoted $P_{\tilde\sM}$ and it holds that
$P=P_{\tilde\sA}\ot P_{\tilde\sB_\iin}\ot P_{\tilde\sE}\ot P_{\tilde\sM}$.)
The registers are prepared as described in Protocol~\ref{prot:prep-sender} and Figure \ref{fig:enc-dec}.

\begin{protocol} \caption{Message preparation for sender  \label{prot:prep-sender} }
\begin{tabularx}{\textwidth}{lX}
  $(\sB_\iin,\tilde\sB_\iin)$: & Teleport-through-authentication state $P_{\tilde\sB_\iin}E\ket{\phi^+}$.  (See Figure \ref{fig:enc-dec}(a).)\\
  $(\tilde\sB_\out,\sB_\out)$: & Teleport-through-de-authentication state obtained by discarding the syndrome registers of $E^*S\ket{\phi^+}$.  (See Figure \ref{fig:enc-dec}(b).) \\
  $\tilde \sA$: & Authenticated input state.  Obtained by applying $P_{\tilde\sA}E$ to the input register $\sA$. \\
  $\tilde \sE$: & Authenticated ancilla $P_{\tilde\sE}E\ket{0}\ket{\mathrm{on}}$.\\
  $\tilde \sM$: & Authenticated magic states $P_{\tilde\sM}E\ket{\mu}$ where $\ket{\mu}=\ket{\mu_1}\cdots\ket{\mu_r}$ and $\ket{\mu_1},\dots,\ket{\mu_r}$ are the $r$ magic states required for $\cont{U}$.
\end{tabularx}
\end{protocol}

\begin{figure}
\begin{center}
\begin{tikzpicture}[xscale=0.8,yscale=0.8]
\node[right] at (-1.25, 1.7) {(a) Teleport-through-authentication};
\node[left] (phi) at (0,0.5) {$\ket{\phi^{+}}$};
\node[right] (Bin) at (5.25,1) {$\sB_\iin$};
\draw[->] (phi.east) .. controls +(0.8,0) and +(-0.8,0) .. (1,1) -- (Bin);
\draw (1,0.5)--(1,-2)--(2.5,-2)--(2.5,0.5)--(1,0.5);
\node at (1.75,-0.75) {$E$};
\node[below left] (D) at (1,0.1) {\scriptsize $\sD$};
\node[below left] (X) at (1,-0.65) {\scriptsize $\sX$};
\node[below left] (Z) at (1,-1.4) {\scriptsize $\sZ$};
\draw (phi.east) .. controls +(0.8,0) and +(-0.8,0) .. (1,0);
\node[left] (zero1) at (0,-0.75) {$\ket{0}$};
\node[left] (zero2) at (0,-1.5) {$\ket{0}$};
\draw (zero1) -- (1,-0.75);
\draw (zero2) -- (1,-1.5);
\draw (3,0.5)--(3,-2)--(4.5,-2)--(4.5,0.5)--(3,0.5);
\node at (3.75,-0.75) {$P_{\tilde\sB_{\iin}}$};
\draw (2.5,0)--(3,0);
\draw (2.5,-0.75)--(3,-0.75);
\draw (2.5,-1.5)--(3,-1.5);
\draw[->] (4.5,0)--(5.25,0);
\draw[->] (4.5,-0.75)--(5.25,-0.75);
\draw[->] (4.5,-1.5)--(5.25,-1.5);
\draw[thick,decorate,decoration={brace,amplitude=4pt}] (5.5,0.25)--(5.5,-1.75) node[midway,right=3pt]{$\tilde\sB_{\iin}$};
\node at (0,-2) {};
\end{tikzpicture}
~~~
\begin{tikzpicture}[xscale=0.8,yscale=0.8]
\node[right] at (-1.25, 1.7) {(b) Teleport-through-de-authentication};
\node[left] (phi) at (0,0.125) {$\ket{\phi^{+}}$};
\node[right] (tBout) at (6,1) {$\tilde\sB_\out$};
\draw[->] (phi.east) .. controls +(0.8,0) and +(-0.8,0) .. (1,1) -- (tBout);
\draw (1,0.5)--(1,-2)--(2.5,-2)--(2.5,0.5)--(1,0.5);
\node at (1.75,-0.75) {$E^{*}$};
\draw (phi.east) .. controls +(0.8,0) and +(-0.8,0) .. (1,-0.75);
\draw (3,0.5)--(3,-2)--(4.5,-2)--(4.5,0.5)--(3,0.5);
\node at (3.75,-0.75) {$S$};
\draw (2.5,0)--(3,0);
\draw (2.5,-0.75)--(3,-0.75);
\draw (2.5,-1.5)--(3,-1.5);
\node[right] (Bout) at (6,0) {$\sB_{\out}$};
\draw[->] (4.5,0)--(Bout);
\draw[->] (4.5,-0.75)--(5.25,-0.75)--(5.25,-2);
\draw[-] (4.5,-1.5)--(5.25,-1.5);
\draw[thick,decorate,decoration={brace,amplitude=4pt}] (5.5,-0.75)--(5.5,-2);
\node[right] at (5.75,-0.8) {syndrome};
\node[right] at (5.75,-1.3) {registers};
\node[right] at (5.75,-1.8) {discarded};
\node at (0,-2) {};
\end{tikzpicture}
\end{center}
\vspace{-1em}
\caption{Circuits for teleporting through authentication and de-authentication}
\label{fig:enc-dec}
\end{figure}
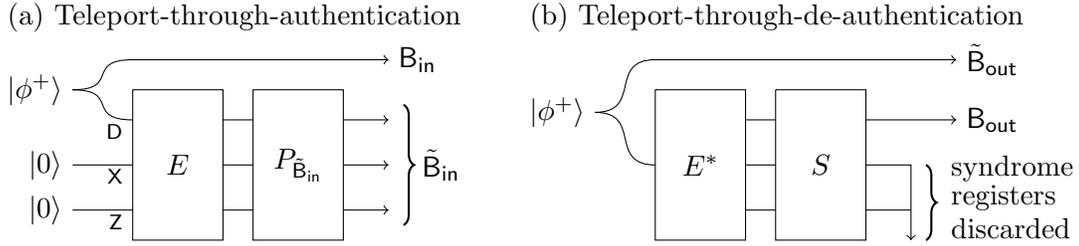

In addition to these registers, the sender prepares an $(r+1)$-round BR-OTP to act as described in Protocol~\ref{prot:BR-OTP}.

\begin{protocol} \caption{Specification of the BR-OTP \label{prot:BR-OTP} }
\begin{enumerate}

\item
Receive (a classical description of) a Pauli $T^\iin$. 

\item
For $i=1,\dots,r$:
  \begin{enumerate}

  \item \label{it:cotp-measure-decode}
  Receive a classical bit string $c_i$. 
  Decode $c_i$ into a classical bit $a_i$ as per the measure-then-decode circuit \eqref{circuit:measure-decode} of Section \ref{sec:auth:measure-decode} with code key $E$ and decryption Pauli $\hat P_{\tilde\sM_i}$ to be specified later.
  Return the decoded bit $a_i$ to the user.

  \item \label{it:cotp-cheating}
  If $c_i$ is inconsistent with $(E,\hat P_{\tilde\sM_i})$---that is, if the measure-then-decode circuit \eqref{circuit:measure-decode} applied to $c_i$ indicates a non-zero error syndrome---then cheating has been detected.
  Remember whether cheating has been detected.

  \end{enumerate}

\item \label{it:cotp-decrypt}
Receive (a classical description of) a Pauli $T^\out$.
If cheating was never detected in step \ref{it:cotp-cheating} then return to the user a (classical description of a) decryption Pauli $\hat S$ to be specified later.
Otherwise return uniformly random bits.
For later convenience we specify that the classical description of $\hat S$ be contained in a register $\sK$.

\end{enumerate}
\end{protocol}

\lightparagraph{This QOTP could be mass-produced}
The state of the authenticated register $\tilde\sA$ depends upon the state of the sender's input register $\sA$.
But the remaining registers could all be prepared (or mass-produced) before $\sA$ is received.
Furthermore, the BR-OTP also does not depend upon $\sA$, but it does depend upon the authentication key for $\tilde\sA$.
This key could be chosen in advance, in which case the BR-OTP could also be mass-produced before $\sA$ is received.

\subsubsection{No need to check integrity of the data registers}
\label{sec:sim:check-magic-only}

Notice that the BR-OTP checks the syndrome registers of \emph{only the magic state register $\tilde\sM$} and ignores any error syndrome present in other registers.
This minimal requirement might at first seem insufficient to establish a secure QOTP as, for example, the receiver is free to tamper with the data registers $(\tilde\sA,\tilde\sB,\tilde\sE)$.
We shall see later that any such attack is equivalent to an attack before and/or after $\Phi$ is applied and hence can be reproduced by a simulator with one-shot access to $\Phi$.

If instead we were to modify our BR-OTP so as to also verify the syndrome registers for the receiver's output register $\sB_\out$ then we could derive a protocol that implements the circuit $U$ directly, as opposed to the controlled-$U$ implemented by the present protocol.
As suggested earlier, however, the formal proof of security for such a protocol is more technically cumbersome than our current proof.

As noted in Section \ref{sec:auth:measure-decode}, the measure-then-decode circuit \eqref{circuit:measure-decode} implemented by the COTP in step \ref{it:cotp-measure-decode} is equivalent to the decode-then-measure circuit \eqref{circuit:decode-measure}.
Thus, we may equivalently assume that in step \ref{it:cotp-measure-decode} the COTP applies the \emph{quantum} circuit $\hat P_{\tilde\sM_i} E^*$ to register $\tilde\sM_i$ followed by measurements of the $X$-syndrome register and the data register.

\subsection{Protocol for an honest receiver, completeness}
\label{sec:qotp-spec-receiver}

The actions for an honest receiver to use the QOTP to obtain $\Phi(\rho)$ are specified in Protocol~\ref{prot:honest-QOTP}.

\begin{protocol}\caption{Honest use of QOTP} \label{prot:honest-QOTP}
\begin{enumerate}

\item
Perform a Bell measurement on $(\sB,\sB_\iin)$ so as to teleport-through-authentication.
This is achieved by first applying the unitary Bell rotation $B$ followed by a measurement $\set{\ket{T^\iin}\bra{T^\iin}}$ where each $\ket{T^\iin}\bra{T^\iin}$ is a projector onto the classical basis state indicating Pauli correction $T^\iin$.
The receiver provides $T^\iin$ as the first input to the BR-OTP.

At this time the contents of $\sB$ have been authenticated and placed in register $\tilde\sB_\iin$.

\item
Run the protocol of Section \ref{sec:universal} with $V=\tildecont{U}$ so as to apply $\tildecont{U}$ to the authenticated registers $(\tilde\sA,\tilde\sB_\iin,\tilde\sE,\tilde\sM)$.
Explicitly,
  \begin{enumerate}

  \item
  Apply $\tilde{\cont{U^{(0)}}}$.

  \item
  For $i=1,\dots,r$:
    \begin{enumerate}

    \item
    Measure the magic state register $\tilde\sM_i$ and provide the result as input to the BR-OTP.

    \item
    The BR-OTP provides as output a single bit indicating whether to apply the associated Clifford correction $\tilde C^{(i)}$.

    \item
    Apply $\tilde{\cont{U^{(i)}}}$.

    \end{enumerate}
  \end{enumerate}
The implementation of $\tildecont{U}$ is now complete.
At this time the register $(\tilde\sA,\tilde\sB_\iin,\tilde\sE)$ holds the authenticated version of $(\sA,\sB,\sE)$ with $\cont{U}$ applied.

\item
Perform a Bell measurement on $(\tilde\sB_\iin,\tilde\sB_\out)$ so as to teleport-through-de-authentication.
As above, the result of this measurement indicates a Pauli correction $T^\out$, which the receiver provides as the final input to the BR-OTP.

At this time the register $\sB_\out$ holds the receiver's output.
This register is encrypted but not authenticated.

\item
For its final output, the BR-OTP provides the Pauli decryption key $\hat S$ to be specified later.
Apply this Pauli to $\sB_\out$ to recover the receiver's output.

\end{enumerate}
\end{protocol}

\subsection{General form of an arbitrary environment}
\label{sec:receiver:arbitrary}

UC security of the protocol of Section \ref{sec:qotp-spec} is proved as follows.
First, consider an arbitrary (possibly cheating) receiver who receives the message from the sender and interacts with the sender's BR-OTP.
Throughout the interaction the receiver also exchanges messages with an \emph{environment}.
UC security is established by exhibiting a simulator that mimics the behaviour of the receiver from the environment's point of view using only calls to the ideal functionality.
Specifically, the environment selects the input registers $(\sA,\sB)$ for both sender and receiver, interacts with either the real receiver or the simulated receiver, then produces a single bit indicating the environment's guess as to whether it interacted with the real or simulated receiver.

By the completeness of dummy adversaries~\cite{Can01,U10}, it suffices to assume that the receiver simply shuttles messages between the environment and the honest sender.
For convenience this dummy receiver can be absorbed into the environment, leaving only an interaction between the environment and the honest sender (or simulated honest sender).
In this section we write down a general form that every such environment must have.

Given the discussion above, we assume without loss of generality that the actions of any environment throughout the interaction 
are as described in Protocol \ref{prot:environment}.

\begin{protocol}\caption{General form of an arbitrary environment} \label{prot:environment}
\begin{enumerate}

\item
Prepare registers $(\sA,\sB,\sW)$ in an arbitrary pure state $\ket{\psi}$ and provide $\sA$ to the sender (or simulated sender) as input.

\item  \label{it:receiver:enc}

Receive from the sender (or simulated sender) quantum registers
$(\tilde\sA,\sB_\iin,\tilde\sB_\iin,\sB_\out,\tilde\sB_\out,\tilde\sE,\tilde\sM)$.

Apply a unitary $K^{(0)}$ to all the registers, then perform a Bell measurement on $(\sB,\sB_\iin)$ so as to teleport-through-authentication.
Provide the resulting Pauli $T^\iin$ as the first input to the BR-OTP, just as an honest receiver would.

\item \label{it:receiver:loop}
For $i=1,\dots,r$:
  \begin{enumerate}

  \item \label{it:receiver:loop:measure}
  Provide the register $\tilde\sM_i$ to the BR-OTP.

  \item \label{it:receiver:loop:decode}
  The BR-OTP returns a single bit $a_i$.

  \item \label{it:receiver:loop:K}
  Apply a unitary $K^{(i)}_{a_i}$ to the remaining registers.
  (That is, every register except $\sB,\sB_\iin,\tilde\sM_1,\dots,\tilde\sM_i$.)
  The subscript $a_i$ indicates that the environment's choice of $K^{(i)}_{a_i}$ could depend upon the bit $a_i$.

  \end{enumerate}

\item \label{it:receiver:dec}
Perform a Bell measurement on $(\tilde\sB_\iin,\tilde\sB_\out)$ so as to teleport-through-de-authentication.
Provide the resulting Pauli $T^\out$ as the final input to the BR-OTP, just as an honest receiver would.

\item \label{it:receiver:dump}
The BR-OTP provides as its final output (a classical description of) a decryption Pauli $\hat S$ to be specified later, stored in a new register $\sK$.

\item \label{it:receiver:final-measurement}
Perform a binary-valued measurement on the remaining registers $\tilde\sA,\sB_\out,\sE,\sW,\sK$.

\end{enumerate}
\end{protocol}

Let us now argue that no generality is lost in assuming an environment described in Protocol \ref{prot:environment}.
\begin{enumerate}

\item

The BR-OTP has by definition classical input/output behaviour that does not preserve coherent superpositions of classical basis states.
As such, any qubits touched by the BR-OTP are assumed to be measured in the computational basis.
Thus, it makes no difference whether data sent to a BR-OTP is measured by the environment or by the BR-OTP, so we are free to assume that any given measurement is performed by whichever party better suits our discussion.

\item

In steps \ref{it:receiver:enc} and \ref{it:receiver:dec} it is assumed that the environment does a proper Bell measurement of $(\sB,\sB_\iin)$ and $(\tilde\sB_\iin,\tilde\sB_\out)$, respectively, and then faithfully reports the results to the BR-OTP just as an honest receiver would.
This assumption is justified because any tampering the environment might wish to do with the measurements or the results thereof can instead be incorporated into the environment's circuits $K^{(i)}_{a_i}$.

For example, an environment who wishes to tamper with the Bell measurement in step \ref{it:receiver:enc} could select the circuit $K^{(0)}$ so as to prepare the registers $(\sB,\sB_\iin)$ in any desired state by swapping out data from the memory workspace $\sW$ (or even other registers from the sender) and then pre-invert the rotation into the Bell basis prior to measurement.
Substituting this choice of $K^{(0)}$ into Protocol~\ref{prot:environment}, we see that such an environment undoes the Bell rotation and thus the resulting measurement of $(\sB,\sB_\iin)$ reproduces exactly the result intended by the environment.

\item

In step \ref{it:receiver:loop:measure} it is assumed that the environment simply transmits the unmeasured register $\tilde\sM_i$ to the BR-OTP.  

As explained above, we are free to make this assumption as the BR-OTP has purely classical input/output behaviour.
This assumption is convenient because it allows us to seamlessly substitute the measure-then-decode procedure with its equivalent decode-then-measure procedure acting on unmeasured quantum data as described in Section \ref{sec:trap:properties}.

\item

Intuitively, the environment's $i$th circuit $K^{(i)}_{a_i}$ could depend upon all prior classical data exchanged with the BR-OTP.
Yet in step \ref{it:receiver:loop:K} we explicitly allow only for dependence upon the most recent bit $a_i$ received from the BR-OTP.
Furthermore, we make the simplifying assumption that each $K^{(i)}_{a_i}$ does not act upon the registers $\sB,\sB_\iin,\tilde\sM_1,\dots,\tilde\sM_{i-1}$ that were measured earlier in the protocol.
No generality is lost, however, because prior circuits are free to copy classical data into fresh space within the workspace register $\sW$ for reference in future rounds.

\end{enumerate}

\subsection{A tabular representation for operators and vectors}
\label{sec:sim:tabular-representation}

Security of our QOTP is established by analysis of the state of all the registers in the environment's posession immediately prior to the final binary-valued measurement in step \ref{it:receiver:final-measurement} of Protocol \ref{prot:environment}.
However, it is difficult to produce a concise description of this state due to the large number of distinct registers and the many rounds of interaction with the BR-OTP.
To combat this difficulty we introduce a tabular representation for operators and vectors.

This new way of describing states is best explained by example.
To this end recall that the initial state of registers $(\sA,\sB,\sW)$ is $\ket{\psi}$ and that register $\sA$ is given to the sender as input.
The sender introduces registers $\sB_\iin,\tilde\sB_\iin,\tilde\sB_\out,\sB_\out,\tilde\sE,\tilde\sM$.
For each choice of code key $E$, encryption keys $P,S$, Bell measurement results $T^\iin,T^\out$, and magic state data- and syndrome-measurement results $a=(a_1,\dots,a_r),s=(s_1,\dots,s_r)$, the unnormalized pure state of the remaining (unmeasured) quantum registers $(\sB_\out,\tilde\sA,\sW,\tilde\sE)$ after step \ref{it:receiver:dec} of Protocol \ref{prot:environment} (that is, immediately before the BR-OTP reveals the decryption key register $\sK$) is given by the table in Figure \ref{fig:tabular-state}.

\begin{figure}
\begin{align*}
  \renewcommand{\arraystretch}{1.3}
  \begin{array}{|c||c|c|c|c|c|c|c|c|}
    \cline{1-1} \cline{6-9}
  \sB & \multicolumn{4}{c|}{} & \multirow{2}{*}{$\bra{T^\iin}B$} & \multirow{12}{*}{$K^{(0)}$} && \ket{\psi}_\sB \\ 
    \cline{1-1} \cline{8-9}
  \sB_\iin & \multicolumn{4}{c|}{} &&&& \multirow{2}{*}{$\ket{\phi^+}$} \\ 
    \cline{1-6} \cline{8-8} 
  \tilde\sB_\iin & \multirow{2}{*}{$\bra{T^\out}B$} & \multirow{6}{*}{$K^{(r)}_{a_r}$} & \multirow{7}{*}{$\dots$} & \multirow{8}{*}{$K^{(2)}_{a_2}$} & \multirow{9}{*}{$K^{(1)}_{a_1}$} && P_{\tilde\sB_\iin} E & \\ 
    \cline{1-1} \cline{8-9}
  \tilde\sB_\out &&&&&&&& \multirow{2}{*}{$\ket{\phi^+}$} \\ 
    \cline{1-2} \cline{8-8} 
  \sB_\out &&&&&&& E^* S & \\ 
    \cline{1-1} \cline{8-9}
  \tilde\sA &&&&&&& P_{\tilde\sA} E & \ket{\psi}_\sA \\ 
    \cline{1-1} \cline{8-9}
  \sW &&&&&&&& \ket{\psi}_\sW \\ 
    \cline{1-1} \cline{8-9}
  \tilde\sE &&&&&&& P_{\tilde\sE} E & \ket{0}\ket{\mathrm{on}} \\ 
    \cline{1-1} \cline{3-3} \cline{8-9}
  \tilde\sM_r && \bra{a_r}\bra{s_r} E^* \hat P_{\tilde\sM_r} &&&&& P_{\tilde\sM_r} E & \ket{\mu_r} \\ 
    \cline{1-1} \cline{3-4} \cline{8-9}
  \vdots & \multicolumn{2}{c|}{} & \ddots &&&& \vdots & \vdots \\ 
    \cline{1-1} \cline{4-5} \cline{8-9}
  \tilde\sM_2 & \multicolumn{3}{c|}{} & \bra{a_2}\bra{s_2} E^* \hat P_{\tilde\sM_2} &&& P_{\tilde\sM_2} E & \ket{\mu_2} \\ 
    \cline{1-1} \cline{5-6} \cline{8-9}
  \tilde\sM_1 & \multicolumn{4}{c|}{} & \bra{a_1}\bra{s_1} E^* \hat P_{\tilde\sM_1} && P_{\tilde\sM_1} E & \ket{\mu_1} \\ 
    \cline{1-1} \cline{6-9}
  \end{array}
\end{align*}
\caption{Tabular representation of the state of the system after step \ref{it:receiver:dec} of Protocol \ref{prot:environment}.}
\label{fig:tabular-state}
\end{figure}

In general, rows in a table correspond to registers in the joint system.
Cells within each row correspond to individual operators or vectors and are ordered from right to left as per the convention for operator composition.
Empty cells implicitly indicate the identity operator.
The table as a whole specifies an operator (or vector) as a composition of product operators on the registers.

In this particular table we have used the notation $\ket{\psi}_\sA,\ket{\psi}_\sB,\ket{\psi}_\sW$ to refer to the portions of $\ket{\psi}$ contained in registers $\sA,\sB,\sW$, respectively.

\subsection{Analysis of the environment's interaction with the real sender}
\label{sec:user}

At the end of its interaction the environment produces a single output bit.
In order to prove UC security we must exhibit a simulator that interacts with the environment such that the distribution on the environment's output bit is nearly identical to that induced by the environment's interaction with the real sender and BR-OTP.

We accomplish this task by deriving an expression for the state $\rho_\mathrm{real}$ of the registers $(\sB_\out,\tilde\sA,\tilde\sE,\sW,\sK)$ in the environment's possession at the end of step \ref{it:receiver:dump} of Protocol \ref{prot:environment}.
In Section \ref{sec:simulator} we will exhibit a simulator, then in Section \ref{sec:sim:analysis} we derive an expression for the associated state $\rho_\mathrm{sim}$ at the end of the simulated interaction and show that the trace distance between $\rho_\mathrm{real}$ and $\rho_\mathrm{sim}$ is negligible, from which the security of our QOTP follows.

Recall that in step \ref{it:receiver:dump} the register $\sK$ is revealed by the BR-OTP.
Thus, if the state of the system at the end of step \ref{it:receiver:dump} is $\rho_\mathrm{real}$ then state of the system at the end of step \ref{it:receiver:dec} must be $\ptr{\sK}{\rho_\mathrm{real}}$.
We  begin our analysis by focussing on the state $\ptr{\sK}{\rho_\mathrm{real}}$.

Let $\abs{\mathscr{P}},\abs{\mathscr{S}}$ denote the number of Pauli operations acting on $(\tilde\sB_\iin,\tilde\sM,\tilde\sP),\tilde\sB_\out$, respectively, so that each encryption Pauli $P,S$ is chosen with probability $1/\abs{\mathscr{P}},1/\abs{\mathscr{S}}$, respectively.
The normalized mixed state $\ptr{\sK}{\rho_\mathrm{real}}$ of $(\sB_\out,\tilde\sA,\sW,\tilde\sE)$ at step \ref{it:receiver:dec} is
\begin{align}
\label{eq:user-channel}
  \ptr{\sK}{\rho_\mathrm{real}} = \frac{1}{\abs{\mathscr{E}}\abs{\mathscr{P}}\abs{\mathscr{S}}}
    \sum_{E,P,S,a,s,T^\iin,T^\out}
    \ket{\eqref{fig:tabular-state}}\bra{\eqref{fig:tabular-state}}
\end{align}
where $\ket{\eqref{fig:tabular-state}}$ denotes the unnormalized pure state depicted in Figure \ref{fig:tabular-state}.

This state can be written much more succinctly.
Let $a=(a_1,\dots,a_r)$ and $s=(s_1,\dots,s_r)$ denote the vectors of magic-state data and syndrome values and write
\begin{align}
  K_a &\defeq K^{(r)}_{a_r} \cdots K^{(1)}_{a_1} K^{(0)} \\
  \bra{a} &\defeq \bra{a_r}\cdots\bra{a_1} \\
  \bra{s} &\defeq \bra{s_r}\cdots\bra{s_1}.
\end{align}
The rows of table \eqref{fig:tabular-state} corresponding to the magic state registers $\tilde\sM_1,\dots,\tilde\sM_r$ can be amalgamated into a single row $\tilde\sM$.
The table \eqref{fig:tabular-state} can be written
\begin{align}
  \label{tab:easyK}
  \ket{\eqref{fig:tabular-state}} =
  \renewcommand{\arraystretch}{1.3}
  \begin{array}{|c||c|c|c|c|}
    \hline
  \sB & \multirow{2}{*}{$\bra{T^\iin}B$} & \multirow{9}{*}{$K_a$} && \ket{\psi}_\sB \\
    \cline{1-1} \cline{4-5}
  \sB_\iin &&&& \multirow{2}{*}{$\ket{\phi^+}$} \\
    \cline{1-2} \cline{4-4}
  \tilde\sB_\iin & \multirow{2}{*}{$\bra{T^\out}B$} && P_{\tilde\sB_\iin} E & \\
    \cline{1-1} \cline{4-5}
  \tilde\sB_\out &&&& \multirow{2}{*}{$\ket{\phi^+}$} \\
    \cline{1-2} \cline{4-4}
  \sB_\out &&& E^* S &  \\
    \cline{1-2} \cline{4-5}
  \tilde\sA &&& P_{\tilde\sA} E & \ket{\psi}_\sA \\
    \cline{1-2} \cline{4-5}
  \sW &&&& \ket{\psi}_\sW \\
    \cline{1-2} \cline{4-5}
  \tilde\sE &&& P_{\tilde\sE} E & \ket{0}\ket{\mathrm{on}} \\
    \cline{1-2} \cline{4-5}
  \tilde\sM & \bra{a}\bra{s} E^* \hat P_{\tilde\sM} && P_{\tilde\sM} E & \ket{\mu} \\
    \hline
  \end{array}
\end{align}
Insert a superfluous $I=(\tildeconta{U})^*(\tildeconta{U})$ acting on registers $(\tilde \sB_\iin,\tilde\sA,\tilde\sE,\tilde\sM)$ into the table \eqref{tab:easyK} immediately prior to $K_a$ and reorder the rows to get
\begin{align}
  \label{tab:superfluous}
  \ket{\eqref{fig:tabular-state}} =
  \renewcommand{\arraystretch}{1.3}
  \begin{array}{|c||c|c|c|c|c|c|}
    \hline
  \sW &&&&&& \ket{\psi}_\sW \\
    \cline{1-2} \cline{4-7}
  \sB & \multirow{2}{*}{$\bra{T^\iin}B$} & \multirow{9}{*}{$K_a$} &&&& \ket{\psi}_\sB \\
    \cline{1-1} \cline{4-7}
  \sB_\iin &&&&&& \multirow{2}{*}{$\ket{\phi^+}$} \\
    \cline{1-2} \cline{4-6}
  \tilde\sB_\iin & \multirow{2}{*}{$\bra{T^\out}B$} && (\tildeconta{U})^*_{\tilde\sB_\iin} & (\tildeconta{U})_{\tilde\sB_\iin} & P_{\tilde\sB_\iin} E & \\
    \cline{1-1} \cline{4-7}
  \tilde\sB_\out &&&&&& \multirow{2}{*}{$\ket{\phi^+}$} \\
    \cline{1-2} \cline{4-6}
  \sB_\out &&&&& E^* S &  \\
    \cline{1-2} \cline{4-7}
  \tilde\sA &&& (\tildeconta{U})^*_{\tilde\sA} & (\tildeconta{U})_{\tilde\sA} & P_{\tilde\sA} E & \ket{\psi}_\sA \\
    \cline{1-2} \cline{4-7}
  \tilde\sE &&& (\tildeconta{U})^*_{\tilde\sE} & (\tildeconta{U})_{\tilde\sE} & P_{\tilde\sE} E & \ket{0}\ket{\mathrm{on}} \\
    \cline{1-2} \cline{4-7}
  \tilde\sM & \bra{a}\bra{s} E^* \hat P_{\tilde\sM} && (\tildeconta{U})^*_{\tilde\sM} & (\tildeconta{U})_{\tilde\sM} & P_{\tilde\sM} E & \ket{\mu} \\
    \hline
  \end{array}
\end{align}
(Here we have used the notation $(\tildeconta{U})^*_{\tilde\sA}$, \emph{etc.}\ to illustrate the fact that the unitary $(\tildeconta{U})^*$ is split over several nonadjacent rows in the table.  If the rows $\tilde\sA,\tilde\sB_\iin,\tilde\sE,\tilde\sM$ were all adjacent then this notation would be unnecessary; all four of those cells could be merged into a single larger cell for the operator $(\tildeconta{U})^*$.
Unfortunately, it is not possible to order the rows so as to eliminate the need to split some operators or states across nonadjacent rows.)

At an intuitive level, if the environment is to avoid being detected as a cheater by the BR-OTP then whatever the environment does in $K_a$ must cancel out $(\tildeconta{U})^*$ so that $\tildeconta{U}$ is implemented on $(\tilde \sB_\iin,\tilde\sA,\tilde\sE,\tilde\sM)$.

Let
\begin{align}
  \label{eq:pauli-decomps}
  K_a (\tildeconta{U})^* &= \sum_{\textrm{Paulis $Q$}} \alpha_Q Q
\end{align}
be a Pauli decomposition of $K_a (\tildeconta{U})^*$.
Substitute this decomposition into \eqref{tab:superfluous} to obtain
\begin{align}
  \ket{\eqref{fig:tabular-state}} = \sum_{\textrm{Paulis $Q$}} \alpha_Q \ket{\eqref{tab:kraus-pauli}}
\end{align}
where the table \eqref{tab:kraus-pauli} is given by
\begin{align}
  \label{tab:kraus-pauli}
  \renewcommand{\arraystretch}{1.3}
  \begin{array}{|c||c|c|c|c|c|c|c|c|c|}
    \hline
  \sW && Q_{\tilde\sW} &&& \ket{\psi}_\sW \\
    \hline
  \sB & \multirow{2}{*}{$\bra{T^\iin}B$} & Q_\sB &&& \ket{\psi}_\sB \\
    \cline{1-1} \cline{3-6}
  \sB_\iin && Q_{\sB_\iin} &&& \multirow{2}{*}{$\ket{\phi^+}$} \\
    \cline{1-2} \cline{3-5}
  \tilde\sB_\iin & \multirow{2}{*}{$\bra{T^\out}B$} & Q_{\tilde\sB_\iin} & (\tildeconta{U})_{\tilde\sB_\iin} & P_{\tilde\sB_\iin} E & \\
    \cline{1-1} \cline{3-6}
  \tilde\sB_\out && Q_{\tilde\sB_\out} &&& \multirow{2}{*}{$\ket{\phi^+}$} \\
    \cline{1-2} \cline{3-5}
  \sB_\out && Q_{\sB_\out} && E^* S &  \\
    \cline{1-2} \cline{3-6}
  \tilde\sA && Q_{\tilde\sA} & (\tildeconta{U})_{\tilde\sA} & P_{\tilde\sA} E & \ket{\psi}_\sA \\
    \cline{1-2} \cline{3-6}
  \tilde\sE && Q_{\tilde\sE} & (\tildeconta{U})_{\tilde\sE} & P_{\tilde\sE} E & \ket{0}\ket{\mathrm{on}} \\
    \cline{1-2} \cline{3-6}
  \tilde\sM & \bra{a}\bra{s} E^* \hat P_{\tilde\sM} & Q_{\tilde\sM} & (\tildeconta{U})_{\tilde\sM} & P_{\tilde\sM} E & \ket{\mu} \\
    \hline
  \end{array}
\end{align}
Let us consider the teleportation operations.
According to Appendix \ref{appendix:teleportation}, the rows for $\sB$ and $\sB_\iin$ can be folded into $\tilde\sB_\iin$ to obtain
\begin{align}
  \ket{\eqref{tab:kraus-pauli}} =
  \abs{\sB}^{-1/2}
  \renewcommand{\arraystretch}{1.3}
  \begin{array}{|c||c|c|c||c|}
    \hline
  \tilde\sB_\iin & \multirow{2}{*}{$\bra{T^\out}B$} & Q_{\tilde\sB_\iin} (\tildeconta{U})_{\tilde\sB_\iin} P_{\tilde\sB_\iin} E Q_{\sB_\iin}^\trans T^\iin Q_\sB & & \ket{\psi}_\sB \\
    \cline{1-1} \cline{3-5}
  \tilde\sB_\out && Q_{\tilde\sB_\out} && \multirow{2}{*}{$\ket{\phi^+}$} \\
    \cline{1-4}
  \sB_\out && Q_{\sB_\out} & E^* S & \\
    \hline
  \multicolumn{5}{|c|}{\textrm{Registers $\tilde\sA,\tilde\sE,\tilde\sM,\sW$ same as \eqref{tab:kraus-pauli}.}} \\ \hline
  \end{array}
\end{align}
(Here $\abs{\sR}$ denotes the dimension of the space associated with register $\sR$.)
We can then fold $\tilde\sB_\iin,\tilde\sB_\out$ into $\sB_\out$ so as to obtain
\begin{align}
  \label{tab:kraus-pauli-teleportation}
  \ket{\eqref{tab:kraus-pauli}} =
  \abs{\sB}^{-1/2}\abs{\tilde\sB_\iin}^{-1/2}
  \renewcommand{\arraystretch}{1.3}
  \begin{array}{|c||c|c|} \hline
  \sB_\out & Q_{\sB_\out} E^* S Q_{\tilde\sB_\out}^\trans T^\out Q_{\tilde\sB_\iin} (\tildeconta{U})_{\tilde\sB_\iin} P_{\tilde\sB_\iin} E Q_{\sB_\iin}^\trans T^\iin Q_\sB & \ket{\psi}_\sB \\ \hline
  \multicolumn{3}{|c|}{\textrm{Registers $\tilde\sA,\tilde\sE,\tilde\sM,\sW$ same as \eqref{tab:kraus-pauli}.}} \\ \hline
  \end{array}
\end{align}
Notice the normalization factors $\abs{\sB}^{-1/2},\abs{\tilde\sB_\iin}^{-1/2}$ picked up in the process---we will include those later.
Let us re-write the entire table, amalgamating $\tildeconta{U}$, $P$, and $\ket{\psi}$.
\begin{align}
  \label{tab:collect-Q}
  \ket{\eqref{tab:kraus-pauli}} =
  \renewcommand{\arraystretch}{1.3}
  \begin{array}{|c||c|c|c|c|c|c|c|}
    \hline
  \sW &&  Q_\sW &&&&& \multirow{3}{*}{$\ket{\psi}$} \\
    \cline{1-7}
  \sB_\out && Q_{\sB_\out} E^* S Q_{\tilde\sB_\out}^\trans T^\out Q_{\tilde\sB_\iin} & \multirow{4}{*}{$\tildeconta{U}$} & \multirow{4}{*}{$P$} & E & Q_{\sB_\iin}^\trans T^\iin Q_\sB & \\
    \cline{1-3} \cline{6-7}
  \tilde \sA && Q_{\tilde\sA} &&& E && \\
    \cline{1-3} \cline{6-8}
  \tilde \sE && Q_{\tilde\sE} &&& E && \ket{0}\ket{\mathrm{on}} \\
    \cline{1-3} \cline{6-8}
  \tilde \sM & \bra{a}\bra{s} E^* \hat P_{\tilde\sM} & Q_{\tilde\sM} &&& E && \ket{\mu} \\
    \hline
  \end{array}
\end{align}
The next step is to use commutation relations so as to massage this table into a nice form for the Pauli sandwich (Appendix \ref{appendix:encode-encrypt}).

To this end let us first dispense with the annoying transposition operations appearing on $Q_{\tilde\sB_\out}^\trans,Q_{\sB_\iin}^\trans$.
Note that for any multi-qubit Pauli operator $P$ the transpose $P^\trans=\pm P$ with the negative phase occurring whenever $P$ is a product of an odd number of $Y$-Paulis.
Let $y(P)\in\set{\pm 1}$ denote this phase so that $P^\trans = y(P)P$ for all Paulis $P$.
We apply this identity to the above table, producing an extra phase factor $y(Q_{\tilde\sB_\out}) y(Q_{\sB_\iin})$.

We require some notation before proceeding further.
For any Paulis $P,Q$ we write
\begin{align}
  c(P,Q) \defeq
    \left\{
    \begin{array}{ll}
      +1 & \textrm{if $PQ=QP$} \\
      -1 & \textrm{if $PQ=-QP$}
    \end{array}
    \right.
\end{align}
for the phase picked up by commuting $P,Q$.
By analogy to the notation introduced in Section \ref{sec:trap:defs}, for any code $E$ and any Pauli $P$ we let $P_E$ denote the Pauli with $P_EE=EP$.
Recall that for any vector $a$ of measurement results the encoded circuit $\tildeconta{U}$ is a Clifford circuit.
As such, for any Pauli $P$ we let $\pi_a(P)$ denote the Pauli with
\begin{align}
  \pi_a(P)(\tildecont{U}) = (\tildeconta{U}) P \enspace .
\end{align}
Finally, recall that $(\tildeconta{U}) E = E (\cont{U_a})$ as noted in Section \ref{sec:encoded-circuits}.

Now, concentrate only on the rightmost four columns of the table \eqref{tab:collect-Q}:
\begin{enumerate}
\item
  $Q_{\sB_\iin}$ commutes with $T^\iin$, picking up a phase of $c(T^\iin,Q_{\sB_\iin})$.
\begin{align}
  \renewcommand{\arraystretch}{1.3}
  \begin{array}{|c|c|c|c|} \hline
   \multirow{4}{*}{$\tildeconta{U}$} & \multirow{4}{*}{$P$} & E & T^\iin Q_{\sB_\iin} Q_{\sB} \\
     \cline{3-4}
   & & E & \\ \cline{3-4}
   & & E & \\ \cline{3-4}
   & & E & \\ \hline
  \end{array}
\end{align}
\item
  Replace $ET^\iin$ with $T^\iin_E E$.
\begin{align}
  \renewcommand{\arraystretch}{1.3}
  \begin{array}{|c|c|c|c|c|} \hline
  \multirow{4}{*}{$\tildeconta{U}$} & \multirow{4}{*}{$P$} &  T^\iin_E & E & Q_{\sB_\iin} Q_{\sB} \\
    \cline{3-5}
  &&& E & \\ \cline{3-5}
  &&& E & \\ \cline{3-5}
  &&& E & \\ \hline
  \end{array}
\end{align}
\item
Since $\tildeconta{U}$ is a Clifford, we can pull $PT^\iin_E$ through $\tildeconta{U}$.
\begin{align}
  \renewcommand{\arraystretch}{1.3}
  \begin{array}{|c|c|c|c|} \hline
  \multirow{4}{*}{$\pi_a(PT^\iin_E)$} & \multirow{4}{*}{$\tildeconta{U}$} & E & Q_{\sB_\iin} Q_{\sB} \\
    \cline{3-4}
  && E & \\ \cline{3-4}
  && E & \\ \cline{3-4}
  && E & \\ \hline
  \end{array}
\end{align}
\item
Then we can pull $E$ through $\tildeconta{U}$.
\begin{align}
  \renewcommand{\arraystretch}{1.3}
  \begin{array}{|c|c|c|c|} \hline
  \multirow{4}{*}{$\pi_a(PT^\iin_E)$} & E & \multirow{4}{*}{$\cont{U_a}$} & Q_{\sB_\iin} Q_{\sB} \\
    \cline{2-2}\cline{4-4}
  & E && \\ \cline{2-2}\cline{4-4}
  & E && \\ \cline{2-2}\cline{4-4}
  & E && \\ \hline
  \end{array}
\end{align}
\end{enumerate}
The entire table is now
\begin{align}
  \label{tab:pre-sandwich}
  \ket{\eqref{tab:kraus-pauli}} =
  \renewcommand{\arraystretch}{1.3}
  \begin{array}{|c||c|c|c|c|c|c|c|}
    \hline
  \sW &&  Q_\sW &&&&& \multirow{3}{*}{$\ket{\psi}$} \\
    \cline{1-7}
  \sB_\out && Q_{\sB_\out} E^* S Q_{\tilde\sB_\out} T^\out Q_{\tilde\sB_\iin} & \multirow{4}{*}{$\pi_a(PT^\iin_E)$} & E & \multirow{4}{*}{$\cont{U_a}$} & Q_{\sB_\iin} Q_\sB & \\
    \cline{1-3} \cline{5-5} \cline{7-7}
  \tilde \sA && Q_{\tilde\sA} && E &&& \\
    \cline{1-3} \cline{5-5} \cline{7-8}
  \tilde \sE && Q_{\tilde\sE} && E &&& \ket{0}\ket{\mathrm{on}} \\
    \cline{1-3} \cline{5-5} \cline{7-8}
  \tilde \sM & \bra{a,s} E^* \hat P_{\tilde\sM} & Q_{\tilde\sM} && E &&& \ket{\mu} \\
    \hline
  \end{array}
\end{align}
We now have an expression for the decryption Pauli $\hat P_{\tilde\sM}$ used by the BR-OTP:
\begin{align}
  \hat P_{\tilde\sM} \defeq \pi_a(PT^\iin_E)_{\tilde\sM}.
\end{align}
In order to clean up the extra phases we introduced by the above commutation relations we write
\begin{align}
  \alpha_{Q,T^\iin} \defeq y(Q_{\tilde\sB_\out}) y(Q_{\sB_\iin}) c(T^\iin,Q_{\sB_\iin}) \alpha_Q
\end{align}
for each choice of Paulis $Q,T^\iin$ so that the original state \eqref{fig:tabular-state} can be written
\begin{align}
  \label{eq:pre-sandwich-state}
  \ket{\eqref{fig:tabular-state}} =  \abs{\sB}^{-1/2}\abs{\tilde\sB_\iin}^{-1/2}
 \sum_{\textrm{Paulis $Q$}} \alpha_{Q,T^\iin} \ket{\eqref{tab:pre-sandwich}}.
\end{align}

Until now we have used only simple commutation relations to derive an alternate expression for the state $\ket{\eqref{tab:kraus-pauli}}$.
Now we wish to employ the Pauli sandwich on register $\tilde\sM$ so that the double sum over $Q_{\tilde\sM}$ becomes a single sum.
To this end consider the table
\begin{align}
  \label{tab:pre-sandwich-M}
  \renewcommand{\arraystretch}{1.3}
  \begin{array}{|c||c|c|c|c|c|c|}
    \hline
  \sW &  Q_\sW &&&&& \multirow{3}{*}{$\ket{\psi}$} \\
    \cline{1-6}
  \sB_\out & Q_{\sB_\out} E^* S Q_{\tilde\sB_\out} T^\out Q_{\tilde\sB_\iin} & \pi_a(PT^\iin_E)_{\sB_\out} & E & \multirow{4}{*}{$\cont{U_a}$} & Q_{\sB_\iin} Q_\sB & \\
    \cline{1-4} \cline{6-6}
  \tilde \sA & Q_{\tilde\sA} & \pi_a(PT^\iin_E)_{\tilde\sA} & E &&& \\
    \cline{1-4} \cline{6-7}
  \tilde \sE & Q_{\tilde\sE} & \pi_a(PT^\iin_E)_{\tilde\sE} & E &&& \ket{0}\ket{\mathrm{on}} \\
    \cline{1-4} \cline{6-7}
  \tilde \sM &&&&&& \ket{\mu} \\
    \hline
  \end{array}
\end{align}
obtained from \eqref{tab:pre-sandwich} by deleting everything in row $\tilde\sM$ occurring after the operator $\cont{U_a}$.
For each Pauli $Q_{\tilde\sM}$ define
\begin{align}
  \label{eq:L-def}
  \ket{Q_{\tilde\sM}} \defeq
  \sum_{\lnot Q_{\tilde\sM}} \alpha_{Q,T_\iin} \ket{\eqref{tab:pre-sandwich-M}}.
\end{align}
Here the summation over $\lnot Q_{\tilde\sM}$ is shorthand for a summation over all Paulis $Q_\sB$, $Q_{\sB_\iin}$, $Q_{\tilde\sB_\iin}$, $Q_{\tilde\sB_\out}$, $Q_{\sB_\out}$, $Q_{\tilde\sA}$, $Q_{\tilde\sE}$, $Q_\sW$---that is, all Paulis comprising $Q$ except the Pauli $Q_{\tilde\sM}$ acting on $\tilde\sM$.
Note that the only dependence of $\ket{Q_{\tilde\sM}}$ on $Q_{\tilde\sM}$ is in the scalar $\alpha_{Q,T^\iin}$.
Keep in mind that $\ket{Q_{\tilde\sM}}$ also depends upon $T^\iin$, $T^\out$, $a$, $s$, $E$, $\pi_a(PT^\iin_E)_{\tilde\sB_\out}$, $\pi_a(PT^\iin_E)_{\tilde\sA}$, $\pi_a(PT^\iin_E)_{\tilde\sE}$, $S$; we have simply omitted these parameters for brevity.

By the Pauli sandwich (Lemma \ref{lm:pauli-sandwich}) we have that the mixed state $\ptr{\sK}{\rho_\mathrm{real}}$ from \eqref{eq:user-channel} equals
\begin{align}
  \ptr{\sK}{\rho_\mathrm{real}} &=
  \frac{1}{\abs{\mathscr{E}}\abs{\mathscr{P}}\abs{\mathscr{S}}\abs{\sB}\abs{\tilde\sB_\iin}}
  \sum_{E,P,S,a,s,T^\iin,T^\out,Q_{\tilde\sM}}
  \bra{a}\bra{s} E^* Q_{\tilde\sM} E \ket{Q_{\tilde\sM}} \bra{Q_{\tilde\sM}} E^* Q_{\tilde\sM} E \ket{a}\ket{s}
  \label{eq:user-channel-post-sandwich}
\end{align}
From here consider two cases corresponding to zero and non-zero syndrome measurements $(s=0,s\neq 0$), indicating acceptance and rejection by the BR-OTP, respectively.
In so doing we decompose $\rho_\mathrm{real}$ into a sum of three positive semidefinite operators
\begin{align}
  \label{eq:real-goal}
  \rho_\mathrm{real} = R_\mathrm{rej} + R_\mathrm{acc} + [\varepsilon_\mathrm{real}]
\end{align}
where $[\varepsilon_\mathrm{real}]$ has trace at most $\varepsilon$.
Then in Section \ref{sec:sim:analysis} we will show that
\begin{align}
  \label{eq:sim-goal}
  \rho_\mathrm{sim} = R_\mathrm{rej} + R_\mathrm{acc} + [\varepsilon_\mathrm{sim}]
\end{align}
for some $[\varepsilon_\mathrm{sim}]$ that has trace at most $\varepsilon$, from which it will follow that
\begin{align}
  \tnorm{\rho_\mathrm{real} - \rho_\mathrm{sim}} = \tnorm{[\varepsilon_\mathrm{real}] - [\varepsilon_\mathrm{sim}]} \leq 2\varepsilon
\end{align}
as desired.

\subsubsection{In the event of acceptance}
\label{sec:user:accept}

Consider the unnormalized mixed state obtained from the expression \eqref{eq:user-channel-post-sandwich} for $\ptr{\sK}{\rho_\mathrm{real}}$ by restricting attention to those terms in the summation with syndrome outcome $s=0$.
(See Appendix \ref{appendix:encode-encrypt} for explanations of the notation for the logical Pauli $Q_{\ell(E)}$ induced by $E$ and for the subset $\mathscr{E}_\sX(Q)\subset\mathscr{E}$ of codes $E$ for which $Q_{\ell(E)}$ is a purely $Z$-Pauli and $Q$ has no $X$-error syndrome.)
It is easy to see (and follows immediately from the work of Appendix \ref{sec:auth:decode-measure}) that this unnormalized mixed state is
\begin{align}
 \label{eq:R-acc}
 \underbrace{
  \frac{1}{\abs{\mathscr{E}}\abs{\mathscr{P}}\abs{\mathscr{S}}\abs{\sB}\abs{\tilde\sB_\iin}}
  \sum_{P,S,a,T^\iin,T^\out,Q_{\tilde\sM}}
    \sum_{E\in\mathscr{E}_\sX(Q_{\tilde\sM})}
    \bra{a}\ket{Q_{\tilde\sM}}\bra{Q_{\tilde\sM}}\ket{a}
   }_{\ptr{\sK}{R_\mathrm{acc}}}
    + \Ptr{\sK}{[\varepsilon_\mathrm{real}]}
\end{align}
for some choice of $[\varepsilon_\mathrm{real}]$.
By the security of the decode-then-measure procedure used by the BR-OTP it must be that $\ptr{\sK}{[\varepsilon_\mathrm{real}]}$ and hence $[\varepsilon_\mathrm{real}]$ have trace at most $\varepsilon$.
Henceforth we concentrate only on the highlighted term of \eqref{eq:R-acc}, which is taken to be $\ptr{\sK}{R_\mathrm{acc}}$ for the operator $R_\mathrm{acc}$ appearing in the desired decomposition \eqref{eq:real-goal} of $\rho_\mathrm{real}$.

Because the syndrome measurement succeeded, the BR-OTP reveals a new register $\sK$ containing a classical description $\ket{\hat S}$ of a decryption key $\hat S$.
(An explicit formula for this key is given in Section \ref{sec:user:key}.)
The unnormalized pure state $\bra{a}\ket{Q_{\tilde\sM}}$ plus register $\sK$ can be written
\begin{align}
  \bra{a}\ket{Q_{\tilde\sM}}\ket{\hat S} = \sum_{\lnot Q_{\tilde\sM}} \alpha_{Q,T^\iin} \ket{\eqref{tab:post-decode}}
\end{align}
where the table \eqref{tab:post-decode} is given by
\begin{align}
  \label{tab:post-decode}
  \renewcommand{\arraystretch}{1.3}
  \begin{array}{|c||c|c|c|c|c|c|}
    \hline
  \sW &  Q_\sW &&&&& \multirow{3}{*}{$\ket{\psi}$} \\
    \cline{1-6}
  \sB_\out & Q_{\sB_\out} E^* S Q_{\tilde\sB_\out} T^\out Q_{\tilde\sB_\iin} & \pi_a(PT^\iin_E)_{\sB_\out} & E & \multirow{4}{*}{$\cont{U_a}$} & Q_{\sB_\iin} Q_\sB & \\
    \cline{1-4} \cline{6-6}
  \tilde \sA & Q_{\tilde\sA} & \pi_a(PT^\iin_E)_{\tilde\sA} & E &&& \\
    \cline{1-4} \cline{6-7}
  \tilde \sE & Q_{\tilde\sE} & \pi_a(PT^\iin_E)_{\tilde\sE} & E &&& \ket{0}\ket{\mathrm{on}} \\
    \cline{1-4} \cline{6-7}
  \tilde \sM &\multicolumn{2}{c|}{}& \bra{a} &&& \ket{\mu} \\
    \hline
  \sK &&&&&& \ket{\hat S} \\
    \hline
  \end{array}
\end{align}
We may now employ the identity \eqref{eq:universal-a} from Section \ref{sec:universal} to get
\begin{align}
  \bra{a}\cont{U_a}\ket{\mu}\ket{\mathrm{on}}=\frac{1}{2^{r/2}}\cont{U}\ket{\mathrm{on}}=\frac{1}{2^{r/2}}U
\end{align}
so the state $\ket{\eqref{tab:post-decode}}$ becomes
\begin{align}
  \label{tab:post-magic}
  \ket{\eqref{tab:post-decode}} =
  \frac{1}{2^{r/2}}
  \renewcommand{\arraystretch}{1.3}
  \begin{array}{|c||c|c|c|c|c|c|}
    \hline
  \sW &  Q_\sW &&&&& \multirow{3}{*}{$\ket{\psi}$} \\
    \cline{1-6}
  \sB_\out & Q_{\sB_\out} E^* S Q_{\tilde\sB_\out} T^\out Q_{\tilde\sB_\iin} & \pi_a(PT^\iin_E)_{\sB_\out} & E & \multirow{3}{*}{$U$} & Q_{\sB_\iin} Q_\sB & \\
    \cline{1-4} \cline{6-6}
  \tilde \sA & Q_{\tilde\sA} & \pi_a(PT^\iin_E)_{\tilde\sA} & E &&& \\
    \cline{1-4} \cline{6-7}
  \tilde \sE & Q_{\tilde\sE} & \pi_a(PT^\iin_E)_{\tilde\sE} & E &&& \ket{0} \\
    \hline
  \sK &&&&&& \ket{\hat S} \\
    \hline
  \end{array}
\end{align}
Write
\begin{align}
  \ket{Q_{\tilde\sM},\hat S} \defeq \sum_{\lnot Q_{\tilde\sM}} \alpha_{Q,T^\iin} \ket{\eqref{tab:post-magic}}
\end{align}
so that
\begin{align}
  \label{eq:R-acc2}
  R_\mathrm{acc} =
  \frac{1}{\abs{\mathscr{E}}\abs{\mathscr{P}}\abs{\mathscr{S}}\abs{\sB}\abs{\tilde\sB_\iin}2^r}
  \sum_{P,S,a,T^\iin,T^\out,Q_{\tilde\sM}}
    \sum_{E\in\mathscr{E}_\sX(Q_{\tilde\sM})}
    \ket{Q_{\tilde\sM},\hat S}\bra{Q_{\tilde\sM},\hat S}
\end{align}
Notice that the receiver/environment learns nothing about the keys $\pi_a(PT^\iin_E)_{\tilde\sA},\pi_a(PT^\iin_E)_{\tilde\sE}$ so the registers $\tilde\sA,\tilde\sE$ remain completely mixed.
Thus, the receiver gets only the $\sB_\out$ portion of the final state as desired.

\subsubsection{In the event of rejection}
\label{sec:user:reject}

Consider the unnormalized mixed state obtained from the expression \eqref{eq:user-channel-post-sandwich} for $\ptr{\sK}{\rho_\mathrm{real}}$ by restricting attention to those terms in the summation with syndrome outcome $s\neq 0$.
(Again, the reader is reffered Appendix \ref{appendix:encode-encrypt} for explanations of the notation $Q_{\ell(E)}$ and $\mathscr{E}_{\sX\emptyset}(Q)$.)
According to the analysis of Appendix \ref{sec:auth:decode-measure} this unnormalized mixed state is
\begin{align}
\label{eq:R-rej}
  \frac{1}{\abs{\mathscr{E}}\abs{\mathscr{P}}\abs{\mathscr{S}}\abs{\sB}\abs{\tilde\sB_\iin}}
  \sum_{P,S,a,s\neq 0,T^\iin,T^\out,Q_{\tilde\sM}}
    \sum_{E\in\mathscr{E}_{\sX\emptyset}(Q_{\tilde\sM})}
    \bra{a}Q_{\tilde\sM,\ell(E)}\ket{Q_{\tilde\sM}}\bra{Q_{\tilde\sM}}Q_{\tilde\sM,\ell(E)}\ket{a}
\end{align}
The operator of \eqref{eq:R-rej} shall be taken to be $\ptr{\sK}{R_\mathrm{rej}}$ for the operator $R_\mathrm{rej}$ appearing in the desired decomposition \eqref{eq:real-goal} of $\rho_\mathrm{real}$.

Because the syndrome measurement failed, the register $\sK$ revealed by the BR-OTP is completely mixed and so we can ignore it for the rest of this analysis.
The unnormalized pure state $\bra{a}Q_{\tilde\sM,\ell(E)}\ket{Q_{\tilde\sM}}$ can be written
\begin{align}
  \bra{a}Q_{\tilde\sM,\ell(E)}\ket{Q_{\tilde\sM}} = \sum_{\lnot Q_{\tilde\sM}} \alpha_{Q,T^\iin} \ket{\eqref{tab:post-decode-reject}}
\end{align}
where the table \eqref{tab:post-decode-reject} is given by
\begin{align}
  \label{tab:post-decode-reject}
  \renewcommand{\arraystretch}{1.3}
  \begin{array}{|c||c|c|c|c|c|c|c|}
    \hline
  \sW && Q_\sW &&&&& \multirow{3}{*}{$\ket{\psi}$} \\
    \cline{1-7}
  \sB_\out && Q_{\sB_\out} E^* S Q_{\tilde\sB_\out} T^\out Q_{\tilde\sB_\iin} & \pi_a(PT^\iin_E)_{\sB_\out} & E & \multirow{4}{*}{$\cont{U_a}$} & Q_{\sB_\iin} Q_\sB & \\
    \cline{1-5} \cline{7-7}
  \tilde \sA && Q_{\tilde\sA} & \pi_a(PT^\iin_E)_{\tilde\sA} & E &&& \\
    \cline{1-5} \cline{7-8}
  \tilde \sE && Q_{\tilde\sE} & \pi_a(PT^\iin_E)_{\tilde\sE} & E &&& \ket{0}\ket{\mathrm{on}} \\
    \cline{1-5} \cline{7-8}
  \tilde \sM & \bra{a} & Q_{\tilde\sM,\ell(E)} &&&&& \ket{\mu} \\
    \hline
  \end{array}
\end{align}
For each $Q_{\tilde\sM,\ell(E)}$ let $\ket{a'}=Q_{\tilde\sM,\ell(E)}\ket{a}$ denote the modified vector of magic state measurement results.
We may now employ the identity \eqref{eq:universal-a} from Section \ref{sec:universal} to get
\begin{align}
  \bra{a'}\cont{U_a}\ket{\mu} = \frac{1}{2^{r/2}}\cont{U_{a-a'}}
\end{align}
so the state $\ket{\eqref{tab:post-decode-reject}}$ becomes
\begin{align}
  \label{tab:post-magic-reject}
  \ket{\eqref{tab:post-decode-reject}} =
  \frac{1}{2^{r/2}}
  \renewcommand{\arraystretch}{1.3}
  \begin{array}{|c||c|c|c|c|c|c|}
    \hline
  \sW &  Q_\sW &&&&& \multirow{3}{*}{$\ket{\psi}$} \\
    \cline{1-6}
  \sB_\out & Q_{\sB_\out} E^* S Q_{\tilde\sB_\out} T^\out Q_{\tilde\sB_\iin} & \pi_a(PT^\iin_E)_{\sB_\out} & E & \multirow{3}{*}{$\cont{U_{a-a'}}$} & Q_{\sB_\iin} Q_\sB & \\
    \cline{1-4} \cline{6-6}
  \tilde \sA & Q_{\tilde\sA} & \pi_a(PT^\iin_E)_{\tilde\sA} & E &&& \\
    \cline{1-4} \cline{6-7}
  \tilde \sE & Q_{\tilde\sE} & \pi_a(PT^\iin_E)_{\tilde\sE} & E &&& \ket{0} \ket{\mathrm{on}}\\
    \hline
  \end{array}
\end{align}
In the case of rejection the sender has no information about the encryption Paulis $\pi_a(PT^\iin_E)_{\sB_\out}$, $\pi_a(PT^\iin_E)_{\tilde\sA}$, and $\pi_a(PT^\iin_E)_{\tilde\sE}$.
Thus, the unitary $E\cont{U_{a-a'}}$ appearing immediately before the encryption Paulis is superfluous---it can be removed without affecting the overall state.
Moreover, the registers $\tilde\sA,\tilde\sE$ remain encrypted and therefore completely mixed from the receiver's point of view.
Thus, one could replace the state $\ket{\eqref{tab:post-decode-reject}}$ with, say, the following
\begin{align}
  \label{tab:post-magic-reject-no-U}
  \ket{\eqref{tab:post-decode-reject}} =
  \frac{1}{2^{r/2}}
  \renewcommand{\arraystretch}{1.3}
  \begin{array}{|c||c|c|c|c|}
    \hline
  \sW &  Q_\sW &&& \ket{\psi}_\sW \\
    \hline
  \sB_\out & Q_{\sB_\out} E^* S Q_{\tilde\sB_\out} T^\out Q_{\tilde\sB_\iin} & \pi_a(PT^\iin_E)_{\sB_\out} & Q_{\sB_\iin} Q_\sB & \ket{\psi}_\sB \\
    \hline
  \tilde \sA & Q_{\tilde\sA} & \pi_a(PT^\iin_E)_{\tilde\sA} && \ket{\mathrm{anything}}\\
    \hline
  \tilde \sE & Q_{\tilde\sE} & \pi_a(PT^\iin_E)_{\tilde\sE} && \ket{\mathrm{anything}}\\
    \hline
  \end{array}
\end{align}
Letting
\begin{align}
  \ket{Q_{\tilde\sM},\mathrm{anything}} \defeq \sum_{\lnot Q_{\tilde\sM}} \alpha_{Q,T^\iin} \ket{\eqref{tab:post-magic-reject-no-U}}
\end{align}
we have
\begin{align}
  \ptr{\sK}{R_\mathrm{rej}} =
  \frac{1}{\abs{\mathscr{E}}\abs{\mathscr{P}}\abs{\mathscr{S}}\abs{\sB}\abs{\tilde\sB_\iin}2^r}
  \sum_{P,S,a,s\neq 0,T^\iin,T^\out,Q_{\tilde\sM}}
    \sum_{E\in\mathscr{E}_{\sX\emptyset}(Q_{\tilde\sM})}
    \ket{Q_{\tilde\sM},\mathrm{anything}}\bra{Q_{\tilde\sM},\mathrm{anything}}
\end{align}

\subsubsection{More analysis to determine the final key}
\label{sec:user:key}

Now that we have analyzed an arbitrary environment/receiver for our QOTP we can easily derive an expression for the decryption key $\hat S$ used by an honest receiver to recover his output register $\sB$ at the end of the protocol.
The derivation in this section is not a prerequisite for the discussion of the simulator in Section \ref{sec:simulator}.
It is included here only for completeness.

We claim that the decryption Pauli $\hat S$ is the logical Pauli induced on the data register by code $E$ and Pauli
\begin{align}
  S T^\out \pi_a(PT^\iin_E)_{\sB_\out} \enspace .
\end{align}
To see this, observe that an honest receiver will place no amplitude on terms for which $Q\neq I$.
Our calculations are therefore simplified by assuming $Q=I$.
Under this assumption the state $\ket{\eqref{tab:post-magic}}$ becomes
\begin{align}
  \ket{\eqref{tab:post-magic}} =
  \renewcommand{\arraystretch}{1.3}
  \begin{array}{|c||c|c|c|c|c|}
    \hline
  \sW &&&&& \multirow{3}{*}{$\ket{\psi}$} \\
    \cline{1-5}
  \sB_\out & E^* S T^\out & \pi_a(PT^\iin_E)_{\sB_\out} & E & \multirow{3}{*}{$U$} & \\
    \cline{1-4}
  \tilde \sA && \pi_a(PT^\iin_E)_{\tilde\sA} & E && \\
    \cline{1-4} \cline{6-6}
  \tilde \sE && \pi_a(PT^\iin_E)_{\tilde\sE} & E && \ket{0} \\
    \hline
  \sK &&&&& \ket{\hat S} \\
    \hline
  \end{array}
\end{align}
and so we see that applying $\hat S$ to register $\sB_\out$ decrypts that register.
Of course, $\tilde\sA,\tilde\sE$ remain encrypted and the purification register $\sW$ is not in the receiver's posession, so the receiver obtains only the $\sB$ portion of $(\Phi\ot\idsup{\sW})(\ket{\psi}\bra{\psi})$ as desired.

\subsection{Specification of the simulator}
\label{sec:simulator}

In Section \ref{sec:user} we derived an expression for the mixed state $\rho_\mathrm{real}$ of the registers $(\sB_\out,\tilde\sA,\tilde\sE,\sW,\sK)$ in the environment's possession at the end of step \ref{it:receiver:dump}.
In this section we specify a simulator, and in the following section we show that the state $\rho_\mathrm{sim}$ of the environment's registers at the end of its interaction with the simulator is close in trace distance to $\rho_\mathrm{real}$, from which security is established.

This simulator must not interact with the sender.
Instead, the simulator is permitted only one-shot, black-box access to the ``ideal functionality'' for $\Phi$.
We represent this ideal functionality by a single call to an oracle for $U$ acting on registers $(\sA,\sB,\sE)$ prepared by the simulator.
The rules for permissible preparation and disposal of these registers are as follows:
\begin{enumerate}

\item
  When given the initial state $\ket{\psi}$ of registers $(\sA,\sB,\sW)$ selected by the environment, the simulator \emph{must} pass the register $\sA$ directly to $U$ without any pre-processing.

\item
  The simulator must prepare the ancillary register $\sE$ in pure state $\ket{0}$.

\item
  Upon receiving the output registers $(\sA,\sB,\sE)$, the simulator must discard registers $\sA,\sE$ without any post-processing.

\end{enumerate}

The simulator constructs registers as listed in Protocol~\ref{prot:sim} and then acts as specified described in the protocol.
The main idea is that our simulator will use the control bit contained in register $\tilde\sE$ to ``switch off'' the application of $U$ that would have been implemented by the real receiver interacting with the sender's BR-OTP.
Instead, the black-box call to the ideal functionality will be embedded at the proper time so as to recover the required action of $U$.
An additional teleportation step is required so that our simulator can embed $U$ at the proper time.

\begin{protocol} \caption{Simulator} \label{prot:sim}

\textbf{Registers prepared by the simulator.} \\
Given the input register $\sA$ our simulator constructs the following registers:
\begin{center}
\begin{tabularx}{\textwidth}{lX}
  $(\sB_\iin,\sS_\iin)$: &
    Simple EPR pairs $\ket{\phi^+}$ for teleportation.\\
  $(\sS_\out,\tilde\sB_\iin)$: &
    Teleport-through-authentication state of Protocol \ref{prot:prep-sender}.\\
  $(\tilde\sB_\out,\sB_\out)$: &
    Teleport-through-de-authentication state of Protocol \ref{prot:prep-sender}.\\
  $\tilde\sA$: &
    Authenticated dummy input state for the sender $P_{\tilde\sA}E\ket{0}$.\\
  $\tilde\sE$: &
    Authenticated dummy ancillary state $P_{\tilde\sE}E\ket{0}\ket{\mathrm{off}}$.\\
  $\tilde \sM$: &
    Authenticated magic states as in Protocol \ref{prot:prep-sender}.\\
  $\sA':$ &
    To be used in the call to the ideal functionality.
    Contains the portion of $\ket{\psi}$ contained in register $\sA$.\\
  $\sE':$ &
    To be used in the call to the ideal functionality.
    Ancillary register in state $\ket{0}$.
\end{tabularx}
\end{center}

\textbf{Execution of the simulator.}

\begin{enumerate}

\item 
Prepare registers $(\sB_\iin,\tilde\sB_\iin,\tilde\sB_\out,\sB_\out,\tilde\sA,\tilde\sE,\tilde\sM,\sW)$ as above and send these registers to the environment.

\item
The environment responds with a Pauli $T^\iin$.
Apply $T^\iin$ to register $\sS_\iin$.
Then use the ideal black-box to apply $U$ to $(\sA',\sS_\iin,\sE')$.

\item
Perform a Bell measurement on $(\sS_\iin,\sS_\out)$ so as to teleport the contents of $\sS_\iin$ through the authentication and place the result in $\tilde\sB_\iin$.
Let $T^\smm$ denote the teleportation Pauli indicated by this Bell measurement.

\item
Execute Protocol \ref{prot:BR-OTP} for the BR-OTP under the assumption that $T^\smm$ was received in the first round.
(This step is not difficult: the responses of the BR-OTP depend only upon the choice of code $E$ and encryption Pauli $P$.)

\end{enumerate}
\end{protocol}

\subsection{Analysis of the environment's interaction with the simulator}
\label{sec:sim:analysis}

Fix a choice of code $E$ and encryption Paulis $P,S$.
We re-use notation from Section \ref{sec:user} as much as possible.
In Section \ref{sec:sim:tabular-representation} we wrote the unnormalized pure state of the entire system as a large table $\ket{\eqref{fig:tabular-state}}$.
We then introduced some shorthand notation in Section \ref{sec:user} that allowed us to write that table more compactly as $\ket{\eqref{tab:easyK}}$.
The equivalent of table \eqref{tab:easyK} for our simulator is as follows.
\begin{align}
  \label{tab:sim-easyK}
  \renewcommand{\arraystretch}{1.3}
  \begin{array}{|c||c|c|c|c|}
    \hline
  \sE' &&& (U)_\sE & \ket{0} \\
    \hline
  \sA' &&& (U)_\sA & \ket{\psi}_\sA \\
    \hline
  \sB & \multirow{2}{*}{$\bra{T^\iin}B$} & (K_a)_\sB && \ket{\psi}_\sB \\
    \cline{1-1} \cline{3-5}
  \sB_\iin && (K_a)_{\sB_\iin} && \multirow{2}{*}{$\ket{\phi^+}$} \\
    \cline{1-4}
  \sS_\iin & \multirow{2}{*}{$\bra{T^\smm}B$} && (U)_\sB T^\iin & \\
    \cline{1-1} \cline{3-5}
  \sS_\out &&&& \multirow{2}{*}{$\ket{\phi^+}$} \\
    \cline{1-4}
  \tilde\sB_\iin & \multirow{2}{*}{$\bra{T^\out}B$} & \multirow{7}{*}{$K_a$} & P_{\tilde\sB_\iin} E & \\
    \cline{1-1} \cline{4-5}
  \tilde\sB_\out &&&& \multirow{2}{*}{$\ket{\phi^+}$} \\
    \cline{1-2} \cline{4-4}
  \sB_\out &&& E^* S &  \\
    \cline{1-2} \cline{4-5}
  \tilde\sA &&& P_{\tilde\sA} E & \ket{0} \\
    \cline{1-2} \cline{4-5}
  \tilde\sE &&& P_{\tilde\sE} E & \ket{0}\ket{\mathrm{off}} \\
    \cline{1-2} \cline{4-5}
  \tilde\sM & \bra{a}\bra{s} E^* \hat P_{\tilde\sM} && P_{\tilde\sM} E & \ket{\mu} \\
    \cline{1-2} \cline{4-5}
  \sW &&&& \ket{\psi}_\sW \\
    \hline
  \end{array}
\end{align}
Then
\begin{align}
  \label{eq:sim-channel}
  \ptr{\sK}{\rho_\mathrm{sim}} =
  \frac{1}{\abs{\mathscr{E}}\abs{\mathscr{P}}\abs{\mathscr{S}}}
  \sum_{E,P,S,a,s,T^\iin,T^\out,T^\smm}
  \Ptr{\sA'\sE'}{\ket{\eqref{tab:sim-easyK}}\bra{\eqref{tab:sim-easyK}} } \enspace .
\end{align}
We then inserted a superfluous $I=(\tildeconta{U})^*(\tildeconta{U})$ into the table \eqref{tab:easyK} and derived the table \eqref{tab:kraus-pauli} by substituting the Pauli decomposition \eqref{eq:pauli-decomps} of $K_a (\tildeconta{U})^*$ and fixing a choice $Q$ of Pauli in that decomposition.
It is a simple matter to repeat those steps in the current setting.
The equivalent of table \eqref{tab:kraus-pauli} for our simulator is
\begin{align}
  \label{tab:sim-Pauli}
  \renewcommand{\arraystretch}{1.3}
  \begin{array}{|c||c|c|c|c|c|c|c|} \hline
  \sE' &&&& (U)_{\sE'} & \ket{0} \\
    \hline
  \sA' &&&& (U)_{\sA'} & \ket{\psi}_{\sA} \\
    \hline
  \sB & \multirow{2}{*}{$\bra{T^\iin}B$} & Q_\sB &&& \ket{\psi}_\sB \\
    \cline{1-1}\cline{3-6}
  \sB_\iin && Q_{\sB_\iin} &&& \multirow{2}{*}{$\ket{\phi^+}$} \\
    \cline{1-5}
  \sS_\iin & \multirow{2}{*}{$\bra{T^\smm}B$} &&& (U)_\sB T^\iin & \\
    \cline{1-1}\cline{3-6}
  \sS_\out &&&&& \multirow{2}{*}{$\ket{\phi^+}$} \\
    \cline{1-5}
  \tilde\sB_\iin & \multirow{2}{*}{$\bra{T^\out}B$} & Q_{\tilde\sB_\iin} & (\tildeconta{U})_{\tilde\sB_\iin} & P_{\tilde\sB_\iin} E & \\
    \cline{1-1}\cline{3-6}
  \tilde\sB_\out && Q_{\tilde\sB_\out} &&& \multirow{2}{*}{$\ket{\phi^+}$} \\
    \cline{1-5}
  \sB_\out && Q_{\sB_\out} && E^* S & \\ \hline
  \tilde \sA && Q_{\tilde\sA} & (\tildeconta{U})_{\tilde\sA} & P_{\tilde\sA} E & \ket{0} \\
    \hline
  \tilde \sE && Q_{\tilde\sE} & (\tildeconta{U})_{\tilde\sE} & P_{\tilde\sE} E & \ket{0}\ket{\mathrm{off}} \\
    \hline
  \tilde \sM & \bra{a}\bra{s} E^* \hat P_{\tilde\sM}' & Q_{\tilde\sM} & (\tildeconta{U})_{\tilde\sM} & P_{\tilde\sM} E & \ket{\mu} \\
    \hline
  \sW && Q_\sW &&& \ket{\psi}_\sW \\
    \hline
  \end{array}
\end{align}
In Section \ref{sec:user} we applied a teleportation identity to derive table \eqref{tab:collect-Q} from table \eqref{tab:kraus-pauli}.
We apply the same teleportation identity here so as to derive the following table from \eqref{tab:sim-Pauli}.
\begin{align} \label{tab:sim-collect-Q}
  \renewcommand{\arraystretch}{1.3}
  \begin{array}{|c||c|c|c|c|c|c|c|c|} \hline
  \sE' &&&&&& \multirow{3}{*}{$U$} && \ket{0} \\
    \cline{1-6} \cline{8-9}
  \sA' &&&&&&&& \ket{\psi}_{\sA} \\
    \cline{1-6} \cline{8-9}
  \sB_\out && Q_{\sB_\out} E^* S Q_{\tilde\sB_\out}^\trans T^\out Q_{\tilde\sB_\iin} & \multirow{4}{*}{$\tildeconta{U}$} & \multirow{4}{*}{$P$} & E T^\smm && T^\iin Q_{\sB_\iin}^\trans T^\iin Q_\sB & \ket{\psi}_\sB \\
    \cline{1-3} \cline{6-9}
  \tilde \sA && Q_{\tilde\sA} &&& E &&& \ket{0} \\
    \cline{1-3} \cline{6-9}
  \tilde \sE && Q_{\tilde\sE} &&& E &&& \ket{0} \ket{\mathrm{off}} \\
    \cline{1-3} \cline{6-9}
  \tilde \sM & \bra{a}\bra{s} E^* \hat P_{\tilde\sM}' & Q_{\tilde\sM} &&& E &&& \ket{\mu} \\
    \hline
  \sW && Q_\sW &&&&&& \ket{\psi}_\sW \\
    \hline
  \end{array}
\end{align}
We then replaced transposed Paulis with their un-transposed equivalents and employed several commutation relations to derive table \eqref{tab:pre-sandwich}.
That derivation can be repeated almost exactly in the present context.
The only significant difference is that, after commuting $Q_{\sB_\iin}$ with $T^\iin$ and picking up a phase of $c(T^\iin,Q_{\sB_\iin})$, the two $T^\iin$ Paulis annihilate each other.
The state $\ket{\eqref{tab:sim-Pauli}}$ becomes
\begin{align}
\label{tab:sim-pre-sandwich}
  \renewcommand{\arraystretch}{1.3}
  \begin{array}{|c||c|c|c|c|c|c|c|c|} \hline
  \sE' &&&&&& \multirow{3}{*}{$U$} && \ket{0} \\
    \cline{1-6} \cline{8-9}
  \sA' &&&&&&&& \ket{\psi}_{\sA} \\
    \cline{1-6} \cline{8-9}
  \sB_\out && Q_{\sB_\out} E^* S Q_{\tilde\sB_\out} T^\out Q_{\tilde\sB_\iin} & \multirow{4}{*}{$\pi_a(PT^\smm_E)$} & E & \multirow{4}{*}{$\cont{U_a}$} && Q_{\sB_\iin} Q_\sB & \ket{\psi}_\sB \\
    \cline{1-3} \cline{5-5} \cline{7-9}
  \tilde \sA && Q_{\tilde\sA} && E &&&& \ket{0} \\
    \cline{1-3} \cline{5-5} \cline{7-9}
  \tilde \sE && Q_{\tilde\sE} && E &&&& \ket{0} \ket{\mathrm{off}} \\
    \cline{1-3} \cline{5-5} \cline{7-9}
  \tilde \sM & \bra{a}\bra{s} E^* \hat P_{\tilde\sM}' & Q_{\tilde\sM} && E &&&& \ket{\mu} \\
    \hline
  \sW && Q_\sW &&&&&& \ket{\psi}_\sW \\
    \hline
  \end{array}
\end{align}
We now see that the decryption Pauli $\hat P_{\tilde\sM}'$ used in the BR-OTP is given by
\begin{align}
  \hat P_{\tilde\sM}' \defeq \pi_a(PT^\smm_E)_{\tilde\sM} \enspace .
\end{align}
By analogy to the formula \eqref{eq:pre-sandwich-state} for the state $\ket{\eqref{fig:tabular-state}}$ in Section \ref{sec:user} we have
\begin{align}
  \label{eq:pre-sandwich-state-reject}
  \ket{\eqref{tab:sim-easyK}} = \Pa{\abs{\sB}\abs{\tilde\sB_\iin}\abs{\sS_\iin}}^{-1/2}
 \sum_{\textrm{Paulis $Q$}} \alpha_{Q,T^\iin} \ket{\eqref{tab:sim-pre-sandwich}} \enspace .
\end{align}
As before, we wish to employ the Pauli sandwich on $\tilde\sM$.
To this end consider the table
\begin{align}
  \label{tab:sim-pre-sandwich-M}
  \renewcommand{\arraystretch}{1.3}
  \begin{array}{|c||c|c|c|c|c|c|c|} \hline
  \sE' &&&&& \multirow{3}{*}{$U$} && \ket{0} \\
    \cline{1-5} \cline{7-8}
  \sA' &&&&&&& \ket{\psi}_{\sA} \\
    \cline{1-5} \cline{7-8}
  \sB_\out & Q_{\sB_\out} E^* S Q_{\tilde\sB_\out} T^\out Q_{\tilde\sB_\iin} & \pi_a(PT^\smm_E)_{\sB_\out} & E & \multirow{4}{*}{$\cont{U_a}$} && Q_{\sB_\iin} Q_\sB & \ket{\psi}_\sB \\
    \cline{1-4} \cline{6-8}
  \tilde \sA & Q_{\tilde\sA} & \pi_a(PT^\smm_E)_{\tilde\sA} & E &&&& \ket{0} \\
    \cline{1-4} \cline{6-8}
  \tilde \sE & Q_{\tilde\sE} & \pi_a(PT^\smm_E)_{\tilde\sE} & E &&&& \ket{0} \ket{\mathrm{off}} \\
    \cline{1-4} \cline{6-8}
  \tilde \sM &&&&&&& \ket{\mu} \\
    \hline
  \sW & Q_\sW &&&&&& \ket{\psi}_\sW \\
    \hline
  \end{array}
\end{align}
obtained from \eqref{tab:sim-pre-sandwich} by deleting everything in row $\tilde\sM$ occurring after the operator $\cont{U_a}$.
By analogy to the state $\ket{Q_{\tilde\sM}}$ defined in \eqref{eq:L-def} in Section \ref{sec:user}, for each Pauli $Q_{\tilde\sM}$ define
\begin{align}
  \ket{Q_{\tilde\sM}'} \defeq \sum_{\lnot Q_{\tilde\sM}} \alpha_{Q,T_\iin} \ket{\eqref{tab:sim-pre-sandwich-M}} \enspace .
\end{align}
By analogy to the expression \eqref{eq:user-channel-post-sandwich} for $\ptr{\sK}{\rho_\mathrm{real}}$ derived in Section \ref{sec:user}, by the Pauli sandwich (Lemma \ref{lm:pauli-sandwich}) we have that the mixed state $\ptr{\sK}{\rho_\mathrm{sim}}$ from \eqref{eq:sim-channel} equals
\begin{align}
  \label{eq:sim-state}
  \ptr{\sK}{\rho_\mathrm{sim}} =
  \frac{1}{\abs{\mathscr{E}}\abs{\mathscr{P}}\abs{\mathscr{S}}\abs{\sB}\abs{\tilde\sB_\iin}\abs{\sS_\iin}}
  \sum_{E,P,S,a,s,T^\iin,T^\out,T^\smm,Q_{\tilde\sM}}
  \bra{a}\bra{s} E^* Q_{\tilde\sM} E \Ptr{\sA'\sE'}{\ket{Q_{\tilde\sM}'} \bra{Q_{\tilde\sM}'}} E^* Q_{\tilde\sM} E \ket{a}\ket{s} \enspace .
\end{align}
As mentioned in Section \ref{sec:user}, we will derive the expression \eqref{eq:sim-goal} for $\rho_\mathrm{sim}$, from which it will follow that the state is close to $\rho_\mathrm{real}$ in trace distance as desired.

\subsubsection{In the event of acceptance}
\label{sec:simulator:accept}

Just as in Section \ref{sec:user:accept}, consider the unnormalized mixed state obtained from the expression \eqref{eq:sim-state} for $\ptr{\sK}{\rho_\mathrm{sim}}$ by restricting attention to those terms in the summation with syndrome outcome $s=0$.
We can follow the argument of Section \ref{sec:user:accept} to see that this unnormalized mixed state is
\begin{align}
\label{eq:sim-R-acc}
  \underbrace{
  \frac{1}{\abs{\mathscr{E}}\abs{\mathscr{P}}\abs{\mathscr{S}}\abs{\sB}\abs{\tilde\sB_\iin}\abs{\sS_\iin}}
  \sum_{P,S,a,T^\iin,T^\out,T^\smm,Q_{\tilde\sM}}
  \sum_{E\in\mathscr{E}_\sX(Q_{\tilde\sM})}
  \bra{a} \Ptr{\sA'\sE'}{\ket{Q_{\tilde\sM}'} \bra{Q_{\tilde\sM}'}} \ket{a}
  }_{\textrm{show equal to $\ptr{\sK}{R_\mathrm{acc}}$}}
  + \Ptr{\sK}{[\varepsilon_\mathrm{sim}]}
\end{align}
for some choice of $[\varepsilon_\mathrm{sim}]$ with trace at most $\varepsilon$.
Henceforth we concentrate only on the indicated term of \eqref{eq:sim-R-acc}.
We will show that this term equals $\ptr{\sK}{R_\mathrm{acc}}$.

Because the syndrome measurement succeeded, our simulator reveals a new register $\sK$ containing a classical description $\ket{\hat S'}$ of a decryption key $\hat S'$.
(More about this key later in this section.)
The unnormalized pure state $\bra{a}\ket{Q_{\tilde\sM}'}$ plus register $\sK$ can be written
\begin{align}
  \bra{a}\ket{Q_{\tilde\sM}'}\ket{\hat S'} = \sum_{\lnot Q_{\tilde\sM}} \alpha_{Q,T^\iin} \ket{\eqref{tab:sim-post-decode}}
\end{align}
where the table \eqref{tab:sim-post-decode} is given by
\begin{align}
  \label{tab:sim-post-decode}
  \renewcommand{\arraystretch}{1.3}
  \begin{array}{|c||c|c|c|c|c|c|c|} \hline
  \sE' &&&&& \multirow{3}{*}{$U$} && \ket{0} \\
    \cline{1-5} \cline{7-8}
  \sA' &&&&&&& \ket{\psi}_{\sA} \\
    \cline{1-5} \cline{7-8}
  \sB_\out & Q_{\sB_\out} E^* S Q_{\tilde\sB_\out} T^\out Q_{\tilde\sB_\iin} & \pi_a(PT^\smm_E)_{\sB_\out} & E & \multirow{4}{*}{$\cont{U_a}$} && Q_{\sB_\iin} Q_\sB & \ket{\psi}_\sB \\
    \cline{1-4} \cline{6-8}
  \tilde \sA & Q_{\tilde\sA} & \pi_a(PT^\smm_E)_{\tilde\sA} & E &&&& \ket{0} \\
    \cline{1-4} \cline{6-8}
  \tilde \sE & Q_{\tilde\sE} & \pi_a(PT^\smm_E)_{\tilde\sE} & E &&&& \ket{0} \ket{\mathrm{off}} \\
    \cline{1-4} \cline{6-8}
  \tilde \sM & \multicolumn{2}{c|}{} & \bra{a} &&&& \ket{\mu} \\
    \hline
  \sW & Q_\sW &&&&&& \ket{\psi}_\sW \\
    \hline
  \sK &&&&&&& \ket{\hat S'} \\
    \hline
  \end{array}
\end{align}
Similar to Section \ref{sec:user}, (except now the control-bit switch is set to $\ket{\mathrm{off}}$) we may now employ the identity \eqref{eq:universal-a} from Section \ref{sec:universal} to get
\begin{align}
  \bra{a}\cont{U_a}\ket{\mu}\ket{\mathrm{off}}=\frac{1}{2^{r/2}}\cont{U}\ket{\mathrm{off}}=\frac{1}{2^{r/2}}I
\end{align}
so the state $\ket{\eqref{tab:sim-post-decode}}$ becomes
\begin{align}
  \label{tab:sim-post-magic}
  \ket{\eqref{tab:sim-post-decode}} =
  \frac{1}{2^{r/2}}
  \renewcommand{\arraystretch}{1.3}
  \begin{array}{|c||c|c|c|c|c|c|}
    \hline
  \sW & Q_\sW &&&&& \ket{\psi}_\sW \\
    \hline
  \sE' &&&& \multirow{3}{*}{$U$} && \ket{0} \\
    \cline{1-4} \cline{6-7}
  \sA' &&&&&& \ket{\psi}_\sA \\
    \cline{1-4} \cline{6-7}
  \sB_\out & Q_{\sB_\out} E^* S Q_{\tilde\sB_\out} T^\out Q_{\tilde\sB_\iin} & \pi_a(PT^\iin_E)_{\sB_\out} & E && Q_{\sB_\iin} Q_\sB & \ket{\psi}_\sB \\
    \hline
  \tilde \sA & Q_{\tilde\sA} & \pi_a(PT^\iin_E)_{\tilde\sA} & E &&& \ket{0} \\
    \hline
  \tilde \sE & Q_{\tilde\sE} & \pi_a(PT^\iin_E)_{\tilde\sE} & E &&& \ket{0} \\
    \hline
  \sK &&&&&& \ket{\hat S'} \\
    \hline
  \end{array}
\end{align}
Write
\begin{align}
  \ket{Q_{\tilde\sM}',\hat S'} \defeq \sum_{\lnot Q_{\tilde\sM}} \alpha_{Q,T^\iin} \ket{\eqref{tab:sim-post-magic}}
\end{align}
so that the highlighted term in \eqref{eq:sim-R-acc} becomes
\begin{align}
\label{eq:sim-N-accept}
  \frac{1}{\abs{\mathscr{E}}\abs{\mathscr{P}}\abs{\mathscr{S}}\abs{\sB}\abs{\tilde\sB_\iin}\abs{\sS_\iin}2^r}
  \sum_{P,S,a,T^\iin,T^\out,T^\smm,Q_{\tilde\sM}}
  \sum_{E\in\mathscr{E}_\sX(Q_{\tilde\sM})}
  \Ptr{\sA'\sE'}{ \ket{Q_{\tilde\sM}',\hat S'}\bra{Q_{\tilde\sM}',\hat S'} } \enspace .
\end{align}
By repeating the analysis of Section \ref{sec:user:key} it is easy to see that the final decryption Pauli $\hat S'$ produced by the simulator is the logical Pauli induced on the data register by code $E$ and Pauli
\begin{align}
  ST^\out\pi_a(PT^\smm_E)_{\tilde\sB_\out} \enspace .
\end{align}
Finally, we claim that the expression \eqref{eq:sim-N-accept} equals $R_\mathrm{acc}$ from \eqref{eq:R-acc2}.
To justify this claim we observe the differences between the tables \eqref{tab:post-magic} and \eqref{tab:sim-post-magic}.
The only differences are
\begin{enumerate}

\item $T^\iin_E$ has been replaced with $T^\smm_E$ in the encryption Paulis for $\sB_\out,\tilde\sA,\tilde\sE$ and in the final decryption key $\hat S'$.

\item $U$ is applied to the auxiliary registers $\sA',\sE'$ that are discarded by the simulator instead of the registers $\tilde\sA,\tilde\sE$ in the environment's possession.
The sender's portion of the input state $\ket{\psi}_\sA$ is contained in register $\sA'$ instead of $\tilde\sA$.

\end{enumerate}
The first difference is eliminated by the uniformly random encryption Pauli $P$; the summation over $T^\smm$ in \eqref{eq:sim-N-accept} can be eliminated by a simple change of variable, summing over $PT^\iin_E T^\smm_E$ instead of $P$.
The second difference is trivial because the registers $\tilde\sA,\tilde\sE$ are encrypted anyway.
In particular, $\ket{\eqref{tab:post-magic}}$ can be obtained from $\ket{\eqref{tab:sim-post-magic}}$ by swapping $(\sA',\sE')$ for $(\tilde\sA,\tilde\sE)$ prior to encryption of those registers.
Since $\tilde\sA,\tilde\sE$ are encrypted, such a swap cannot be detected.

\subsubsection{In the event of rejection}

Just as in Section \ref{sec:user:reject}, consider the unnormalized mixed state obtained from the expression \eqref{eq:sim-state} for $\ptr{\sK}{\rho_\mathrm{sim}}$ by restricting attention to those terms in the summation with syndrome outcome $s\neq 0$.
We can follow the argument of Section \ref{sec:user:reject} to see that this unnormalized mixed state is
\begin{align}
  \label{eq:sim-R-rej}
  \frac{1}{\abs{\mathscr{E}}\abs{\mathscr{P}}\abs{\mathscr{S}}\abs{\sB}\abs{\tilde\sB_\iin}\abs{\sS_\iin}}
  \sum_{P,S,a,s\neq 0,T^\iin,T^\out,T^\smm,Q_{\tilde\sM}}
    \sum_{ E\in\mathscr{E}_{\sX\emptyset}(Q_{\tilde\sM}) }
    \bra{a}Q_{\tilde\sM,\ell(E)}\Ptr{\sA'\sE'}{\ket{Q_{\tilde\sM}'}\bra{Q_{\tilde\sM}'}}Q_{\tilde\sM,\ell(E)}^*\ket{a} \enspace .
\end{align}
We will show that the expression \eqref{eq:sim-R-rej} equals $\ptr{\sK}{R_\mathrm{rej}}$.

Because the syndrome measurement failed, the register $\sK$ revealed by the simulator is completely mixed and so we can ignore it for the rest of this analysis.
The unnormalized pure state $\bra{a}Q_{\tilde\sM,\ell(E)}\ket{Q_{\tilde\sM}'}$ can be written
\begin{align}
  \bra{a}Q_{\tilde\sM,\ell(E)}\ket{Q_{\tilde\sM}'} = \sum_{\lnot Q_{\tilde\sM}} \alpha_{Q,T^\iin} \ket{\eqref{tab:sim-post-decode-reject}}
\end{align}
where the table \eqref{tab:sim-post-decode-reject} is given by
\begin{align}
  \label{tab:sim-post-decode-reject}
  \renewcommand{\arraystretch}{1.3}
  \begin{array}{|c||c|c|c|c|c|c|c|c|} \hline
  \sE' &&&&&& \multirow{3}{*}{$U$} && \ket{0} \\
    \cline{1-6} \cline{8-9}
  \sA' &&&&&&&& \ket{\psi}_{\sA} \\
    \cline{1-6} \cline{8-9}
  \sB_\out && Q_{\sB_\out} E^* S Q_{\tilde\sB_\out} T^\out Q_{\tilde\sB_\iin} & \pi_a(PT^\smm_E)_{\sB_\out} & E & \multirow{4}{*}{$\cont{U_a}$} && Q_{\sB_\iin} Q_\sB & \ket{\psi}_\sB \\
    \cline{1-5} \cline{7-9}
  \tilde \sA && Q_{\tilde\sA} & \pi_a(PT^\smm_E)_{\tilde\sA} & E &&&& \ket{0} \\
    \cline{1-5} \cline{7-9}
  \tilde \sE && Q_{\tilde\sE} & \pi_a(PT^\smm_E)_{\tilde\sE} & E &&&& \ket{0} \ket{\mathrm{off}} \\
    \cline{1-5} \cline{7-9}
  \tilde \sM & \bra{a} & Q_{\tilde\sM} &&&&&& \ket{\mu} \\
    \hline
  \sW && Q_\sW &&&&&& \ket{\psi}_\sW \\
    \hline
  \end{array}
\end{align}
For each $Q_{\tilde\sM,\ell(E)}$ let $\ket{a'}=Q_{\tilde\sM,\ell(E)}\ket{a}$ denote the modified vector of magic state measurement results.
As in Section \ref{sec:user:reject} we may now employ the identity \eqref{eq:universal-a} from Section \ref{sec:universal} to get
\begin{align}
  \bra{a'}\cont{U_a}\ket{\mu} = \frac{1}{2^{r/2}}\cont{U_{a-a'}}
\end{align}
so the state $\ket{\eqref{tab:sim-post-decode-reject}}$ becomes
\begin{align}
  \label{tab:sim-post-magic-reject}
  \ket{\eqref{tab:sim-post-decode-reject}} =
  \frac{1}{2^{r/2}}
  \renewcommand{\arraystretch}{1.3}
  \begin{array}{|c||c|c|c|c|c|c|c|c|} \hline
  \sW & Q_\sW &&&&&& \ket{\psi}_\sW \\
    \hline
  \sE' &&&&& \multirow{3}{*}{$U$} && \ket{0} \\
    \cline{1-5} \cline{7-8}
  \sA' &&&&&&& \ket{\psi}_{\sA} \\
    \cline{1-5} \cline{7-8}
  \sB_\out & Q_{\sB_\out} E^* S Q_{\tilde\sB_\out} T^\out Q_{\tilde\sB_\iin} & \pi_a(PT^\smm_E)_{\sB_\out} & E & \multirow{3}{*}{$\cont{U_{a-a'}}$} && Q_{\sB_\iin} Q_\sB & \ket{\psi}_\sB \\
    \cline{1-4} \cline{6-8}
  \tilde \sA & Q_{\tilde\sA} & \pi_a(PT^\smm_E)_{\tilde\sA} & E &&&& \ket{0} \\
    \cline{1-4} \cline{6-8}
  \tilde \sE & Q_{\tilde\sE} & \pi_a(PT^\smm_E)_{\tilde\sE} & E &&&& \ket{0} \ket{\mathrm{off}} \\
    \hline
  \end{array}
\end{align}
In the case of rejection the sender has no information about the encryption Paulis $\pi_a(PT^\iin_E)_{\sB_\out}$, $\pi_a(PT^\iin_E)_{\tilde\sA}$, and $\pi_a(PT^\iin_E)_{\tilde\sE}$.
Thus, any unitaries that appear immediately before the encryption Paulis are superfluous---they can be removed without affecting the overall state.
In particular, we get the same state even after we replace $E\cont{U_{a-a'}}$ with the identity in the table \eqref{tab:sim-post-magic-reject}.
Furthermore, since $\sA',\sE'$ are discarded by the simulator, the same logic allows us to replace $U$ with the identity.
Finally, since registers $(\tilde\sA,\tilde\sE)$ remain encrypted, one could replace the state $\ket{\eqref{tab:sim-post-magic-reject}}$ with
\begin{align}
  \label{tab:sim-post-magic-reject-no-U}
  \ket{\eqref{tab:sim-post-decode-reject}} =
  \frac{1}{2^{r/2}}
  \renewcommand{\arraystretch}{1.3}
  \begin{array}{|c||c|c|c|c|c|} \hline
  \sE' &&&& \ket{0} \\
    \hline
  \sA' &&&& \ket{\psi}_{\sA} \\
    \hline
  \sB_\out & Q_{\sB_\out} E^* S Q_{\tilde\sB_\out} T^\out Q_{\tilde\sB_\iin} & \pi_a(PT^\smm_E)_{\sB_\out} & Q_{\sB_\iin} Q_\sB & \ket{\psi}_\sB \\
    \hline
  \tilde \sA & Q_{\tilde\sA} & \pi_a(PT^\smm_E)_{\tilde\sA} && \ket{\mathrm{anything}} \\
    \hline
  \tilde \sE & Q_{\tilde\sE} & \pi_a(PT^\smm_E)_{\tilde\sE} && \ket{\mathrm{anything}} \\
    \hline
  \sW & Q_\sW &&& \ket{\psi}_\sW \\
    \hline
  \end{array}
\end{align}
Summing over $PT^\iin_E T^\smm_E$ instead of $P$ and discarding registers $(\sA',\sE')$ we see that the above state can be interchanged with $\ket{\eqref{tab:post-magic-reject-no-U}}$ from Section \ref{sec:user:reject}.
The desired expression for $R_\mathrm{rej}$ follows.

\subsection{Result}

We have thus shown that, no environment of the given in Section~\ref{sec:receiver:arbitrary} can distinguish $\rho_{\mathrm{real}}$, as calculated in Section~\ref{sec:user}, from the output $\rho_{\mathrm{sim}}$, as calculated in Section~\ref{sec:sim:analysis}.  Moreover, we argued in Section~\ref{sec:receiver:arbitrary} that any arbitrary environment is equivalent to an environment of the stated form.

This yields the proof of Theorem~\ref{thm:main-quantum}, that the protocol described in Sections~\ref{sec:qotp-spec} and \ref{sec:qotp-spec-receiver} is statistically quantum-UC-secure realization of $\cFunc{\Phi}{OTP}$ in the case of a corrupt user, in the $\cFunc{}{BR-OTP}$-hybrid model.

By combining this with the result that $\cFunc{}{BR-OTP}$ can be realized statistically UC-secure in the $\cFunc{}{OTM}$-hybrid model (Corollary~\ref{cor:reactive-COTP}), the quantum lifting theorem~\cite{U10}, and the quantum UC transitivity lemma (Lemma~\ref{lem:R-transitive}), we achieve the central result of the paper, Theorem~\ref{thm:main}: a protocol for non-interactive statistically quantum-UC-secure (in the case of an honest sender and potentially corrupt receiver) one-time programs, assuming secure OTM tokens (i.e., in the $\cFunc{}{OTM}$-hybrid model).
\qed

}

\ifthenelse{\equal{\compileACM}{0}}{
\section{UC security of delegating quantum computations}
}
{ 
\section{UC security of delegating \\ quantum computations}
}
\label{sec:UC-sec-of-ABE10}


We show in this section that our main proof technique can also be used to establish the statistical quantum-UC-security of a family of protocols for delegating quantum computations, closely related to the protocol of Aharonov~\etal~\cite{AharonovBE10}.
Originally studied in the context of \emph{quantum interactive proof systems}, the protocol of Aharonov~\etal~was not originally shown to be secure according to any rigorous cryptographic security definition. \looseness=-1

We generalize the protocol of Aharonov~\etal~to support delegated quantum computation (in contrast to only deciding membership in a language) by making two minor modifications.  
First we instantiate the protocol using any encode-encrypt quantum authentication scheme that admits computing on authenticated data (such as the trap authentication scheme or the signed polynomial scheme as used by Aharonov~\etal).
Analogously to our main protocol, we also introduce as an aid in the proof a control-bit so that the circuit being implemented is a controlled-unitary. 

The ideal functionality we achieve is described in Functionality~\ref{ideal-funct:QOTP-delegated}. Following~\cite{AharonovBE10}, we describe the functionality in terms of a \emph{prover} and \emph{verifier}.

\begin{functionality}
\caption{Ideal functionality $\cFunc{\Phi}{delegated}$ for a quantum channel  $\Phi : \sA \rightarrow \sC$.
\label{ideal-funct:QOTP-delegated}}
\begin{enumerate}
\item  \textbf{Create:} Upon input register $\sA$ from the
verifier, send \verb"create" to the prover and store the contents of register~$\sA$.
\item \textbf{Execute:} Upon input \verb"execute" from the prover, evaluate $\Phi$ on register $\sA$, and send the contents of the output register $\sC$ to the verifier.
\end{enumerate}
\end{functionality}

\begin{theorem}\label{thm:UC-ABOE10}Let $\Phi$ be polynomial-time quantum computable
functionality. Then there exists an efficient, quantum
interactive protocol which  statistically
quantum-UC-emulates  $\cFunc{\Phi}{delegated}$ in the case of a corrupt prover, in
the plain model, and where the only quantum power required of the verifier is to encode  the input and auxiliary quantum registers,  and to decode the output. In particular, all the interaction is classical except for the first and last messages.
\end{theorem}

The proof of Theorem~\ref{thm:UC-ABOE10} follows as a special case of our main possibility result.
In the case of a general~$\Phi$, the registers that the verifier prepares in Theorem~\ref{thm:UC-ABOE10} are polynomial-size in the security parameter.
In the interactive proof scenario of Aharonov~\etal, the input to $\Phi$ is the all-$\ket{0}$ product state, the output is a single classical bit, and it suffices to implement $\cFunc{\Phi}{delegated}$ with only constant security.
Given these assumptions, the only quantum power required of the verifier is the ability to prepare constant-sized quantum registers in the first round.

\begin{proof} [Proof sketch.]
We view the verifier as the sender in the QOTP, but since the verifier has the input, and receives the output, she can encode and decode these herself, so we do away with the necessity of the encoding and decoding gadgets. Also, classical interaction is permitted, so we replace the BR-OTP with interaction. The simulator and proof are the same.
\end{proof}

\appendix
\appendixpage
\addappheadtotoc


\section{One-time programs for classical, bounded reactive functionalities ($\cFunc{}{BR-OTP}$)}

\label{sec:Appendix-BRCOTPs}

\begin{functionality}
\caption{\label{ideal-funct:b-COTP} Ideal functionality
$\cFunc{g_1,\ldots ,g_\ell}{BR-OTP}$ for a bounded,
sender-oblivious reactive functionality. Here, $g_1(a, b_1)\mapsto
(m_1, s_1)$ and $g_i(b_i, s_{i-1}) \mapsto (m_i, s_i)$ ($i=2,\ldots,
\ell$) are classical functions; $a$~represents the sender's input,
$b_i$ represents the receiver's input for round $i$ (which can
depend on the input-output behaviour of previous rounds), $s_i$
represents the internal state after round $i$, and~$m_i$ is the
message returned to the receiver after round~$i$. We assume
$s_\ell=\bot$.}
\begin{enumerate}
\item \textbf{Create:} Upon input  $a$ from the sender,
send \verb"create" to the recipient and store $a$.
\item \textbf{Execute:} Upon input $(i, b_i)$ where $i \in \{1,\ldots ,\ell\}$
from the recipient, do the following:
\begin{enumerate}
\item For $j=1,\ldots ,i-1$, if $g_j$ has not been evaluated, abort. If
$g_i$ has already been evaluated, abort.
\item Compute $g_i(b_i, s_{i-1})= (m_i, s_i)$ (if $i=1$, compute $g_1(a,b_1)= (m_1,
s_1))$.
\item Output $m_i$ to the receiver. Store $s_i$ and note that $g_i$ has been
evaluated.
\item If $m_\ell$ has just been output, delete any trace of this instance.
\end{enumerate}
\end{enumerate}
\end{functionality}

In this section, we use standard techniques to   extend Theorem~\ref{thm:COTP}  to sender-oblivious, polynomial-time
computable, \emph{bounded reactive} classical two-party functionalities (given
as Functionality~\ref{ideal-funct:b-COTP}). We achieve this result in a
straightforward way. In fact, Goyal~\etal~\cite[p.~40]{GISVW10} suggest
that the result below would follow from their work. This appendix provides all
the details, using some of the techniques of Goyal~\etal~\cite{GISVW10}.

A \emph{message authentication code (MAC)} is a pair of algorithms $(MAC, VF)$, where $MAC_{k}(m)$ constructs a tag $\sigma$ for a message $m$ under a key $k$, and $VF_{k}(m, \sigma)$ returns 1 if $\sigma$ is a valid tag for $m$ under key $k$.  A MAC is an \emph{unconditional one-time secure MAC} if
\begin{equation}
 \Pr \left[ VF_{k}(m, \sigma) = 1 : k \getsr \{0,1\}^{\kappa}, (m, \sigma) \gets \cA() \right]
\end{equation}
 is negligible in a security parameter $\kappa$ for all probabilistic algorithms $\cA$.  An example of such a MAC is as follows.  The key $k$ is a pair of $\kappa$-length binary strings $(a, b)$.  The tag for a message $m \in \{0,1\}^{\kappa}$ is $MAC_{(a,b)}(m) = a \cdot m + b$, where all operations are in $GF(2^{\kappa})$.

The basic idea of the construction (see Protocol~\ref{prot:R-OTP}) is
to create a one-time program for each next-message function~$g_i$ in
the ideal functionality $\cFunc{g_1,\ldots
g_\ell}{BR-OTP}$. For this, we need a way to guarantee that
the receiver executes the COTPs in the correct order, and also let
the receiver pass the functionality's state information from one
COTP to another without revealing this information to the receiver.
Both of these goals can be achieved via a standard authentication
and encryption mechanism: we define a family of functions $f_i$ that
are based on $g_i$ and output an unconditionally secure
authenticated encryption of the sender's internal state at the end
of round~$i$; this information should be supplied by the receiver as
an additional input for $f_{i+1}$. If the authentication fails, then
$f_i$ outputs~$\bot$.

\begin{protocol} \caption{Protocol for a bounded,
sender-oblivious reactive functionality, $\cFunc{g_1,\ldots
g_\ell}{BR-OTP}$ in the $\cFunc{f}{COTP}$-hybrid
model.} \label{prot:R-OTP}
\begin{enumerate}

\item \textbf{Key generation.}
For $1 \leq i \leq \ell-1$, the $\text{sender}$ randomly chooses keys $k_1^i$ and $k_0^i$.

\item \textbf{Definition of functions.} For each $1 \leq i \leq \ell$,
given $g_{i}$, define $f_i$ as follows.
\begin{enumerate}

\item \textbf{Inputs of $f_{i}$.} 
	\begin{itemize} 
	\item For $i = 1$, function $f_1$ takes as input from the sender a string~$a$, a key  $k_1^1$ for message authentication scheme $MAC(\cdot)$, and  key $k_0^1$  which has the same length as $s_1$. Function $f_{1}$ takes  as input from the receiver a bit-string $b_1$.
	\item For $i > 1$, function $f_i$ takes as input from the sender  keys $k_1^{i-1}$ and $k_1^i$ for message authentication scheme  $MAC(\cdot)$ and  keys $k_0^{i-1}$ and  $k_0^i$  which have the same length as $s_{i-1}$ and $s_i$, respectively. Function $f_{i}$ takes  as input from the receiver bit-strings $(b_i, c_0^{i-1}, c_1^{i-1})$.
	\end{itemize}

\item \textbf{Computation of $f_{i}$.}
	\begin{itemize}
	\item If $i=1$, compute $g_1(a, b_1)$ to obtain the message $m_1$ as well as internal state $s_1$.
	\item If $i >1$, check that $VF_{{k_1}^{i-1}}(c_0^{i-1}, c_1^{i-1})=1$. If not, output~$\bot$. Else, set $s_{i-1} = c_0^{i-1} \oplus k_0^{i-1}$. Compute $g_i(b_i, s_{i-1})$ to obtain the $i^\text{th}$ message $m_i$ as well as internal state $s_i$.
	\end{itemize}

\item \textbf{Outputs of $f_{i}$.}
	\begin{itemize}
	\item If $i < \ell$, output $(m_i, s_i \oplus k_0^i, MAC_{k_1^i}(s_i \oplus k_0^i))$.
	\item If $i=\ell$, output $(m_i, \bot)$.
	\end{itemize}
	
\end{enumerate}

\item \textbf{COTP construction.}
The sender uses $\cFunc{f_i}{COTP}$ to create one-time programs for $f_1, \ldots, f_\ell$.

\item \textbf{Evaluation.} The receiver evaluates $g_1, \ldots, g_{\ell}$ by doing the following:
	\begin{itemize}
	\item If $i=1$, run $\cFunc{f_{1}}{COTP}$ on input $b_{1}$ and obtain $(m_{1}, c_{0}^{1},c_{1}^{1})$.
	\item If $i>1$, run $\cFunc{f_{i}}{COTP}$ on input $(b_{i},c_{0}^{i-1},c_{1}^{i-1})$ and obtain $(m_{i}, c_{0}^{i},c_{1}^{i})$.
	\end{itemize}
\end{enumerate}
\end{protocol}

\begin{theorem}
\label{thm:r-COTP-COTP-hybrid} Protocol~\ref{prot:R-OTP} statistically classical-UC emulates $\cFunc{g_1,\ldots ,g_\ell}{BR-OTP}$ in the $\cFunc{f}{COTP}$-hybrid model, assuming $MAC$ is an unconditional one-time secure MAC.
\end{theorem}

By Theorem~\ref{thm:COTP}, there exists a non-interactive protocol
$\rho$ that classical-UC-emulates $\cFunc{}{COTP}$ in the
$\cFunc{}{OTM}$-hybrid model. Thus we have Corollary~\ref{cor:reactive-COTP} as given in the main text.

We finish this section with a proof of Theorem~\ref{thm:r-COTP-COTP-hybrid}.

\begin{proof}[Proof of Theorem~\ref{thm:r-COTP-COTP-hybrid}]
The proof proceeds by considering security against either a malicious sender or a malicious receiver.

\noindent\textbf{Security against a malicious sender.}
 Let $\cA$ be an adversary corrupting the sender in
 Protocol~\ref{prot:R-OTP}. We construct a simulator
 $\text{Sim}_{\cA}$ for $\cA$.

 The simulator is defined as follows: execute $\cA$ on input
 given by $\cZ$. This execution determines inputs for $\cF_{f_i}^{\text{COTP}}$, and in
 particular, it determines the sender's input~$a$.

We say that the senders's keys are  \emph{consistent} if, for each
$i$, the keys $k_1^{i}$ and $k_0^{i-1}$ are self-consistent between
instantiations of $\cFunc{f_i}{COTP}$ and
$\cFunc{f_{i+1}}{COTP}$ (i.e., the ideal
functionalities are called with the corresponding same inputs). If the
sender's keys are consistent, then an execution of
 Protocol~\ref{prot:R-OTP} with an honest receiver will not abort.

Thus, if  $\text{Sim}_{\cA}$ detects that the sender's keys
are \emph{not} consistent, $\text{Sim}_{\cA}$ inputs
$\mathsf{ABRT}$ into the ideal functionality
$\cF_{g_1,\ldots g_\ell}^{\text{COTP}}$. Otherwise,
$\text{Sim}_{\cA}$ inputs~$a$ for the sender.

We thus have the following:
\begin{enumerate}
\item  protocol~\ref{prot:R-OTP} is non-interactive;
\item the probability of $\mathsf{ABRT}$ in the ideal model is independent of the receiver's input;
\item an honest sender will not cause the protocol to abort.
\end{enumerate}

Together, these imply that the real and ideal networks are perfectly indistinguishable.

\noindent\textbf{Security against a malicious receiver.}

Let $\cA$ be an adversary corrupting the receiver in
Protocol~\ref{prot:R-OTP}. We construct a simulator
$\text{Sim}_{\cA}$, given as Simulator~\ref{sim:BR-OTP}.

\begin{simulator} \caption{Simulator $\text{Sim}_{\cA}$, for the proof of Protocol~\ref{prot:R-OTP}
against a malicious receiver. \label{sim:BR-OTP}}
\begin{enumerate}
\item For $i=1, \ldots, \ell$, choose keys $k_i^0$ and $k_i^0$ as in
Protocol~\ref{prot:R-OTP}.
\item Start the execution of  $\cA$ on the input given by $\cZ$; set $i = 1$.
\item \label{step-call-i}
Execute  $\cA$ until a call to the ideal functionality
$\cF_{f_{i}}^{\text{COTP}}$ occurs. Do the following in
order to simulate interaction with the ideal functionality.
\begin{enumerate}
\item For $j=1,\ldots ,i-1$, if $\cF_{f_{j}}^{\text{COTP}}$ has not already been evaluated, abort.
\item Let $\sigma_i$ be $\cA$'s input into
$\cF_{f_{i}}^{\text{COTP}}$.
\begin{enumerate}
\item if $i=1$, forward $\sigma_1$ to $\cF_{g_1,\ldots, g_\ell}^{\text{COTP}}$ and receive as response $m_1$. Choose a random
state~$w$, and return $(m_1, w \oplus k_0^1, MAC_{k_1^1}(w \oplus
k_0^1))$ to $\cA$.
\item if $i > 1$, interpret $\sigma_i$ as $(b_i, c_0^{i-1}, c_1^{i-1})$. Check that
$VF_{{k_1}^{i-1}}(c_0^{i-1}, c_1^{i-1})=1$. If not, output~$\bot$.
Forward $b_i$ to $\cF_{g_1,\ldots, g_\ell}^{\text{COTP}}$ and
receive as response $m_i$. Choose a random state $w_i$, and return
$(m_i, w_i \oplus k_0^i, MAC_{k_1^i}(w_i \oplus k_0^i))$ to
$\cA$.
\end{enumerate}
\item If $i < \ell$, return to
step~\ref{step-call-i}. Otherwise, output $\cA$'s output.
\end{enumerate}
\end{enumerate}
\end{simulator}

The real and ideal networks are indistinguishable: in the real world, the
adversary can succeed in out-of-order querying with probability at
most the probability of generating a forged MAC, which is
negligible since the MAC is secure and $\cA$ can succeed in making an OTP accept an incorrect
internal state also with negligible probability.
\end{proof}


\section{Properties of encode-encrypt authentication schemes}
\label{appendix:encode-encrypt}

\subsection{Security against Pauli attacks implies security against general attacks}
\label{sec:encode-encrypt:general}

Parts of this section are reproduced from Aharonov, Ben-Or, and Eban \cite{AharonovBE10}.

\subsubsection{How Pauli attacks affect the state of the system}
\label{sec:auth:families}

A Pauli attack on a family $\mathscr{E}$ of codes has the following form.
\begin{enumerate}

\item
The data register $\sD$ is encoded under a uniformly random choice of code $E\in\mathscr{E}$ by preparing $(\sX,\sZ)$ in the $\ket{0}$ state and applying $E$ to $(\sD,\sX,\sZ)$.

\item
A malicious attacker applies a Pauli operator $Q$ to $(\sD,\sX,\sZ)$.

\item \label{it:auth-decode}
Data is decoded by applying the inverse operator $E^*$ to $(\sD,\sX,\sZ)$.

\item \label{it:auth-syndrome}
The syndrome registers $(\sX,\sZ)$ are measured in the computational basis.
A non-zero measurement result indicates cheating.

\end{enumerate}
The first three steps of this protocol induce a channel $\Psi_Q$ of the form
\[
  \Psi_Q : \sD \to (\sD,\sX,\sZ) : \rho \mapsto \frac{1}{\abs{\mathscr{E}}} \sum_{E\in\mathscr{E}} E^*QE(\rho\ot\ket{0}\bra{0})E^*Q^*E \enspace .
\]
Let $[X],[Z]$ denote the projectors onto the state $\ket{0}$ for syndrome registers $\sX,\sZ$, respectively.
The state of the data register $\sD$ after the syndrome measurement of step \ref{it:auth-syndrome} is
\begin{align}
  \Ptr{\sX\sZ}{[X][Z]\Psi_Q(\rho)} + \Ptr{\sX\sZ}{(I-[X][Z])\Psi_Q(\rho)} \enspace .
\end{align}
Let us examine the terms associated with each measurement outcome. 
For each Pauli $Q$ we define the following partition of $\mathscr{E}$:
\begin{center}
\begin{tabularx}{\textwidth}{lX}
  $\mathscr{E}_{\sX\sZ}(Q)$: & The set of codes $E\in\mathscr{E}$ for which $Q$ acts trivially on logical data and has no error syndrome.\\
  $\mathscr{E}_{\sX\sZ!}(Q)$: & The set of codes $E\in\mathscr{E}$ for which $Q$ acts \emph{non}-trivially on logical data but has no error syndrome.\\
  $\mathscr{E}_{\sX\sZ\emptyset}(Q)$: & The set of codes $E\in\mathscr{E}$ for which $Q$ has non-zero error syndrome.
\end{tabularx}
\end{center}
For each code $E\in\mathscr{E}$ let $Q_{\ell(E)}$ denote the logical Pauli induced by $Q,E$ and observe that if $E\in\mathscr{E}_{\sX\sZ}(Q)$ then $Q_{\ell(E)}=I$.
Then
\begin{align}
  \label{eq:Pauli-attack}
  \Ptr{\sX\sZ}{[X][Z]\Psi_Q(\rho)} &=
  \underbrace{
    \frac{1}{\abs{\mathscr{E}}}\sum_{E\in\mathscr{E}_{\sX\sZ }(Q)} Q_{\ell(E)}\rho Q_{\ell(E)}^*
  }_{ \frac{\abs{\mathscr{E}_{\sX\sZ}(Q)}}{{\abs{\mathscr{E}}}} \rho }
  +
  \underbrace{
    \frac{1}{\abs{\mathscr{E}}}\sum_{E\in\mathscr{E}_{\sX\sZ!}(Q)} Q_{\ell(E)}\rho Q_{\ell(E)}^*
  }_{[\varepsilon(Q)]}\\
  \label{eq:bad-Pauli-attack}
  \Ptr{\sX\sZ}{(I-[X][Z])\Psi_Q(\rho)} &=
  \frac{1}{\abs{\mathscr{E}}}\sum_{E\in\mathscr{E}_{\sX\sZ\emptyset}(Q)} Q_{\ell(E)}\rho Q_{\ell(E)}^* \enspace .
\end{align}
If $\mathscr{E}$ is $\epsilon$-secure against Pauli attacks (according to Definition \ref{def:secure-against-Pauli-attacks}) then by definition the term marked $[\varepsilon(Q)]$ in \eqref{eq:Pauli-attack} has trace at most $\epsilon$ and so it must be that $\abs{\mathscr{E}_{\sX\sZ!}(Q)}/\abs{\mathscr{E}}\leq\epsilon$ for all $Q$.

If $Q=I$ then $Q_{\ell(E)}=I$ for all $E\in\mathscr{E}$ and $\mathscr{E}_{\sX\sZ}(Q)=\mathscr{E}$.
In this case the terms $[\varepsilon(Q)]$ and \eqref{eq:bad-Pauli-attack} both vanish and the state of the system is simply $\rho$ as required for any authentication scheme.

\subsubsection{The Pauli sandwich}

As mentioned in Section \ref{sec:techniques-Q-auth}, encode-encrypt authentication schemes have the desirable property that the Pauli encryption breaks up an arbitrary attack into a probabilistic mixture of Pauli attacks.
This property hinges upon a key lemma that we call the \emph{Pauli sandwich};
it is a direct consequence of the \emph{Pauli twirl}~\cite{dankert2009exact}.
A succinct proof can be found in Aharonov, Ben-Or, and Eban~\cite{AharonovBE10}.

\begin{lemma}[Pauli sandwich (see \cite{AharonovBE10})]
\label{lm:pauli-sandwich}

  Let $W$ be an arbitrary operator acting on $n$-qubits and let $Q,Q'$ be Pauli operators acting on $n$-qubits.
  It holds that
  \begin{align}
    \frac{1}{4^n} \sum_{\textrm{Paulis $P$}} P^*QP W P^*Q'^*P =
    \left\{
    \begin{array}{ll}
      Q W Q^* & \textrm{if $Q=Q'$}\\
      0 & \textrm{otherwise}
    \end{array}
    \right.
  \end{align}

\end{lemma}

To see how the Pauli sandwich breaks up an arbitrary attack into a mixture of Pauli attacks let $U$ be any operator acting on a register of $n$-qubits, let $U=\sum_{\textrm{Paulis $Q$}} \alpha_Q Q$ be a decomposition of $U$ as a linear combination of Pauli operators, and consider the channel
\begin{align}
  \Psi : \rho \mapsto
  \frac{1}{4^n} \sum_{\textrm{Paulis $P$}} P^*UP \rho P^*U^*P
  = \frac{1}{4^n} \sum_{\textrm{Paulis $P$}} \sum_{\textrm{Paulis $Q,Q'$}} \alpha_Q\ol{\alpha_{Q'}} P^*QP \rho P^*Q'^*P
\end{align}
It follows immediately from the Pauli sandwich (Lemma \ref{lm:pauli-sandwich}) that
\begin{align}
  \Psi : \rho \mapsto \sum_{\textrm{Paulis $Q$}} \abs{\alpha_Q}^2 Q \rho Q^*
\end{align}
That is, if an adversary applies a unitary attack $U$ to a register encrypted with a uniformly random Pauli $P$ then the resulting state after decryption is indistinguishable from a simplified attack where the adversary instead applies a Pauli attack $Q$ chosen at random according to the distribution $\set{\abs{\alpha_Q}^2}$.

This observation holds for arbitrary operators $U$ (not just unitary operators) so it can easy be extended to arbitrary non-unitary channel attacks $\Phi$ by applying a similar analysis for each operator in a Kraus decomposition of $\Phi$.

\subsubsection{How general attacks affect the state of the system}
\label{sec:auth:security-arbitrary}

An arbitrary attack on an encode-encrypt authentication scheme based on a family $\mathscr{E}$ of codes has the following form.
\begin{enumerate}
\item
  A data register $\sD$ is authenticated by encoding it with a random code $E\in\mathscr{E}$ (implicitly introducing syndrome registers $\sX,\sZ$) and then encrypting the encoded $n$-qubit system $(\sD,\sX,\sZ)$ with a random Pauli $P$.
\item
  A malicious attacker applies a channel $\Phi$ to the $n$ qubits comprising the registers $(\sD,\sX,\sZ)$.
\item
  The system is decrypted by applying the Pauli $P^*$ and decoded by applying the inverse circuit $E^*$.
\item \label{it:attack:decode}
  The syndrome registers $(\sX,\sZ)$ are measured in the computational basis.
  A non-zero measurement result indicates cheating.
\end{enumerate}
As mentioned previously,
it suffices to restrict attention only to unitary attacks  $\Phi:X\mapsto UXU^*$;
security against arbitrary channels is recovered by applying the same analysis to each operator in a Kraus decomposition of $\Phi$.
The first three steps of this protocol induce a channel $\Psi_U$ of the form
\begin{align}
  \Psi_U : \sD &\to (\sD,\sX,\sZ) \\
  \rho &\mapsto \frac{1}{4^n\abs{\mathscr{E}}} \sum_{E\in\mathscr{E}} \sum_{\textrm{Paulis $P$}} E^* P^* U P E \Pa{\rho\ot\ket{0}\bra{0}} E^* P^* U^* P E
\end{align}
Let $U=\sum_{\textrm{Paulis $Q$}} \alpha_Q Q$ be a decomposition of $U$ into a linear combination of Paulis so that the channel $\Psi_U$ could equivalently be written
\begin{align}
  \Psi_U : \rho &\mapsto \frac{1}{4^n\abs{\mathscr{E}}} \sum_{E\in\mathscr{E}} \sum_{\textrm{Paulis $P$}} \sum_{\textrm{Paulis $Q,Q'$}} \alpha_Q\ol{\alpha_{Q'}} E^* P^* Q P E \Pa{\rho\ot\ket{0}\bra{0}} E^* P^* Q'^* P E \\
  &= \frac{1}{\abs{\mathscr{E}}} \sum_{E\in\mathscr{E}} \sum_{\textrm{Paulis $Q$}} \abs{\alpha_Q}^2 E^* Q E \Pa{\rho\ot\ket{0}\bra{0}} E^* Q^* E \\
  &= \sum_{\textrm{Paulis $Q$}} \abs{\alpha_Q}^2 \Psi_Q(\rho)
  \label{eq:authentication-double-sum}
\end{align}
with the first equality following form the Pauli sandwich (Lemma \ref{lm:pauli-sandwich}).

The state of the register $\sD$ after the syndrome measurement of step \ref{it:attack:decode} is
\begin{align}
  \Ptr{\sX\sZ}{[X][Z]\Psi_U(\rho)} + \Ptr{\sX\sZ}{(I-[X][Z])\Psi_U(\rho)}.
\end{align}
As before, let us examine the terms associated with each measurement outcome. 
We have
\begin{align}
  \Ptr{\sX\sZ}{[X][Z]\Psi_U(\rho)}
  &= \sum_{\textrm{Paulis $Q$}} \abs{\alpha_Q}^2 \Ptr{\sX\sZ}{[X][Z]\Psi_Q(\rho)} \\
  &= \sum_{\textrm{Paulis $Q$}} \abs{\alpha_Q}^2 
  \Pa{ \frac{\abs{\mathscr{E}_{\sX\sZ}(Q)}}{{\abs{\mathscr{E}}}} \rho + [\varepsilon(Q)] }
  \label{eq:auth-final-state}
\end{align}
where the second equality is from \eqref{eq:Pauli-attack}.
From the expression \eqref{eq:auth-final-state} one can see that the probability of acceptance is a nondecreasing function of the modulus squared of $\alpha_I$.
In particular, the more weight $U$ places on the identity in its Pauli decomposition, the better the chance that the attack $U$ is successful (and the least likely it is to have any effect on the state).

For completeness we explicitly write the term associated with outcome $I-[X][Z]$:
\begin{align}
  \Ptr{\sX\sZ}{(I-[X][Z])\Psi_U(\rho)} &= \sum_{\textrm{Paulis $Q$}} \abs{\alpha_Q}^2 \Ptr{\sX\sZ}{(I-[X][Z])\Psi_Q(\rho)} \\
  &= \sum_{\textrm{Paulis $Q$}} \abs{\alpha_Q}^2
  \Pa{
  \frac{1}{\abs{\mathscr{E}}}\sum_{E\in\mathscr{E}_{\sX\sZ\emptyset}(Q)} Q_{\ell(E)}\rho Q_{\ell(E)}^* }
\end{align}

\subsection{Measure-then-decode equals decode-then-measure for CSS codes}
\label{sec:auth:measure-decode}

CSS codes have the property that measurement of logical data in the computational basis can be implemented by bitwise (transversal) measurement of the physical qubits in the encoding, followed by a purely classical decoding process.

More concretely, suppose $E$ is a CSS code that encodes one logical qubit into $n$ physical qubits.
For each value of the bit $a\in\set{0,1}$ there exists a set $D_E(a)$ of $n$-bit strings such that the encoded logical basis state $E\ket{a}$ is an equal superposition of strings in $D_E(a)$:
\begin{align}
  E\ket{a} = \frac{1}{\sqrt{\abs{D_E(a)}}} \sum_{b\in D_E(a)} \ket{b}
\end{align}
Bitwise measurement of $E\ket{a}$ in the computational basis yields a string $b\in D_E(a)$ selected uniformly at random.
The sets $D_E(0),D_E(1)$ are disjoint so the logical measurement result $a$ can be deduced by identifying the set $D_E(a)$ from which $b$ was drawn.
Moreover, given an arbitrary $n$-bit string $c$ there is an efficient classical algorithm that computes the function
\begin{align}
  \operatorname{Decode}_E : c \mapsto (a,s)
\end{align}
where $s$ is an $n$-bit syndrome string with the property that $c\oplus s\in D_E(a)$.

Suppose instead that our logical qubit is authenticated under an encode-encrypt scheme based on a family $\mathscr{E}$ of CSS codes and fix a choice of key $(E,P)$ indicating the CSS code $E$ and Pauli encryption $P$.
In this case measurement of the logical qubit can still be implemented by bitwise measurement of the authenticated qubit.
The only difference is that the $X$-portion $P|_X$ of the Pauli key $P$ indicates a mask to be applied before the classical decoding process.
This simple process could be drawn as follows.
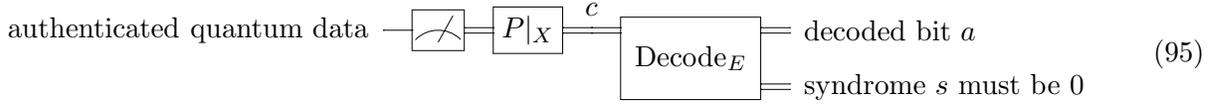
\begin{gather}
\begin{aligned}
\label{circuit:measure-decode}
\Qcircuit @C=1em @R=0.7em {
  \lstick{\textrm{authenticated quantum data}} & \meter & \cgate{P|_X} & \ustick{c} \cw & \multicgate{1}{\textrm{Decode}_E} & \rstick{\textrm{decoded bit $a$}} \cw \\
  & & & & \pureghost{\textrm{Decode}_E} & \rstick{\textrm{syndrome $s$ must be $0$}} \cw
}
\end{aligned}
\end{gather}
The analysis later in this paper is better facilitated by writing this process in a different form.
We claim that the above process \eqref{circuit:measure-decode} of measure-then-decode on \emph{classical} data is equivalent to the following process of decode-then-measure on \emph{quantum} data:
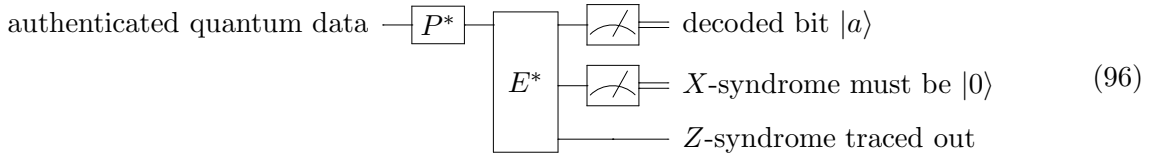
\begin{gather}
\begin{aligned}
\label{circuit:decode-measure}
\Qcircuit @C=1em @R=0.7em {
  \lstick{\textrm{authenticated quantum data}} & \gate{P^*} & \multigate{2}{E^*} & \meter & \rstick{\textrm{decoded bit $\ket{a}$}} \cw \\
  & & \pureghost{E^*} & \meter & \rstick{\textrm{$X$-syndrome must be $\ket{0}$}} \cw \\
  & & \pureghost{E^*} & \qw & \rstick{\textrm{$Z$-syndrome traced out}} \qw
}
\end{aligned}
\end{gather}

To prove the claim it suffices to show that the decode-then-measure circuit \eqref{circuit:decode-measure} accepts an encrypted $n$-qubit computational basis state $P\ket{c}$ if and only if $c\in D_E(a)$ for some bit $a$.
It is clear that if $\ket{c}$ does not meet this condition then it will be rejected by \eqref{circuit:decode-measure}, so suppose that it does.
It is a property of every CSS code that $\ket{c}$ can be written as a superposition of states of the form $QE\ket{a}$ where $Q$ is a purely $Z$-Pauli.
That is,
\begin{align}
  \ket{c} = \sum_\textrm{$Z$-Paulis $Q$} \alpha_{c,Q} QE\ket{a}
\end{align}
for some complex amplitudes $\set{\alpha_{c,Q}}$.
Thus, the encrypted basis state $P\ket{c}$ is accepted by \eqref{circuit:decode-measure} and decodes to $\ket{a}$ as desired.

\subsubsection{How general attacks affect the decode-then-measure process}
\label{sec:auth:decode-measure}

In Section \ref{sec:auth:security-arbitrary} we analyzed the effect of an arbitrary attack on the state of a system protected by an encode-encrypt authentication scheme.
In this section we are interested in the effect of such an attack on the decode-then-measure process when the scheme is based on a family $\mathscr{E}$ of CSS codes.
The protocol for such an attack is identical to the protocol of Section \ref{sec:auth:security-arbitrary} except that step \ref{it:attack:decode} is replaced with the following.
\begin{enumerate}
\setcounter{enumi}{3}
\item \label{it:decode-then-measure-measurement}
  The data register $\sD$ is measured in the computational basis.
  The syndrome register $\sX$ is measured in the computational basis;
  a non-zero measurement result indicates cheating.
  The syndrome register $\sZ$ is discarded.
\end{enumerate}
As usual, it suffices to restrict attention only to unitary channels $\Phi:X\mapsto UXU^*$ for some unitary $U$ with Pauli decomposition $U=\sum_{\textrm{Paulis $Q$}} \alpha_Q Q$.
The state of the register $\sD$ after the measurements step \ref{it:decode-then-measure-measurement} is
\begin{align}
  \label{eq:decode-measure-attack}
  \sum_{\textrm{Paulis $Q$}} \abs{\alpha_Q}^2
  \sum_{a\in\set{0,1}} \ket{a}\bra{a} \Pa{
    \sum_{M\in\set{[X],I-[X]}}
    \Ptr{\sX\sZ}{M\Psi_Q(\rho)}
  } \ket{a}\bra{a}.
\end{align}
As before, let us examine the terms associated with each measurement outcome $\set{[X],I-[X]}$.
The analysis here is slightly more complicated than in previous sections, owing to the fact that only the $X$-syndrome is verified in the decode-then-measure protocol.
For each Pauli $Q$ we define the following partition of $\mathscr{E}$:
\begin{center}
\begin{tabularx}{\textwidth}{lX}
  $\mathscr{E}_{\sX}(Q)$: & The set of codes $E\in\mathscr{E}$ for which $Q_{\ell(E)}$ is a purely $Z$-Pauli and $Q$ has no $X$-error syndrome.\\
  $\mathscr{E}_{\sX!}(Q)$: & The set of codes $E\in\mathscr{E}$ for which $Q_{\ell(E)}$ has nontrivial $X$-component but $Q$ has no $X$-error syndrome.\\
  $\mathscr{E}_{\sX\emptyset}(Q)$: & The set of codes $E\in\mathscr{E}$ for which $Q_{\ell(E)}$ has non-zero $X$-error syndrome.
\end{tabularx}
\end{center}
Then for each Pauli $Q$ we have
\begin{align}
  \label{eq:Pauli-attack-measure}
  \Ptr{\sX\sZ}{[X]\Psi_Q(\rho)} &=
    \frac{1}{\abs{\mathscr{E}}}\sum_{E\in\mathscr{E}_{\sX }(Q)} Q_{\ell(E)}\rho Q_{\ell(E)}^*
  +
  \underbrace{
    \frac{1}{\abs{\mathscr{E}}}\sum_{E\in\mathscr{E}_{\sX!}(Q)} Q_{\ell(E)}\rho Q_{\ell(E)}^*
  }_{[\varepsilon_\sX(Q)]}\\
  \label{eq:-bad-Pauli-attack-measure}
  \Ptr{\sX\sZ}{(I-[X])\Psi_Q(\rho)} &=
  \frac{1}{\abs{\mathscr{E}}}\sum_{E\in\mathscr{E}_{\sX\emptyset}(Q)} Q_{\ell(E)}\rho Q_{\ell(E)}^* \enspace .
\end{align}
It is a property of every CSS code $E$ that if $Q$ is a purely $X$- or $Z$-Pauli then the logical Pauli $Q_{\ell(E)}$ induced by $E$ is also a purely $X$- or $Z$-Pauli.
Thus, if $\mathscr{E}$ is $\epsilon$-secure against Pauli attacks and if $Q_{\ell(E)}$ has nontrivial $X$-portion then $Q$ must have nontrivial $X$-syndrome with probability at most $\varepsilon$ taken over the choice of code $E\in\mathscr{E}$.
In particular, the term marked $[\varepsilon_\sX(Q)]$ in \eqref{eq:Pauli-attack-measure} has trace at most $\epsilon$ and it must be that $\abs{\mathscr{E}_{\sX!}(Q)}/\abs{\mathscr{E}}\leq\epsilon$.

On the other hand, if $Q_{\ell(E)}$ is a purely $Z$-Pauli (which is always the case when $E\in\mathscr{E}_\sX(Q)$) then
\begin{align}
  \ket{a}\bra{a}Q_{\ell(E)}\rho Q_{\ell(E)}^*\ket{a}\bra{a} = \ket{a}\bra{a}\rho\ket{a}\bra{a}
\end{align}
for each $a\in\set{0,1}$.
It follows that the term 
in \eqref{eq:decode-measure-attack} associated with outcome $[X]$ is equal to
\begin{align}
  \label{eq:measure-accept}
  \sum_{\textrm{Paulis $Q$}} \abs{\alpha_Q}^2
  \sum_{a\in\set{0,1}} \ket{a}\bra{a} \Pa{\frac{\abs{\mathscr{E}_{\sX}(Q)}}{\abs{\mathscr{E}}}\rho + [\varepsilon_\sX(Q)]} \ket{a}\bra{a}.
\end{align}
Moreover, if $Q$ is a purely $Z$-Pauli then for every $E\in\mathscr{E}$ it holds that $Q_{\ell(E)}$ is also a purely $Z$-Pauli and that $Q$ has no $X$-syndrome under $E$.
That is, $\mathscr{E}_\sX(Q)=\mathscr{E}$ and the terms $[\varepsilon_\sX(Q)]$ and \eqref{eq:-bad-Pauli-attack-measure} both vanish.
Thus, one can see from the expression \eqref{eq:measure-accept} that the probability of acceptance
is a nondecreasing function of $\sum_\textrm{$Z$-Paulis $Q$} \abs{\alpha_Q}^2$.
In particular, the more weight $U$ places on purely $Z$-Paulis in its Pauli decomposition, the better the chance that the attack $U$ is successful (and the least likely it is to have any effect on the measured state).

For completeness we explicitly write the term in \eqref{eq:decode-measure-attack} associated with outcome $I - [X]$:
\begin{align}
  \sum_{\textrm{Paulis $Q$}} \abs{\alpha_Q}^2
  \sum_{a\in\set{0,1}} \ket{a}\bra{a} \Pa{
    \frac{1}{\abs{\mathscr{E}}}\sum_{E\in\mathscr{E}_{\sX\emptyset}(Q)} Q_{\ell(E)}\rho Q_{\ell(E)}^* 
  } \ket{a}\bra{a}.
\end{align}


\section{Analysis of teleportation}
\label{appendix:teleportation}

Suppose that a pair of $n$-qubit registers $(\mathsf{In},\mathsf{Out})$ is prepared in a ``teleport-through-$C$'' state for some $n$-qubit unitary $C$.
Such a state is constructed by preparing $n$ copies of the entangled state $(\ket{00}+\ket{11})/\sqrt{2}$ in the registers $(\mathsf{In},\mathsf{Out})$ and then applying $C$ to register $\mathsf{Out}$.
Suppose further that an $n$-qubit register $\sD$ is prepared in an arbitrary pure state $\ket{\psi}$ with the intention that this state be teleported through $C$ to qubit $\mathsf{Out}$ by way of a standard Bell measurement on $(\sD,\mathsf{In})$.

The Bell measurement is implemented in the usual way by applying an $n$-fold Bell rotation $B$ (each composed of $\cnot$ and $H$ gates) to the registers $(\sD,\mathsf{In})$ followed by a measurement of those two registers in the computational basis.
After this measurement the registers $(\sD,\mathsf{In})$ are in the classical basis state $\ket{T}$ indicating a uniformly random Pauli correction $T$.
Conditioned on this measurement outcome, it is easily seen that the pure state of register $\mathsf{Out}$ is $CT\ket{\psi}$.

In the basic teleportation protocol $C$ is the identity and so the state $\ket{\psi}$ can be recovered by applying the appropriate Pauli correction $T$, thus ``teleporting'' the state $\ket{\psi}$ from register $\sD$ to register $\mathsf{Out}$.
If $C$ is a Clifford circuit then the state $C\ket{\psi}$ can be recovered by applying an appropriate Pauli correction $T_C$, thus ``teleporting'' the state $\ket{\psi}$ from register $\sD$ ``through $C$'' and into register $\mathsf{Out}$.

\subsection{Teleportation under attack}
\label{sec:teleportation:attack}

Now suppose that all three registers $(\sD,\mathsf{In},\mathsf{Out})$ are passed through a channel $\Phi$ prior to the Bell measurement.
One could think of $\Phi$ as noise, or perhaps a malicious attack on the registers.
How does $\Phi$ affect the teleportation?
In particular, what is the state of the system after the Bell measurement is complete?

Let's start with the special case where $\Phi$ is a product unitary $U=U_\sD\ot U_\mathsf{In}\ot U_\mathsf{Out}$.
In this case the pure state of the entire system after the Bell rotation is easily seen to be
\begin{align}
  \label{eq:teleport-product}
  \frac{1}{2^{n/2}}\sum_{\textrm{Paulis $T$}} \ket{T} \ot U_\mathsf{Out} C U_\mathsf{In}^\trans T U_\sD \ket{\psi}.
\end{align}
Measurement of $(\sD,\mathsf{In})$ then yields a uniformly random outcome $T$ leaving register $\mathsf{Out}$ in the pure state
\begin{align}
  U_\mathsf{Out} C U_\mathsf{In}^\trans T U_\sD \ket{\psi}.
\end{align}
The mapping \eqref{eq:teleport-product} is linear in each of $U_\sD,U_\mathsf{In},U_\mathsf{Out}$ so it can be used to deduce the effect of arbitrary, possibly non-product unitaries $U$ on the state of the registers.
This is accomplished by decomposing $U$ into a linear combination of product unitaries and then applying the above identity to each term in that decomposition.
For example, let
\begin{align}
  \label{eq:U-Paulis}
  U = \sum_{\textrm{Paulis $P$}} \alpha_P P
\end{align}
be a decomposition of $U$ into a linear combination of Pauli operators of the form $P=P_\sD\ot P_\mathsf{In}\ot P_\mathsf{Out}$.
By the above analysis it holds that the pure state of the entire system after the Bell rotation is given by
\begin{align}
  \label{eq:teleport-unitary}
  \frac{1}{2^{n/2}}\sum_{\textrm{Paulis $T$}} \ket{T} \ot \Pa{\sum_{\textrm{Paulis $P$}} \alpha_P P_\mathsf{Out} C P_\mathsf{In}^\trans T P_\sD}\ket{\psi}.
\end{align}
Measurement of $(\sD,\mathsf{In})$ yields an outcome $T$ leaving register $\mathsf{Out}$ in the unnormalized pure state
\begin{align}
  \Pa{\sum_{\textrm{Paulis $P$}} \alpha_P P_\mathsf{Out} C P_\mathsf{In}^\trans T P_\sD}\ket{\psi}.
\end{align}
The distribution of measurement outcomes $T$ obtained by measurement of $(\sD,\mathsf{In})$ need not be uniform, owing to the potential for interference in the sum over amplitudes $\alpha_P$.

This analysis applies even to non-unitary operators $U$.
Thus, a similar expression can be derived for arbitrary, possibly non-unitary channels $\Phi$ by applying the above analysis to each individual Kraus operator in a Kraus decomposition of $\Phi$.

\subsection{Teleportation under attack, tabular analysis}
\label{sec:tabular-representation}

It is a useful exercise to repeat the analysis of the previous section in light of the tabular notation introduced in Section \ref{sec:sim:tabular-representation}.
Let us recall the specification of the attack on teleportation:
\begin{enumerate}

\item
Registers $(\mathsf{In},\mathsf{Out})$ are prepared in the pure state $C\ket{\phi^+}$ where $\ket{\phi^+}$ is shorthand for $n$ EPR pairs.

\item
An arbitrary attack unitary $U$ is applied to $(\sD,\mathsf{In},\mathsf{Out})$.

\item
A Bell measurement is applied to $(\sD,\mathsf{In})$ by via Bell rotation $B$ followed by a measurement in the computational basis, which is denoted $\set{\ket{T}\bra{T}}$.

\end{enumerate}
The entire procedure may be viewed as a channel
\begin{align}
  \Phi:\sD &\to (\sD,\mathsf{In},\mathsf{Out})\\
  \rho &\mapsto \sum_{\textrm{Paulis $T$}} \ket{T}\bra{T} B U C \Pa{ \rho\ot\ket{\phi^+}\bra{\phi^+} } C^* U^* B^* \ket{T}\bra{T} \enspace .
\end{align}
For each fixed Pauli $T$ consider the Kraus operator
\begin{align}
  \label{eq:teleportation-Kraus}
  \ket{T}\bra{T} B U C\ket{\phi^+}
\end{align}
belonging to the channel $\Phi$.
This Kraus operator is represented by the following table.
\begin{align} \label{tab:teleportation-Kraus}
  \renewcommand{\arraystretch}{1.2}
  \begin{array}{|c||c|c|c|c|c|} \hline
  \sD & \multirow{2}{*}{$\ket{T}\bra{T}$} & \multirow{2}{*}{$B$} & \multirow{3}{*}{$U$} & & \\
  \cline{1-1}\cline{5-6}
  \mathsf{In} & & & & & \multirow{2}{*}{$\ket{\phi^+}$} \\
  \cline{1-3}\cline{5-5}
  \mathsf{Out} & & & & C & \\ \hline
  \end{array}
\end{align}
Substituting the Pauli-decomposition \eqref{eq:U-Paulis} of $U$ into the Kraus operator \eqref{eq:teleportation-Kraus} yields
\begin{align}
  \eqref{eq:teleportation-Kraus} = \sum_{\textrm{Paulis $P$}} \alpha_P \ket{T}\bra{T} B P C\ket{\phi^+} \enspace .
\end{align}
Fix a choice of Pauli $P=P_\sD\ot P_\mathsf{In}\ot P_\mathsf{Out}$ and consider the operator $\ket{T}\bra{T} B P C\ket{\phi^+}$.
This operator can be written in tabular form
\begin{align} \label{tab:teleportation-Kraus-Paulis}
  \renewcommand{\arraystretch}{1.2}
  \begin{array}{|c||c|c|c|c|c|} \hline
  \sD & \multirow{2}{*}{$\ket{T}\bra{T}$} & \multirow{2}{*}{$B$} & P_\sD & & \\
  \cline{1-1}\cline{4-6}
  \mathsf{In} & & & P_\mathsf{In} & & \multirow{2}{*}{$\ket{\phi^+}$} \\
  \cline{1-5}
  \mathsf{Out} & & & P_\mathsf{Out} & C & \\ \hline
  \end{array}
\end{align}
so that
\begin{align}
  \boxed{\eqref{tab:teleportation-Kraus}} = \sum_{\textrm{Paulis $P$}} \alpha_P \boxed{\eqref{tab:teleportation-Kraus-Paulis}}
\end{align}
The observations of the previous sections can be phrased as an identity between tables:
\begin{align} \label{tab:teleportation}
  \boxed{\eqref{tab:teleportation-Kraus-Paulis}} =
  \renewcommand{\arraystretch}{1.2}
  \begin{array}{|c||c|c|c|c|c|c|} \hline
  \sD & & & & & & \multirow{2}{*}{$\frac{1}{2^{n/2}}\ket{T}$} \\
  \cline{1-6}
  \mathsf{In} & & & & & & \\
  \hline
  \mathsf{Out} & P_\mathsf{Out} & C & P_\mathsf{In}^\trans & T & P_\sD & \\ \hline
  \end{array}
\end{align}
Often it is more convenient to simply remove the rows $(\sD,\mathsf{In})$ and write
\begin{align} \label{tab:teleportation-two}
  \renewcommand{\arraystretch}{1.2}
  \begin{array}{|c||c|c|c|c|c|} \hline
  \mathsf{Out} & P_\mathsf{Out} & C & P_\mathsf{In}^\trans & T & P_\sD \\ \hline
  \end{array}
\end{align}
so that 
\begin{align}
  \boxed{\eqref{tab:teleportation-Kraus}} = \frac{1}{2^{n/2}}\ket{T} \ot \sum_{\textrm{Paulis $P$}} \alpha_P \boxed{\eqref{tab:teleportation-Kraus-Paulis}}
\end{align}

\section*{Acknowledgements}
\addcontentsline{toc}{section}{Acknowledgements}

We gratefully acknowledge helpful discussions with Daniel Gottesman, Harry Buhrman, Christian Schaffner, Bruce Richmond and Dominique Unruh. 
A.B.\ acknowledges support from the Canadian Institute for Advanced Research (CIFAR), Canada's NSERC and Industry Canada.
G.G.\ acknowledges support from  Industry Canada, Ontario's Ministry of Research and Innovation, NSERC, DTO-ARO, CIFAR, and QuantumWorks.
Part of this research conducted while D.S.\ was a visitor at the University of Waterloo's IQC.

 \small
\bibliographystyle{halphads}
\addcontentsline{toc}{section}{References}
\newcommand{\etalchar}[1]{$^{#1}$}

\end{document}